\documentclass[]{article}

\usepackage{amsmath,amssymb,amsthm,amsfonts}   
\numberwithin{equation}{section}   
\usepackage{physics}   
\usepackage[english]{babel}   
\usepackage[margin=3cm]{geometry}   
\usepackage{csquotes}   
\usepackage{mathtools}   
\usepackage{relsize}   
\usepackage{indentfirst}  
\usepackage{thm-restate}   
\usepackage{mathrsfs}   
\usepackage[dvipsnames]{xcolor}   
\usepackage{paracol}    
\usepackage{enumitem}    
\usepackage{soul}       
\usepackage{aligned-overset}    
\usepackage[thinc]{esdiff}  
\usepackage{upgreek}    
\usepackage{datetime}   
\usepackage{tensor}     
\usepackage{caption}    
\usepackage{subcaption} 
\usepackage{appendix}   
\allowdisplaybreaks

\usepackage[dvipsnames]{xcolor}     
\usepackage{hyperref}	            
\newcommand\myshade{85}         
\colorlet{mylinkcolor}{violet}  
\colorlet{mycitecolor}{YellowOrange}    
\colorlet{myurlcolor}{Aquamarine}       
\hypersetup{            
    linkcolor  = blue!\myshade!black,      
    citecolor  = mycitecolor!\myshade!black,   
    urlcolor   = myurlcolor!\myshade!black, 
    colorlinks = true,      
    hypertexnames=false,
}               
\usepackage{cleveref}           
\usepackage[                                
style = alphabetic,                         
maxbibnames=9,maxcitenames=9,               
doi=true,isbn=true,giveninits=true,         
backend=biber,                              
url=false,arxiv=true]{biblatex}             
\theoremstyle{plain}   
 
\newtheorem{biglemma}{Lemma}
\theoremstyle{plain}   
\newtheorem{assumption}{Assumption}     
\crefname{assumption}{assumption}{assumptions}

\newtheorem{lemma}{Lemma}[section]   
\crefname{lemma}{lemma}{lemmas}
\newtheorem{theorem}{Theorem}[section]  

\newtheorem{definition}{Definition}[section]    %
\newtheorem{prop}[lemma]{Proposition}   
\crefname{prop}{proposition}{propositions}
\newtheorem{cor}[lemma]{Corollary}   
\newtheorem{remark}[lemma]{Remark}   
\newtheorem{dfn}{Definition}    
\usepackage{pageslts}   
\usepackage{fancyhdr}   
\theoremstyle{plain}
\newtheorem{example}{Example}

\usepackage{xparse} 
\ExplSyntaxOn       
\NewDocumentCommand{\definealphabet}{mmmm}{
  \int_step_inline:nnn { `#3 } { `#4 }{
    \cs_new_protected:cpx { #1 \char_generate:nn { ##1 }{ 11 } }{
      \exp_not:N #2 { \char_generate:nn { ##1 } { 11 } }}}}              
\ExplSyntaxOff      
\definealphabet{mb}{\mathbb}{A}{Z}      
\definealphabet{mc}{\mathcal}{A}{Z}     
\definealphabet{scr}{\mathscr}{A}{Z}    
\definealphabet{mf}{\mathfrak}{A}{Z}    

\DeclareMathOperator{\states}{\mbS_D}   
\DeclareMathOperator{\tens}{\mbT_{d,D}}
\DeclareMathOperator{\mapspace}{\mcL(\mbM_D)}   
\DeclarePairedDelimiterX{\corr}[2]{\mathrm{corr}(}{)}{#1,#2}       
\DeclarePairedDelimiterX{\cov}[2]{\mathrm{cov}(}{)}{#1,#2}         
\DeclarePairedDelimiterX{\dis}[2]{\mathrm{d}(}{)}{#1, #2} 

\DeclareMathOperator{\pr}{Pr}           
\newcommand{\adj}{^\dagger}    
\DeclareMathOperator{\matrices}{\mbM_D} 
\DeclarePairedDelimiterX{\inner}[2]{\langle}{\rangle}{#1| #2} 
\DeclareMathOperator{\ten}{\mbC^d\otimes(\mbC^D)^{\otimes 2}}   

\newcommand{\Law}{\mathrm{Law}}

\renewcommand{\tr}[1]{\operatorname{tr}\left[#1\right]} 

\renewcommand{\vec}[1]{\text{vec}(#1)}

\newcommand{\avg}[1]{\mbE\left[#1\right]}           
\DeclarePairedDelimiterX{\cnum}[1]{\mathrm{c}(}{)}{#1}  
\newcommand{\floor}[1]{\left\lfloor #1 \right\rfloor}   
\newcommand{\ceil}[1]{\left\lceil #1 \right\rceil}      
\DeclarePairedDelimiterX{\diam}[1]{\text{diam}(}{)}{#1} 
\newcommand{\proj}{{\mathlarger{\boldsymbol{\cdot}}}} 
\let\oldker\ker         
\renewcommand{\ker}[1]{\oldker\left(#1\right)}  

\usepackage{tikz}       
\usepackage{tikzscale}  
\usetikzlibrary{positioning, shapes.misc, arrows.meta, fit, backgrounds, decorations.markings, shapes.geometric, decorations.pathreplacing}
\usepackage{floatrow}   
\usepackage{adjustbox}  
\usepackage[section]{placeins}  
\definecolor{newgreen}{RGB}{64,176,166} 
\definecolor{newyellow}{RGB}{225,190,106}   
\definecolor{newred}{RGB}{220,50,32}    
\tikzset{tensor/.style={
    rectangle, draw=blue!50, fill=gray!20, thick,
    minimum width=1cm, minimum height=1cm,
    inner sep=0pt, text width=1cm, align=center, font=\small, alias=thisone, 
    append after command={
        (thisone.west) -- ++(-0.5cm, 0)
        (thisone.east) -- ++(0.5cm, 0)
       (thisone.north) -- ++(0, 0.5cm)
    }
}}
\tikzset{tensorT/.style={
    rectangle, draw=blue!50, fill=gray!20, thick,
    minimum width=1cm, minimum height=1cm,
    inner sep=0pt, text width=1cm, align=center, font=\small, alias=thisone, 
    append after command={
        (thisone.west) -- ++(-0.5cm, 0)
        (thisone.east) -- ++(0.5cm, 0)
       (thisone.south) -- ++(0, -0.5cm)
    }
}}

\title{Correlation Lengths for Stochastically Generated Matrix Product States}
\usepackage{authblk}

\author[1]{Lubashan Pathirana\thanks{lpk@math.ku.dk}}
\author[2]{Albert H Werner \thanks{werner@math.ku.dk}}
\affil[1,2]{Department of Mathematical Sciences and QMATH, University of Copenhagen, Denmark.}
\date{}
\begin{document}
\pagenumbering{arabic}
\lhead{\thepage}
\maketitle
\vspace{-1cm}

\begin{abstract}
    We introduce a general model of stochastically generated matrix product states (MPS) in which the local tensors share a common distribution and form a strictly stationary sequence, without requiring spatial independence.
    Under natural conditions on the associated transfer operators, we prove the existence of thermodynamic limits of expectation values of local observables and establish almost-sure exponential decay of two-point correlations.
    In the homogeneous (random translation-invariant) case, for any error tolerance in probability, the two-point function decays exponentially in the distance between the two sites, with a deterministic rate.
    In the i.i.d.\ case, the exponential decay still holds with a deterministic rate, with the probability approaching one exponentially fast in the distance. 
    For strictly stationary ensembles with decaying spatial dependence, the correlation decay quantitatively reflects the mixing profile: $\rho$–mixing yields polynomial bounds with high probability, while stretched-exponential (resp. exponential) decay in $\rho$ (resp. $\beta$) yields stretched-exponential (resp. exponential) decay of the two-point function, again with correspondingly strong high-probability guarantees.  
    Altogether, the framework unifies and extends recent progress on stationary ergodic ensembles and Gaussian translation-invariant ensembles, providing a transfer-operator route to typical correlation decay in random MPS.
\end{abstract}




\section{Introduction}
\label{sec:intro}


\subsection{Matrix Product States}
\label{subsection:MPS}


    Matrix product states (MPS) play a central role in quantum many-body theory and quantum information. They provide efficient representations of not-too-strongly entangled quantum states \cite{fannes1992finitely, PerezGarcia2007} and form the basis of powerful numerical methods, such as the density matrix renormalization group (DMRG) \cite{White_1992,schollwock2011} and its descendants. MPS and tensor-network methods have become key tools for understanding quantum many-body systems, for instance, in the many-body localized phase where constrained entanglement growth preserves computational tractability \cite{abanin2019,orus2019}. MPS and their multidimensional generalization—tensor networks (TN)—also appear prominently in variational quantum algorithms (VQA) \cite{Haghshenas_2022,Rudolph_2023,Watanabe_2024}. A key structural feature of MPS is exponential clustering of correlations in the thermodynamic limit, a property that underlies their algorithmic usefulness \cite{fannes1992finitely, PerezGarcia2007}. For \emph{deterministic, translation–invariant} MPS, exponential clustering is well understood via the spectral theory of transfer operators \cite{orus2014,hastings2007}.
    
    Random MPS arise naturally in the study of noisy quantum circuits, statistical ensembles of tensor networks, and models of disordered quantum matter such as many-body localized phases \cite{Brandao2016, Nahum2017, Bauer2013}. The behavior of certain \emph{random} MPS (and more generally random tensor networks) has only recently begun to be systematically explored \cite{Garnerone_2010,Haferkamp_2021,abanin2019,Movassagh_2022,Lancien_2021,Cheng_2023, Chen_2024, L_io_2025,Hayden2016, Qi2018}. 
    Understanding their typical correlation properties is crucial for bridging quantum information theory with statistical mechanics in disordered settings. Moreover, randomized MPS have proved useful for analyzing the presence (and absence) of barren plateaus in quantum machine learning \cite{Garcia2023,Liu_2022}.
    
    Motivated by these developments, and building on recent advances in the mathematical theory of tensor networks \cite{cirac2021,bridgeman2017,orus2014,orus2019}, we introduce and analyze a general model of \emph{stochastically generated matrix product states}. 
    In this model, the local order-three tensors $(\mcA_n)_{n\in\mbZ}\subset \mbC^d\otimes(\mbC^D)^{\otimes 2}$ share a common marginal law $\nu$ but may be arbitrarily correlated (in a stochastic sense) across sites. 
    The term \emph{stochastic MPS} has been used in a different sense in the literature—namely for nonnegative MPS representations of \emph{classical} probability distributions (e.g.\ \cite{Temme_2010}); in contrast, our \emph{stochastically generated MPS} are \emph{quantum} MPS whose local tensors (equivalently, transfer operators) are produced by an underlying stochastic mechanism, and the randomness is in the generation of the state itself.
    This framework unifies random translation-invariant, i.i.d., and more general stochastically-correlated ensembles, extending prior results for ergodic sequences of MPS \cite{Movassagh_2022} and Gaussian homogeneous (translation-invariant) MPS \cite{Lancien_2021}, and complementing studies of Haar-random states \cite{Garnerone_2010}. Furthermore, our analysis of transfer operators directly interfaces with open-system dynamics realized by repeated interactions with a stochastic environment \cite{Attal2006,Pathirana_2023,Pathirana_2025,bruneau2008random,bruneau2010infinite,nechita2012random,bougron2022markovian}.


\subsection{Stochastically Generated MPS}
\label{subsection:STIMPS}


    \begin{figure}[ht]
         \begin{tikzpicture}[scale=1, thick]   
            \tikzstyle{line} = [thin, gray]
            \tikzstyle{vertical} = [thin, blue]
            \def\dy{1} 
            \node at (-7,0.25) {$\ket{\Psi_{N}}:=$};
            \node[tensor] (a1) at (-5,0) {\small{$\mcA_{N}$}};
            \node[tensor] (a2) at (-3,0) {\small{$\mcA_{N-1}$}};
            \node at (-1,0) {\(\ldots\)};
            \node[tensor] (a3) at (0,0) {\small{$\mcA_0$}};
            \node at (1,0) {\(\ldots\)};
            \node[tensor] (a4) at (3,0) {\small{$\mcA_{-N+1}$}};
            \node[tensor] (a5) at (5,0) {\small{$\mcA_{-N}$}};
            \foreach \i in {1,...,5} {
            \draw[vertical] (a\i) --++ (0,\dy);
            }
            \draw[line] (-6,0) --(a1);    
            \draw[line] (a1) -- (a2);
            \draw[line] (a2) -- (-1.5,0);
            \draw[line] (-0.5,0) -- (a3);
            \draw[line] (a3) -- (0.5,0);
            \draw[line] (1.5,0) -- (a4);
            \draw[line] (a4) -- (a5);
            \draw[line] (a5) -- (6,0);
            \draw[line] (-6,0) -- (-6,-1);
            \draw[line] (-6,-1) -- (6,-1);
            \draw[line] (6,-1) -- (6,0);
        \end{tikzpicture}
        \caption{Random MPS on $2N+1$ sites with periodic boundary conditions. Throughout, site labels increase from right to left in our tensor-network diagrams.}
    \label{fig:MPS_N}
    \end{figure}  

    In this subsection we introduce a unifying framework for random matrix product states (MPS) that subsumes two prominent regimes studied in the literature: (i) \emph{ergodic} sequences of local tensors as in \cite{movassagh2021theory}, and (ii) \emph{homogeneous} (translation–invariant) random MPS sampled from specific multivariate Gaussian laws as in \cite{Lancien_2021}. 
    The common structural feature in both settings is \emph{strict stationarity} of the sequence of local tensors. 
    We therefore place strict stationarity at the foundation of our model and derive consequences for correlation decay under mild dynamical hypotheses on the associated transfer superoperators.

    \paragraph{The Random Model}
        Fix integers $d,D\ge2$ and consider the space of rank–three local tensors with physical index $d$ and bond index $D$,
            \[
                \mbC^{d}\,\otimes\,(\mbC^{D})^{\otimes2}.
            \]
        A \emph{local tensor} is encoded as $\mcA=(A_1,\ldots,A_d)\in\ten$ with $A_p\in\mbM_D$.  
        Let $(\Omega,\mcF,\Pr)$ be a probability space and let $(\mcA_n)_{n\in\mbZ}$ be a $\ten$–valued stochastic process.

    \begin{definition}[Stochastically Generated Local Tensors]
        \label{def:strict-stationary-local}
        We say that $(\mcA_n)_{n\in\mbZ}$ is \emph{strictly stationary} (or simply stationary) if, for each $r\in\mbZ$, the joint law of $(\mcA_{n_1},\ldots,\mcA_{n_k})$ coincides with that of $(\mcA_{n_1+r},\ldots,\mcA_{n_k+r})$ for every \(k\in\mbN\), every \(n_1,\ldots,n_k\in\mbZ\), and every \(r\in\mbZ\).
        We call a strictly stationary sequence $(\mcA_n)_{n\in\mbZ}$ of $\ten$–valued random variables \emph{stochastically generated local tensors}.  
    \end{definition}

    \begin{remark}
        We use the term “stochastically generated” because strict stationarity precisely means that there exists a measure–preserving dynamical system $(\Omega,\mcF,\Pr,\theta)$ and a measurable sampling rule $\mcS:\Omega\to\ten$ such that $\mcA_n=\mcS\circ\theta^n$ for all $n\in\mbZ$; in other words, the local data are generated by a stationary stochastic process.
        Here $\theta:\Omega\to \Omega$ is $\mcF$-measurable, $\pr$-preserving and invertible with an $\mcF$-measurable inverse.
        We collect such canonical constructions—Markov-, Bernoulli-, conditionally Bernoulli-, and renewal-modulated instances—in Appendix~\ref{app:examples}.

    \end{remark}

    Given a realization $\omega\in\Omega$, placement of $\mcA_m(\omega),\ldots,\mcA_n(\omega)$ on the sites $m,\ldots,n$ and contracting the adjacent bond indices and the two open bond indices at the boundary yields the (generally unnormalized) random MPS vector on $[m,n]$,
    \begin{equation}
    \label{eq:MPS}
       \ket{\Psi_{[m,n]}^\omega}
        :=
        \sum_{p_m,\ldots,p_n=1}^{d}
        \tr{
            A_{p_n}^{(n),\omega}
            \cdots
            A_{p_m}^{(m),\omega}
        }
        \ket{p_m,\ldots,p_n}.
    \end{equation}
    While $\ket{\Psi_{[m,n]}(\omega)}$ need not be translation–invariant for fixed $\omega$, its \emph{law} is translation–invariant whenever $(\mcA_n)$ is strictly stationary. 
    For this reason, we sometimes refer to such ensembles as \emph{stochastically translation–invariant} (STI) MPS.
    It is not assumed a priori that $\ket{\Psi_{[m,n]}(\omega)}$ is normalized.
    Equivalently, we do not assume that $\bra{\Psi_{[m,n]}(\omega)}\ket{\Psi_{[m,n]}(\omega)}=1$.
     
    \begin{figure}
        \begin{subfigure}{0.47\textwidth}
            \resizebox{\textwidth}{!}{
            \begin{tikzpicture}
                \draw (-6,0) --  (-6, -1) ;
                \draw (-5, 0) node[tensor] {$\mathcal{A}_{N}$};
                \draw (-3, 0) node[tensor] {$\mathcal{A}_{N-1}$};
                \draw (-1.5,0) node[] {\ldots};
                \draw (0, 0) node[tensor] {$\mathcal{A}_{0}$};
                 \draw (1.5,0) node[] {\ldots};
                \draw (3, 0) node[tensor] {$\mathcal{A}_{-N+1}$};
                \draw (5, 0) node[tensor] {$\mathcal{A}_{-N}$};
                \draw (6,0) --  (6, -1) ;
                \draw(6,-1) -- (-6,-1);
                \draw[red] (-5,-0.5) -- (-5,-0.95);
                \draw[red] (-5,-1.05) -- (-5,-2);
                \node at (-5,-2) {\large\textcolor{red}{$\uparrow$}};
                \draw[red] (-3,-0.5) -- (-3,-0.95);
                \draw[red] (-3,-1.05) -- (-3,-2);
                \node at (-3,-2) {\large\textcolor{red}{$\uparrow$}};
                \draw[red] (-0,-0.5) -- (-0,-0.95);
                \draw[red] (-0,-1.05) -- (-0,-2);
                \node at (0,-2) {\large\textcolor{red}{$\uparrow$}};
                \draw[red] (3,-0.5) -- (3,-0.95);
                \draw[red] (3,-1.05) -- (3,-2);
                \node at (3,-2) {\large\textcolor{red}{$\uparrow$}};
                \draw[red] (5,-0.5) -- (5,-0.95);
                \draw[red] (5,-1.05) -- (5,-2);
                \node at (5,-2) {\large\textcolor{red}{$\uparrow$}};
                \draw [decorate,decoration={brace,amplitude=5pt,mirror,raise=4ex},red]
                (-5,-2) -- (5,-2) node[midway,yshift=-3em]{{\color{black}$\omega$}};
            \end{tikzpicture}
            }
            \caption{Random Translation-Invariant MPS.}
            \label{fig:TI_sampled}
        \end{subfigure}
        \begin{subfigure}{0.47\textwidth}
            \resizebox{\textwidth}{!}{
                \begin{tikzpicture}
                    \draw (-6,0) --  (-6, -1) ;
                    \draw (-5, 0) node[tensor] {$\mathcal{A}_{N}$};
                    \draw (-3, 0) node[tensor] {$\mathcal{A}_{N-1}$};
                    \draw (-1.5,0) node[] {\ldots};
                    \draw (0, 0) node[tensor] {$\mathcal{A}_{0}$};
                     \draw (1.5,0) node[] {\ldots};
                    \draw (3, 0) node[tensor] {$\mathcal{A}_{-N+1}$};
                    \draw (5, 0) node[tensor] {$\mathcal{A}_{-N}$};
                    \draw (6,0) --  (6, -1) ;
                    \draw(6,-1) -- (-6,-1);
                    \draw[red] (-5,-0.5) -- (-5,-0.95);
                    \draw[red] (-5,-1.05) -- (-5,-2);
                    \draw (-5,-2.5) node[] {$\omega_N$};
                    \node at (-5,-2) {\large\textcolor{red}{$\uparrow$}};
                    \draw[red] (-3,-0.5) -- (-3,-0.95);
                    \draw[red] (-3,-1.05) -- (-3,-2);
                    \draw (-3,-2.5) node[] {$\omega_{N-1}$};
                    \node at (-3,-2) {\large\textcolor{red}{$\uparrow$}};
                    \draw[red] (-0,-0.5) -- (-0,-0.95);
                    \draw[red] (-0,-1.05) -- (-0,-2);
                    \draw (-0,-2.5) node[] {$\omega_0$};
                    \node at (0,-2) {\large\textcolor{red}{$\uparrow$}};
                    \draw[red] (3,-0.5) -- (3,-0.95);
                    \draw[red] (3,-1.05) -- (3,-2);
                    \draw (3,-2.5) node[] {$\omega_{-N+1}$};
                    \node at (3,-2) {\large\textcolor{red}{$\uparrow$}};
                    \draw[red] (5,-0.5) -- (5,-0.95);
                    \draw[red] (5,-1.05) -- (5,-2);
                    \draw (5,-2.5) node[] {$\omega_{-N}$}; 
                    \node at (5,-2) {\large\textcolor{red}{$\uparrow$}};
                \end{tikzpicture}
            }
            \caption{Random IID sampled MPS.}
            \label{fig:IID_sampled}
        \end{subfigure}
        \caption{Some random MPS.}
    \end{figure}

    Strict stationarity encompasses, as extreme instances, the following two regimes:
    At one extreme, \emph{maximal} stochastic correlation is realized by placing the \emph{same} random tensor at every site, yielding a \emph{translation–invariant} (TI) (also called homogeneous) random MPS (\Cref{fig:TI_sampled}).
    At the other extreme, the tensors are independent and identically distributed (i.i.d.), producing the least stochastically correlated ensemble (\Cref{fig:IID_sampled}).
    Between these poles lie intermediate families in which spatial stochastic correlations decay with the distance separating sites (e.g.\ $\rho$–, $\beta$–, $\psi$–, or $\phi$–mixing families).
    Thus, the present framework simultaneously covers random translation–invariant and i.i.d.\ models, as well as ensembles with progressively weaker correlations as the separation increases.
    Furthermore, by taking the marginal to be a Dirac mass, this setting also allows one to study completely deterministic TI-MPS. 

    \paragraph{Two Simple Examples}
    
        \begin{example}
        \normalfont
         Fix $M$ rank–three tensors $\mcA^{(1)},\dots,\mcA^{(M)}$ and, at each site $n$, choose one of these $M$ tensors independently with probability $1/M$.
         \Cref{fig:unif_sampled} depicts the resulting ensemble; dotted arrows indicate the uniform choice, and distinct colors emphasize independence.
        \end{example}
    
        \begin{example}[Markov–modulated MPS]
        \label{ex:MMPS}
        \normalfont
            Let $X=(X_n)_{n\in\mbZ}$ be an irreducible, aperiodic, time–homogeneous Markov chain on a finite state space $S=\{1,\dots,m\}$ with stationary distribution $\pi$.  For each $i\in S$ fix a \emph{branch law} $\lambda_{B_i}\in\mcP(\ten)$, and let $\{B^{(i)}_n\}_{n\in\mbZ}$ be i.i.d.\ with one–site law $\lambda_{B_i}$, independently across $i$ and independently of $X$.  On the product probability space, carry the left shift $\theta$ and define the local tensor at site $n$ by
            \[
            \mcA_n(\omega) \;:=\; B^{(X_n(\omega))}_n(\omega)\in\ten.
            \]
            Then $(\mcA_n)_{n\in\mbZ}$ is strictly stationary (because each factor process is shift–invariant), with one–site marginal $\Law(\mcA_0)=\sum_{i=1}^m \pi_i\,\lambda_{B_i}$.  Conditionally on the whole path of $X$, the variables $(\mcA_n)$ are independent with $\Law(\mcA_n\mid X)=\lambda_{B_{X_n}}$; unconditionally they inherit the temporal dependence of $X$.
        \end{example}

     \begin{figure}[ht]
        \centering
        \scalebox{0.8}{
            \begin{tikzpicture}[node distance=1.8cm and 1.2cm, >=Stealth]
                \node[tensor] (A1) at (0,4) {$\mathcal{A}_{k_1}$};
                \node[tensor] (A2) at (2,4) {$\mathcal{A}_{k_2}$};
                \node[tensor] (A3) at (4,4) {$\mathcal{A}_{k_3}$};
                \node[tensor] (A4) at (6,4) {$\mathcal{A}_{k_4}$};
                \node at (7,4) {\(\cdots\)};
                \node[tensor] (A5) at (8,4) {$\mathcal{A}_{k_n}$};
                \draw[thick] (A1.east) -- (A2.west);
                \draw[thick] (A2.east) -- (A3.west);
                \draw[thick] (A3.east) -- (A4.west);
                \draw[thick] (A4.east) -- (6.75,4);  
                \draw[thick] (7.25,4) -- (A5.west);
                \foreach \x in {A1,A2,A3,A4,A5} {
                  \draw[thick] (\x.north) -- ++(0,0.8);
                }
                \draw[thick] (A1.west) -- (-1,4);
                \draw[thick] (-1,4)--(-1,3);
                \draw[thick] (A5.east) -- (9,4);
                \draw[thick] (9,4) --(9,3);
                \draw[thin] (9,3)--(-1,3);
                \node[tensor] (T1) at (0,0) {$\mathcal{A}_1$};
                \node[tensor] (T2) at (2,0) {$\mathcal{A}_2$};
                \node[tensor] (T3) at (4,0) {$\mathcal{A}_3$};
                \node at (6,0) {\(\cdots\)};
                \node[tensor] (TN) at (8,0) {$\mathcal{A}_M$};
                \begin{pgfonlayer}{background}
                  \node[draw=gray, fill=brown!10, thick, rounded corners, inner sep=0.5cm, fit=(T1)(T2)(T3)(TN)] {};
                \end{pgfonlayer}
                \node at (4,-1.4) {\textbf{Pool of tensors (picked uniformly at random at each site)}};
                \foreach \source in {T1,T2,T3,TN} {
                  \draw[dashed,->,thick,red] (\source) -- (A1);
                }
                \foreach \source in {T1,T2,T3,TN} {
                  \draw[dashed,->,thick,blue] (\source) -- (A2);
                }
                \foreach \source in {T1,T2,T3,TN} {
                  \draw[dashed,->,thick,green!60!black] (\source) -- (A3);
                }
                \foreach \source in {T1,T2,T3,TN} {
                  \draw[dashed,->,thick,orange] (\source) -- (A4);
                }
                \foreach \source in {T1,T2,T3,TN} {
                  \draw[dashed,->,thick] (\source) -- (A5);
                }
            \end{tikzpicture}
        }
        \caption{Random MPS generated by uniform sampling from a finite collection of tensors.}
        \label{fig:unif_sampled}
    \end{figure}

    \paragraph{Transfer Maps and Standing Assumptions (informal)}
        To each site \(n\), we associate the completely positive (not necessarily trace--preserving) transfer superoperator
        \[
            \phi_n^\omega:\mbM_D\to\mbM_D,
            \qquad
            \phi_n^\omega(X)=\sum_{p=1}^d A_p^{(n),\omega}\,X\,(A_p^{(n),\omega})\adj,
        \]
        and we write
        \[
            \Phi_{[a,b]}^\omega:=\phi_b^\omega\circ\cdots\circ\phi_a^\omega
        \]
        for forward compositions.
        The analysis proceeds under the following structural assumptions:
        \begin{description}
            \item[(A1)] Neither the transfer maps nor their adjoints annihilate any quantum state.
            \item[(A2)] Finite forward compositions of the transfer maps become positivity improving after some random but almost surely finite depth.
        \end{description}
        The assumptions used in \cite{movassagh2021theory} imply both conditions above.
        In the homogeneous Gaussian models of \cite{Lancien_2021}, the full set of hypotheses considered here is satisfied; see Appendix~\ref{section:examples} for a self-contained discussion.
        
        \paragraph{Preview of Results (informal)}
            Under the structural assumptions above, one first obtains projective boundary states and rank--one asymptotics for long transfer blocks.
            The periodic thermodynamic limit is obtained by combining the two exterior transfer blocks into a single positive map before applying the rank--one approximation.
            This avoids any separate lower bound on the overlap of the two boundary processes.
            Consequently, the same assumptions yield fast decay of connected two--point correlations, with high probability, in a range of stochastic regimes:
            \begin{itemize}
                \item \textbf{Random TI--MPS.} For any error tolerance \(\epsilon\in(0,1)\), the two--point function decays \emph{exponentially} in the separation with a \emph{deterministic} rate and prefactor, with probability at least \(1-\epsilon\).
                \item \textbf{I.i.d.\ local tensors.} The two--point function decays \emph{exponentially} with \emph{exponentially} high probabilities. 
                \item \textbf{Stochastically decorrelating ensembles.} If the maximal spatial \(\rho\)-mixing coefficients satisfy \(\rho_n\to0\), then for every \(k\in\mbN\) we obtain \emph{polynomial} decay \( |n-m|^{-k}\) with probability at least \(1-|n-m|^{-k}\).
                If \(\rho_n\) decays at least stretched--exponentially, then the two--point function decays \emph{stretched--exponentially}, with probability approaching \(1\) at a stretched--exponential rate in the separation.
                Under \emph{exponentially} decaying \(\beta\)-mixing, we obtain \emph{exponential} decay with probability approaching \(1\) at a \emph{subexponential} rate.
            \end{itemize}
        Precise statements appear in \Cref{sec:Main_Results}.


\subsection{Norm Conventions and Matrix Representations}
\label{subsection:Norm_Conventions}


    Throughout, let $\mbM_D$ denote the space of all $D\times D$ complex matrices. 
    We equip $\mbM_D$ with the Schatten $p$-norms for $p\in[1,\infty]$. 
    For $p\in[1,\infty)$ and $\rho\in\mbM_D$,
    \[
        \norm{\rho}_{p}
        \;=\;
        \left(\tr{|\rho|^{\,p}}\right)^{1/p}
        \qquad\text{with}\qquad
        |\rho|:=\sqrt{\rho\adj\rho},
    \]
    and the Schatten $\infty$-norm (spectral norm) is
    \[
        \norm{\rho}_{\infty} \;=\; \sup\{\|\rho u\|_{2} : u\in\mbC^{D},\ \|u\|_{2}=1\}.
    \]
    We use $\inner{\,\cdot\,}{\,\cdot\,}$ for the Hilbert–Schmidt inner product on $\mbM_D$,
    \begin{equation}
        \inner{A}{B} \;=\; \tr{A\adj B}
        \qquad
        (A,B\in\mbM_D).
    \end{equation}

    Let
    \[
      \states \;:=\; \{\rho\in\mbM_D:\ \rho\ge 0,\ \tr{\rho}=1\}
      \qquad\text{and}\qquad
      \states^{\mathrm{o}} \;:=\; \{\rho\in\states:\ \rho>0\},
    \]
    denote, respectively, the set of $D$–dimensional density matrices and the subset of strictly positive definite density matrices. 
    We write $\partial\states:=\states\setminus\states^{\mathrm{o}}$ for the rank–deficient states.
    For $\rho\in\states^{\mathrm{o}}$ let $\eta(\rho)$ denote its smallest eigenvalue.
    
    Let $\mapspace$ denote the space of linear maps $\phi:\mbM_D\to\mbM_D$. 
    We endow $\mapspace$ with the induced operator norm (the $p\!\to\!p$ norm):
    \[
        \norm{\phi}_{p\to p}
        \;=\;
        \sup\{\norm{\phi(\rho)}_{p} : \rho\in\mbM_D,\ \norm{\rho}_{p}\le 1\}\qquad(1\le p\le\infty).
    \]
   
    Every $\phi\in\mapspace$ admits a unique $D^2\times D^2$ matrix representation $K(\phi)$ via the standard vectorization map $\vec{\,\cdot\,}:\mbM_D\to\mbC^{D^2}$:
    \[
        \vec{\phi(\rho)} \;=\; K(\phi)\,\vec{\rho}\qquad(\rho\in\mbM_D).
    \]
    The assignment $\phi\mapsto K(\phi)$ is a linear isomorphism $\mapspace\to\mbM_{D^2}$, often called the \emph{natural (Liouville) representation}. 
    When the context is clear, we do not distinguish notationally between a superoperator and its Liouville matrix.

    We define the \emph{superoperator trace} by
    \[
        \Tr{\phi}
        \;:=\;
        \tr{K(\phi)}.
    \]
    Here $\tr{\,\cdot\,}$ on the right-hand side is the usual trace on $\mbM_{D^2}$.
    To avoid ambiguity, we reserve $\Tr{\,\cdot\,}$ for the trace of a superoperator and $\tr{\,\cdot\,}$ for the ordinary matrix trace, with the matrix dimension determined by context. 
    Equivalently, for any orthonormal basis $\{e_i\}_{i=1}^D$ of $\mbC^D$,
    \[
        \Tr{\phi}
        =
        \sum_{i,j=1}^{D}
        \tr{
            \ket{e_j}\bra{e_i}
            \,
            \phi\left(\ket{e_i}\bra{e_j}\right)
        }.
    \]
    In particular, if
    \[
      \phi(\rho) \;=\; \sum_{i,j=1}^{d} A_i\,\rho\,B_j^{T},
    \]
    where $(\,\cdot\,)^{T}$ denotes matrix transpose, then, with the convention $\vec{A\rho B^{T}}=(B\otimes A)\vec{\rho}$,
    \[
      K(\phi)\;=\;\sum_{i,j=1}^{d} B_j\otimes A_i,
      \qquad
      \Tr{\phi}\;=\;\sum_{i,j=1}^{d} \tr{A_i}\,\tr{B_j}.
    \]

    We shall also use the following elementary consequence of the Liouville representation.

    \begin{lemma}
    \label{lemma:superoperator_trace_cyclicity}
        Let \(\phi_1,\ldots,\phi_r\in\mapspace\).
        Then
        \[
            K\left(\phi_1\circ\cdots\circ\phi_r\right)
            =
            K(\phi_1)\cdots K(\phi_r).
        \]
        Consequently, for every \(j\in\{1,\ldots,r\}\),
        \[
            \Tr{\phi_1\circ\cdots\circ\phi_r}
            =
            \Tr{
                \phi_j\circ\phi_{j+1}\circ\cdots\circ\phi_r
                \circ\phi_1\circ\cdots\circ\phi_{j-1}
            }.
        \]
    \end{lemma}

    \begin{proof}
        We first prove the composition identity for two maps.
        Let \(\phi,\psi\in\mapspace\), and let \(\rho\in\mbM_D\).
        By the defining property of the Liouville representation,
        \[
        \begin{aligned}
            \vec{(\phi\circ\psi)(\rho)}
            &=
            \vec{\phi\left(\psi(\rho)\right)}
            \\
            &=
            K(\phi)\vec{\psi(\rho)}
            \\
            &=
            K(\phi)K(\psi)\vec{\rho}.
        \end{aligned}
        \]
        Since this identity holds for every \(\rho\in\mbM_D\), and since vectorization is a linear isomorphism, it follows that,
        \[
            K\left(\phi\circ\psi\right)
            =
            K(\phi)K(\psi).
        \]
        Iterating this identity gives
        \[
            K\left(\phi_1\circ\cdots\circ\phi_r\right)
            =
            K(\phi_1)\cdots K(\phi_r).
        \]
        Hence
        \[
            \Tr{\phi_1\circ\cdots\circ\phi_r}
            =
            \tr{K(\phi_1)\cdots K(\phi_r)}.
        \]
        By cyclicity of the ordinary trace on \(\mbM_{D^2}\),
        \[
        \begin{aligned}
            \tr{K(\phi_1)\cdots K(\phi_r)}
            &=
            \tr{
                K(\phi_j)\cdots K(\phi_r)
                K(\phi_1)\cdots K(\phi_{j-1})
            }
            \\
            &=
            \Tr{
                \phi_j\circ\phi_{j+1}\circ\cdots\circ\phi_r
                \circ\phi_1\circ\cdots\circ\phi_{j-1}
            }.
        \end{aligned}
        \]
        This proves the claim.
    \end{proof}
    
    Given $\phi\in\mapspace$, its adjoint $\phi\adj\in\mapspace$ (with respect to the Hilbert–Schmidt inner product) is the unique map satisfying
    \begin{equation}
        \inner{\phi(A)}{B}\;=\; \inner{A}{\phi\adj(B)}
        \qquad\text{for all }A,B\in\mbM_D,
    \end{equation}
    and the Liouville representation respects adjoints:
    \[
        K(\phi\adj)\;=\;K(\phi)\adj.
    \]
    For background on matrix norms and operator theory, see \cite{horn2012matrix,bhatia2013matrix}; for superoperators, see \cite{watrous2018theory}.


\subsection{Local Observables \& Transfer Operators}
\label{subsection:Observables_Transfer_op}


    \paragraph{Matrix Product State (MPS)}
        Let $(\mcA_n)_{n\in\mbZ}\subseteq\ten$ be the sequence of random tensors introduced above.
        For each $N\in\mbN$, we build (in general, unnormalized) random MPS on  $2N+1$ sites,
        \(
            [-N,N]:=\{-N,-N{+}1,\ldots,-1,0,1,\ldots,N{-}1,N\},
        \)
        by contracting the virtual (bond) indices of successive tensors and closing the remaining two open virtual legs (periodic boundary conditions).
        This yields a family $(\ket{\Psi_N^\omega})_{N\in\mbN}$; see \Cref{fig:MPS_N}.
        Strictly speaking, the physical state on $[-N,N]$ is the normalized pure state
        \[
            \frac{\ket{\Psi_N^\omega}\!\bra{\Psi_N^\omega}}{\inner{\Psi_N^\omega}{\Psi_N^\omega}},
        \]
        whenever \(\ket{\Psi_N^\omega}\neq0\). 
        But for brevity, we refer to the (unnormalized) vector $\ket{\Psi_N^\omega}$ as “the random MPS on $2N{+}1$ sites”.

    \medskip
    
    \paragraph{Transfer Operators}
        For each local tensor $\mcA_n=(A^{(n)}_1,\ldots,A^{(n)}_d)\in\ten$, the \emph{transfer operator} at site $n$ is the linear map
        \begin{equation}
        \label{eq:transfer_operator}
            \phi_n^\omega:\mbM_D\to\mbM_D,
            \qquad
            \phi_n^\omega(X)\;=\;\sum_{p=1}^d A^{(n),\omega}_p\,X\,(A^{(n),\omega}_p)\adj,
        \end{equation}
        which is completely positive but not necessarily trace-preserving.
        With the Liouville convention from \S\ref{subsection:Norm_Conventions}, the $D^2\times D^2$ matrix of $\phi_n^\omega$ is
        \begin{equation}
        \label{eq:K-phi}
            K(\phi_n^\omega)\;=\;\sum_{p=1}^d \overline{A^{(n),\omega}_p}\,\otimes\,A^{(n),\omega}_p.
        \end{equation}

    \medskip  
    
    \paragraph{Observable-dependent Transfer Operators}
        Let \(\{\ket{p}\}_{p=1}^d\) be the chosen physical basis, and write $(\mbO_n)_{p,q}:=\bra{p}\mbO_n\ket{q}$. 
        The \(\mbO_n\)-transfer superoperator is
        \begin{equation}
        \label{eq:On-superop}
            \mcO_n^\omega:\mbM_D\to\mbM_D,
            \qquad
            \mcO_n^\omega(X)
            :=
            \sum_{p,q=1}^d
            (\mbO_n)_{p,q}\,
            A_q^{(n),\omega}\,
            X\,
            \left(A_p^{(n),\omega}\right)\adj.
        \end{equation}
        Its Liouville matrix is
        \begin{equation}
        \label{eq:K-On}
            K\left(\mcO_n^\omega\right)
            =
            \sum_{p,q=1}^d
            (\mbO_n)_{p,q}\,
            \overline{A_p^{(n),\omega}}
            \otimes
            A_q^{(n),\omega}.
        \end{equation}
        In particular, setting \(\mbO_n=\mbI_d\) recovers the standard transfer operator:
        \[
            \mcO_n^\omega=\phi_n^\omega,
            \qquad
            K(\mcO_n^\omega)=K(\phi_n^\omega).
        \]
        For an observable \(\mbO_{[m,n]}\) supported on the block \(\{m,m+1,\ldots,n\}\), we write \(\mcO_{[m,n]}^\omega\) for the associated transfer superoperator obtained by inserting \(\mbO_{[m,n]}\) on the physical legs over that block and contracting as above.

    \medskip

    \paragraph{Expectation Values}
        Fix \(m\le n\) and \(N>\max\{|m|,|n|\}\).
        The normalized quantum expectation of a local observable \(\mbO_{[m,n]}\) in the state \(\ket{\Psi_N^\omega}\) is
        \[
            \frac{\bra{\Psi_N^\omega}\,\mbO_{[m,n]}\,\ket{\Psi_N^\omega}}
                 {\inner{\Psi_N^\omega}{\Psi_N^\omega}},
        \]
        whenever the denominator is not $0$.
        We write
        \[
            \Phi_N^{\mathrm R,\omega}
            :=
            \phi_N^\omega\circ\cdots\circ\phi_{n+1}^\omega,
            \qquad
            \Phi_N^{\mathrm L,\omega}
            :=
            \phi_{m-1}^\omega\circ\cdots\circ\phi_{-N}^\omega,
        \]
        and
        \[
            \Phi_{[m,n]}^\omega
            :=
            \phi_n^\omega\circ\cdots\circ\phi_m^\omega.
        \]
        Using standard tensor-network calculus together with \eqref{eq:K-phi}--\eqref{eq:K-On}, one obtains
        \begin{equation}
        \label{eq:block-exp-as-supertrace}
            \bra{\Psi_N^\omega}\,\mbO_{[m,n]}\,\ket{\Psi_N^\omega}
            =
            \Tr{
                \Phi_N^{\mathrm R,\omega}
                \circ
                \mcO_{[m,n]}^\omega
                \circ
                \Phi_N^{\mathrm L,\omega}
            }.
        \end{equation}
        Similarly,
        \begin{equation}
        \label{eq:block-norm-as-supertrace}
            \inner{\Psi_N^\omega}{\Psi_N^\omega}
            =
            \Tr{
                \Phi_N^{\mathrm R,\omega}
                \circ
                \Phi_{[m,n]}^\omega
                \circ
                \Phi_N^{\mathrm L,\omega}
            }.
        \end{equation}
        Hence
        \begin{equation}
        \label{eq:block-normalized-exp-as-supertrace}
            \frac{\bra{\Psi_N^\omega}\,\mbO_{[m,n]}\,\ket{\Psi_N^\omega}}
                 {\inner{\Psi_N^\omega}{\Psi_N^\omega}}
            =
            \frac{
                \Tr{
                    \Phi_N^{\mathrm R,\omega}
                    \circ
                    \mcO_{[m,n]}^\omega
                    \circ
                    \Phi_N^{\mathrm L,\omega}
                }
            }{
                \Tr{
                    \Phi_N^{\mathrm R,\omega}
                    \circ
                    \Phi_{[m,n]}^\omega
                    \circ
                    \Phi_N^{\mathrm L,\omega}
                }
            }.
        \end{equation}
        The one-site formula is obtained from \eqref{eq:block-exp-as-supertrace} by taking \(m=n\).
        Figures~\ref{fig:mcO}--\ref{fig:transfer} illustrate these constructions.
    
    \begin{figure}[ht]
        \captionsetup[subfigure]{}
        \centering
        \subcaptionbox{the $\mbO_n$-transfer operator, $\mcO_n$.}[.24\textwidth]
            {
                \begin{tikzpicture}    
                    \tikzstyle{vertical} = [thin, blue]
                        \node[tensor] (a) at (0,0) {$\mathcal{A}_{n}$};
                        \node[tensorT] (b) at (0,3) {$\overline{\mathcal{A}}_{n}$};
                        \draw[vertical] (a) -- (b);
                        \draw[fill=newyellow!40] (0,1.5) circle [radius=0.7];
                        \draw (0,1.5) node[] {$\mbO_n$};
                    
                        \begin{scope}[thin, black, decoration={
                            markings,
                            mark=at position 0.5 with {\arrow[red]{>}}
                        }] 
                        \draw[postaction={decorate}]  (1,0) -- (a);
                        \draw[postaction={decorate}]  (1,3) -- (b);
                        \draw[postaction={decorate}]  (a) -- (-1,0);
                        \draw[postaction={decorate}]  (b) -- (-1,3);
                        \end{scope}
                    \end{tikzpicture}
            }
        \subcaptionbox{$\mbO_{[m,n]}$-transfer operator, $\mcO_{[m,n]}$.}[.74\textwidth]
            {
                \begin{tikzpicture}
                    \tikzstyle{vertical} = [thin, blue]
                        \node[tensor] (a) at (-3,0) {$\mathcal{A}_{n}$};
                        \node[tensor] (b) at (-1,0) {$\mathcal{A}_{n-1}$};
                        \node at (0.5,0) {$\ldots$};
                        \node[tensor] (c) at (2,0) {$\mathcal{A}_{m+1}$};
                        \node[tensor] (d) at (4,0) {$\mathcal{A}_{m}$};
                        
                        \node[tensorT] (e) at (-3,3) {$\overline{\mathcal{A}}_{n}$};
                        \node[tensorT] (f) at (-1,3) {$\overline{\mathcal{A}}_{n-1}$};
                        \node at (0.5,3) {$\ldots$};
                        \node[tensorT] (g) at (2,3) {$\overline{\mathcal{A}}_{m+1}$};
                        \node[tensorT] (h) at (4,3) {$\overline{\mathcal{A}}_{m}$};
                    
                         \begin{scope}[thin, black, decoration={
                            markings,
                            mark=at position 0.5 with {\arrow[red]{<}}
                        }] 
                        \draw[postaction={decorate}]  (-4,0) -- (a);
                        \draw[postaction={decorate}]  (a) -- (b);
                        \draw[postaction={decorate}]  (b) -- (0,0);
                        \draw[postaction={decorate}]  (1,0) -- (c);
                        \draw[postaction={decorate}]  (c) -- (d);
                        \draw[postaction={decorate}]  (d) -- (5,0);
                    
                        \draw[postaction={decorate}]  (-4,3) -- (e);
                        \draw[postaction={decorate}]  (e) -- (f);
                        \draw[postaction={decorate}]  (f) -- (0,3);
                        \draw[postaction={decorate}]  (1,3) -- (g);
                        \draw[postaction={decorate}]  (g) -- (h);
                        \draw[postaction={decorate}]  (h) -- (5,3);
                        \end{scope}
                        
                        \draw[vertical] (a) -- (e);
                        \draw[vertical] (b) -- (f);
                        \draw[vertical] (c) -- (g);
                        \draw[vertical] (d) -- (h);
                     
                        \draw[fill=newyellow!40, rounded corners=0.2cm] (-3.5, 1) rectangle (4.5, 2);
                        \draw (0.5,1.5) node[] {$\mbO_{[m,n]}$};
                    \end{tikzpicture}
                \label{fig:transfer_m_n}
            }
        \caption{Transfer operators}
        \label{fig:mcO}
    \end{figure}
        
    \begin{figure}[ht]
        \centering
            \begin{tikzpicture}
            \draw (-5, 0) node[tensor] {$\mathcal{A}_{N}$};
            \draw (-3, 0) node[tensor] {$\mathcal{A}_{N-1}$};
            \draw (-1.5,0) node[] {\ldots};
            \draw (0, 0) node[tensor] {$\mathcal{A}_{n}$};
            \draw (1.5,0) node[] {\ldots};
            \draw (3, 0) node[tensor] {$\mathcal{A}_{-N+1}$};
            \draw (5, 0) node[tensor] {$\mathcal{A}_{-N}$};
            \draw (-5, 3) node[tensorT] {$\overline{\mathcal{A}}_{N}$};
            \draw (-3, 3) node[tensorT] {$\overline{\mathcal{A}}_{N-1}$};
            \draw (-1.5,3) node[] {\ldots};
            \draw (0, 3) node[tensorT] {$\overline{\mathcal{A}}_{n}$};
            \draw (1.5,3) node[] {\ldots};
            \draw (3, 3) node[tensorT] {$\overline{\mathcal{A}}_{-N+1}$};
            \draw (5, 3) node[tensorT] {$\overline{\mathcal{A}}_{-N}$};
            \draw (-6,0) -- (-6,-1);
            \draw (6,0) -- (6,-1);
            \draw (-6,-1) -- (6,-1);
            \draw (-6,3) -- (-6,4);
            \draw (6,3) -- (6,4);
            \draw (-6,4) -- (6,4);
            \draw (-5,1) -- (-5,2);
            \draw (-3,1) -- (-3,2);
            \draw[fill=newyellow!40] (0,1.5) circle [radius=0.7];
            \draw (0,1.5) node[] {$\mbO_n$};
            \draw (5,1) -- (5,2);
            \draw (3,1) -- (3,2);
        \end{tikzpicture}
            \caption{Expectation of observable $\mbO_n$ in (unnormalized) state $\ket{\Psi_N}$.}
        \label{fig:transfer}
        
    \end{figure} 
  

\section{Main Results}
\label{sec:Main_Results}

    
    As established above, normalized expectation values of local observables in random matrix product states (MPS) can be expressed in terms of the associated transfer superoperators.
    We now state our main results under the standing assumptions below.

    \paragraph{Stationary realization and notation.}
        Throughout this section we work on an invertible probability-preserving dynamical system \((\Omega,\mcF,\Pr,\theta)\), together with a measurable sampling map \(\mcS:\Omega\to\ten\), and we write
        \[
            \mcA_n(\omega):=\mcS(\theta^n\omega),
            \qquad n\in\mbZ.
        \]
        Thus \((\mcA_n)_{n\in\mbZ}\) is a strictly stationary sequence of local tensors with common marginal law on \(\ten\).
        This realization is available for every strictly stationary sequence, and we adopt it throughout.
        The associated transfer superoperators are denoted by \((\phi_n^\omega)_{n\in\mbZ}\).
        For \(\omega\in\Omega\) and \(n\in\mbN\), we write
        \begin{equation}
        \label{eq:Phi-n}
            \Phi^{(n)}_\omega
            :=
            \phi^\omega_{n-1}\circ\cdots\circ\phi^\omega_{1}\circ\phi^\omega_{0},
        \end{equation}
        with the convention \(\Phi^{(0)}_\omega:=\mathrm{Id}_{\mbM_D}\).

    A linear map \(\Psi:\mbM_D\to\mbM_D\) is called \emph{strictly positive} (or \emph{positivity improving}) if \(\Psi(X)> 0\) for every \(X\in\mbM_D\) with \(X\ge0\) and \(X\neq0\).
    We also write $\Psi \proj X = \frac{\Psi(X)}{\tr{\Psi(X)}}$ whenever the denominator is non-zero. 

    \begin{assumption}
    \label{assumption1}
        With probability one,
        \[
            \ker{\phi_0}\cap\states=\emptyset
            \qquad\text{and}\qquad
            \ker{\phi_0\adj}\cap\states=\emptyset.
        \]
    \end{assumption}

    Assumption~\ref{assumption1} ensures that neither \(\phi_0\) nor \(\phi_0\adj\) annihilates any state.
    By strict stationarity, the same kernel property then holds almost surely for every \(\phi_n\) and \(\phi_n\adj\).

    \begin{assumption}
    \label{assumption2}
        Almost surely there exists \(n_*(\omega)\in\mbN\) such that \(\Phi^{(n_*(\omega))}_\omega\) is strictly positive.
    \end{assumption}

    Assumption~\ref{assumption2} asserts that, with full probability, some finite composition of the local transfer maps is positivity improving.
    In particular (see \Cref{prop:ESP}), once strict positivity is attained it persists: for all \(k\ge n_*(\omega)\), the map \(\Phi^{(k)}_\omega\) is strictly positive.

    The next lemma records the boundary-state structure produced by \Cref{assumption1,assumption2}.
    Its proof is deferred to the appendix, where we establish a stronger quantitative version that will also be used in the proof of \Cref{thm:thermodynamic_limit}.
    See \Cref{lemma:boundary_states_quantitative} in \Cref{section:appen_1}.
    
    \begin{biglemma}
    \label{lemma:boundary_states_main}
        Let \((\mcA_n)_{n\in\mbZ}\) be the stationary realization above, and let \((\phi_n^\omega)_{n\in\mbZ}\) be the associated transfer maps.
        Assume \Cref{assumption1,assumption2}.
        Then there exists a \(\theta\)-invariant set \(\Omega_0\subseteq\Omega\) with \(\pr(\Omega_0)=1\) such that, for every \(\omega\in\Omega_0\), there exist two families of states
        \[
            \{Z_n(\omega)\}_{n\in\mbZ}\subset\states^{\mathrm o},
            \qquad
            \{Z_n'(\omega)\}_{n\in\mbZ}\subset\states^{\mathrm o},
        \]
        with the following properties.
        \begin{enumerate}
            \item[\emph{(A)}] \textbf{Projective limits.}
            For every fixed \(k\in\mbZ\),
            \begin{equation}
            \label{eq:boundary_limits_forward_main}
                \lim_{N\to\infty}
                \bigl(\phi_k^\omega\circ\cdots\circ\phi_{-N}^\omega\bigr)\proj\states
                =
                \{Z_k(\omega)\},
            \end{equation}
            and
            \begin{equation}
            \label{eq:boundary_limits_dual_main}
                \lim_{N\to\infty}
                \bigl(\phi_{k+N}^\omega\circ\cdots\circ\phi_k^\omega\bigr)\adj\proj\states
                =
                \{Z_k'(\omega)\},
            \end{equation}
            in trace norm, uniformly over the initial state.
        
            \item[\emph{(B)}] \textbf{Cocycle relations.}
            For every \(k\in\mbZ\),
            \begin{equation}
            \label{eq:boundary_cocycle_main}
                Z_k(\omega)
                =
                \frac{\phi_k^\omega\!\bigl(Z_{k-1}(\omega)\bigr)}
                     {\tr{\phi_k^\omega(Z_{k-1}(\omega))}},
                \qquad
                Z_k'(\omega)
                =
                \frac{(\phi_k^\omega)\adj\!\bigl(Z_{k+1}'(\omega)\bigr)}
                     {\tr{(\phi_k^\omega)\adj(Z_{k+1}'(\omega))}}.
            \end{equation}
        
            \item[\emph{(C)}] \textbf{Asymptotic rank--one approximation.}
            For every \(m\le n\), writing
            \[
                \phi^\omega_{[m,n]}
                :=
                \phi_n^\omega\circ\cdots\circ\phi_m^\omega,
            \]
            there exists a measurable nonnegative random variable \(\varepsilon_{m,n}:\Omega\to[0,\infty)\) such that
            \begin{equation}
            \label{eq:rank_one_main_qualitative}
                \norm{
                    \frac{\phi^\omega_{[m,n]}}
                         {\tr{(\phi^\omega_{[m,n]})\adj(\mbI_D)}}
                    -
                    \Xi^\omega_{[m,n]}
                }_{1\to1}
                \le
                \varepsilon_{m,n}(\omega),
            \end{equation}
            where
            \[
                \Xi^\omega_{[m,n]}(X)
                :=
                \tr{Z_m'(\omega)\,X}\,Z_n(\omega),
                \qquad X\in\mbM_D.
            \]
            Moreover, for \(\omega\in\Omega_0\),
            \[
                \varepsilon_{m,n}(\omega)\to0
                \quad\text{as }n\to+\infty\text{ with }m\text{ fixed},
            \]
            and
            \[
                \varepsilon_{m,n}(\omega)\to0
                \quad\text{as }m\to-\infty\text{ with }n\text{ fixed}.
            \]
            \end{enumerate}
    \end{biglemma}

    Lemma~\ref{lemma:boundary_states_main} yields the boundary states that govern the asymptotic left and right transfer dynamics.
    The following theorem shows that these boundary states determine the periodic thermodynamic limit under the same two standing assumptions.
    The proof uses cyclicity of the superoperator trace to combine the two exterior transfer blocks before applying a deterministic rank--one approximation.
    Thus no lower bound on the overlap \(\tr{Z'_{-N}(\omega)Z_N(\omega)}\) is required.
    
    With these conventions in place, we obtain the following thermodynamic-limit theorem.
    
    \begin{restatable}[]{thm}{thermodynamic}
    \label{thm:thermodynamic_limit}
        Let \((\mcA_n)_{n\in\mbZ}\) be a strictly stationary sequence with common marginal law on \(\ten\), and let \((\phi_n^\omega)_{n\in\mbZ}\) be the associated transfer superoperators.
        Assume \Cref{assumption1,assumption2}.
        Let \(\{Z_k(\omega)\}_{k\in\mbZ}\) and \(\{Z_k'(\omega)\}_{k\in\mbZ}\) be the boundary families furnished by \Cref{lemma:boundary_states_main}.
        Then, with probability \(1\), for every \(m,n\in\mbZ\) with \(m\le n\),
        \begin{equation}
        \label{eq:TDL-functional}
            \lim_{N\to\infty}
            \frac{\bra{\Psi_N^\omega}\,\mbO_{[m,n]}\,\ket{\Psi_N^\omega}}
                 {\inner{\Psi_N^\omega}{\Psi_N^\omega}}
            =
            \frac{
                \tr{
                    Z'_{n+1}(\omega)\,
                    \mcO_{[m,n]}^\omega\left(Z_{m-1}(\omega)\right)
                }
            }{
                \tr{
                    Z'_{n+1}(\omega)\,
                    \Phi_{[m,n]}^\omega\left(Z_{m-1}(\omega)\right)
                }
            }
            =:
            \mcT_\omega(\mbO_{[m,n]}),
        \end{equation}
        where \(\mbO_{[m,n]}\) is any Hermitian operator on \((\mbC^d)^{\otimes(n-m+1)}\) supported on the sites \([m,n]\), where \(\mcO_{[m,n]}^\omega\) is the observable-dependent transfer superoperator associated with \(\mbO_{[m,n]}\), and where \(\Phi_{[m,n]}^\omega:=\phi_n^\omega\circ\cdots\circ\phi_m^\omega\).
    \end{restatable}

    The denominator in \eqref{eq:TDL-functional} is strictly positive almost surely.
    Indeed, \(Z'_{n+1}(\omega)\in\states^{\mathrm o}\) by \Cref{lemma:boundary_states_main}.
    Moreover, \(\Phi_{[m,n]}^\omega\left(Z_{m-1}(\omega)\right)\ge0\), and repeated use of \Cref{assumption1} shows that this positive matrix is nonzero.
    Hence
    \[
        \tr{
            Z'_{n+1}(\omega)\,
            \Phi_{[m,n]}^\omega\left(Z_{m-1}(\omega)\right)
        }
        >
        0.
    \]
    We also point out that in the theorem statement above, the quantity \(\inner{\Psi_N^\omega}{\Psi_N^\omega}\) is strictly positive almost surely for sufficiently large $N$. 

    The thermodynamic limit in \Cref{thm:thermodynamic_limit} furnishes a rigorous infinite–volume expectation functional.
    In particular, it allows us to define \emph{two–point correlations} (connected correlations) between local observables in the thermodynamic limit.

    \begin{dfn}
    \label{dfn:correlation_function}
        For Hermitian local observables $\mbO_n$ and $\mbO_m$ supported at sites $n,m\in\mbZ$, respectively, the (connected) two–point correlation function is
        \begin{align}
              f^\omega(n,m)
              &:= \lim_{N\to\infty}
                  \left|
                    \frac{\bra{\Psi_N^\omega}\,\mbO_n\mbO_m\,\ket{\Psi_N^\omega}}
                         {\inner{\Psi_N^\omega}{\Psi_N^\omega}}
                    -
                    \frac{\bra{\Psi_N^\omega}\,\mbO_n\,\ket{\Psi_N^\omega}\;
                          \bra{\Psi_N^\omega}\,\mbO_m\,\ket{\Psi_N^\omega}}
                         {\inner{\Psi_N^\omega}{\Psi_N^\omega}^{2}}
                  \right| \label{eq:f-omega-def} \\
              &= \big|\,\mcT_\omega(\mbO_n\mbO_m)\;-\;\mcT_\omega(\mbO_n)\,\mcT_\omega(\mbO_m)\,\big|. \nonumber
        \end{align}
        Here $\mbO_n\mbO_m$ is understood, e.g.\ for $m<n$, as
        $\mbO_m\otimes\mbI_{[m+1,n-1]}\otimes\mbO_n$, where
        $\mbI_{[a,b]}:=\bigotimes_{i=a}^b \mbI_d$.
    \end{dfn}

    \medskip

    We next quantify the almost–sure decay of the two–point function, with a disorder–dependent rate and prefactor.

    \begin{restatable}[]{thm}{decayrandom}
    \label{thm:decay_random}
        Let \((\mcA_n)_{n\in\mbZ}\) be a strictly stationary sequence of local tensors with common marginal law on \(\ten\), and let \((\phi_n^\omega)_{n\in\mbZ}\) be the associated transfer operators.
        Assume \Cref{assumption1,assumption2}.
        Then there exists an almost surely positive random variable \(\alpha(\omega)>0\) such that, for every lattice point \(x\in\mbZ\), there exists an almost surely finite random prefactor \(g_x(\omega)\in(0,\infty)\) with the following property: for \(\pr\)-almost every \(\omega\), for all sites
        \[
            m<x<n
            \qquad\text{with}\qquad
            n-m\ge2,
        \]
        and for all local observables \(\mbO_n,\mbO_m\),
        \begin{equation}
        \label{eq:exp-decay}
            f^\omega(n,m)
            \le
            \norm{\mbO_n}_\infty\,\norm{\mbO_m}_\infty\,g_x(\omega)\,e^{-\alpha(\omega)(n-m)}.
        \end{equation}
        In particular, the two-point function decays exponentially in the separation \(|n-m|\), with a disorder-dependent rate and site- and disorder-dependent prefactor.
        
        Moreover:
        \begin{enumerate}
            \item
            If \((\mcA_n)_{n\in\mbZ}\) is i.i.d., or more generally stationary ergodic, then \(\alpha\) may be chosen deterministic.
        
            \item
            If \((\mcA_n)_{n\in\mbZ}\) is random homogeneous, then the prefactor may be chosen independent of \(x\): there exists an almost surely finite \(g:\Omega\to(0,\infty)\) such that, for \(\pr\)-almost every \(\omega\),
            \[
                f^\omega(n,m)
                \le
                g(\omega)\,\norm{\mbO_n}_\infty\,\norm{\mbO_m}_\infty\,e^{-\alpha(\omega)|n-m|}
                \qquad\text{for all }|n-m|\ge2.
            \]
        \end{enumerate}
    \end{restatable}


\subsection{Obtaining Uniform Bounds with High Probability} 
\label{subsection:uniform_bounds}


    With \Cref{thm:decay_random} in hand, we now aim to obtain $\omega$–uniform prefactors and exponential rates with probabilities arbitrarily close to $1$.
    We begin with the two extreme cases—translation–invariant and i.i.d.—and then treat decorrelating (mixing) ensembles.

    \paragraph{Random Translation–invariant MPS}
        For random TI–MPS, we strengthen \Cref{thm:decay_random}: given any error tolerance $\epsilon\in(0,1)$, there exist \emph{deterministic} constants $K(\epsilon)>0$ and $\lambda(\epsilon)>0$ such that the two–point function decays exponentially in the separation, uniformly in $\omega$, with probability at least $1-\epsilon$.

        \begin{restatable}[]{thm}{TICase}
        \label{thm:TIcase}
            Let \((\mcA_n)_{n\in\mbZ}\) be random homogeneous and assume  \Cref{assumption1,assumption2}.
            For each \(\epsilon\in(0,1)\) there exist \(K(\epsilon)>0\) and \(\lambda(\epsilon)>0\) such that, for any local observables \(\mbO_n,\mbO_m\) at sites \(n,m\in\mbZ\) with \(2\le |n-m|\),
            \begin{equation}
            \label{ineq:ticase}
                \pr\left\{
                    f^\omega(n,m)
                    \le
                    K(\epsilon)\,\norm{\mbO_n}_\infty\,\norm{\mbO_m}_\infty\,e^{-\lambda(\epsilon)|n-m|}
                \right\}
                \ge
                1-\epsilon.
            \end{equation}
        \end{restatable}

    \paragraph{Independent and Identically Distributed Case}
        For an i.i.d.\ sequence of local tensors, we obtain exponential decay with \emph{deterministic} rate and prefactor, and with probabilities that approach $1$ exponentially fast in the separation.
        
        \begin{restatable}[]{thm}{IID}
        \label{thm:IID}
            Let \((\mcA_n)_{n\in\mbZ}\) be i.i.d.\ and satisfy \Cref{assumption1,assumption2}.
            Then there exist constants \(C_{\mathrm{pr}}>0\) and \(\beta>0\) such that for every pair of local observables \(\mbO_n,\mbO_m\) at sites \(n,m\in\mbZ\) with \(2\le |n-m|\),
            \begin{equation}
            \label{eq:IID_decay}
                \pr\!\left\{
                    \omega:
                    f^\omega(n,m)
                    \le
                    C_{\mathrm{pr}}\,\norm{\mbO_n}_\infty\,\norm{\mbO_m}_\infty\,e^{-\beta|n-m|}
                \right\}
                \ge
                1-e^{-\beta|n-m|}.
            \end{equation}
        \end{restatable}
    
    \paragraph{Cases with Decaying Stochastic Correlations}
        In many physically relevant models, the local tensors are dependent, but their stochastic dependence weakens with spatial separation.
        We quantify this via standard mixing coefficients (see \cite{bradley2005basic}); in particular we use the maximal spatial $\rho$–mixing coefficients $(\rho_n)_{n\ge1}$ (defined in \Cref{section:stochastic_corr}).

        The following theorem shows that if $\rho_n\to0$, then the two–point function decays with arbitrarily high polynomial speed with probability polynomially close to one.
        
        \begin{restatable}[]{thm}{DECAY}
        \label{thm:decay_cor}
            Let \((\mcA_n)_{n\in\mbZ}\) be a strictly stationary sequence of local tensors whose associated transfer maps satisfy \Cref{assumption1,assumption2}.
            If the associated \(\rho\)-mixing coefficients \((\rho_n)\) obey \(\rho_n\to0\) as \(n\to\infty\), then for each \(k\in\mbN\) there exists \(C_{\mathrm{poly}}(k)>0\) such that, for all local observables \(\mbO_n,\mbO_m\) with \(2\le |n-m|\),
            \begin{equation}
            \label{eq:rho_poly_decay}
                \pr\left\{
                    \omega:
                    f^\omega(n,m)
                    \le
                    C_{\mathrm{poly}}(k)\,\norm{\mbO_n}_\infty\,\norm{\mbO_m}_\infty\,|n-m|^{-k}
                \right\}
                \ge
                1-|n-m|^{-k}.
            \end{equation}
        \end{restatable}

    If, moreover, $\rho_n$ decays at least stretched–exponentially (which is the case for Markov modulated MPS modulated via a strictly stationary, finite–state, irreducible, aperiodic Markov chain, see \cref{app:markov_modulated}), the bound improves to stretched--exponential decay.

    \begin{restatable}[]{thm}{fastdecay}
    \label{thm:fast_decay}
        Let \((\mcA_n)_{n\in\mbZ}\) be a strictly stationary sequence of local tensors whose associated transfer maps satisfy \Cref{assumption1,assumption2}.
        Assume there exist \(C_\rho,c_\rho>0\) and \(\gamma\in(0,1]\) such that
        \[
            \rho_n \le C_\rho\,e^{-c_\rho n^\gamma}
            \qquad\text{for all }n\in\mbN.
        \]
        Then for any \(\alpha\in(0,\gamma)\) there exist constants \(K_\alpha,k_\alpha>0\) such that, for all local observables \(\mbO_n,\mbO_m\) with \(2\le |n-m|\),
        \begin{equation}
        \label{eq:rho_stretched_decay}
            \pr\!\left\{
                \omega:
                f^\omega(n,m)
                \le
                K_\alpha\,\norm{\mbO_n}_\infty\,\norm{\mbO_m}_\infty\,e^{-k_\alpha|n-m|^\alpha}
            \right\}
            \ge
            1-e^{-k_\alpha|n-m|^\alpha}.
        \end{equation}
    \end{restatable}
    
    Other mixing notions such as $\psi$, $\varphi$, $\alpha$, and $\beta$–mixing fit into the same scheme; by the known hierarchy among these coefficients (see \S\ref{subsec:mixing_coeffs}), \Cref{thm:decay_cor} holds with $\rho$–mixing replaced by $\psi_n\to0$ or $\varphi_n\to0$, and \Cref{thm:fast_decay} holds under stretched–exponential decay of $\psi_n$ or $\varphi_n$.
    Our final result, presented below, for the cases with decaying stochastic correlations, utilizes $\beta$–mixing with an exponential rate, complementing the picture for all five classical mixing conditions.

    \begin{restatable}[]{thm}{fastdecaybeta}
    \label{thm:fast_beta_decay}
        Let \((\mcA_n)_{n\in\mbZ}\) be a strictly stationary sequence of local tensors whose associated transfer maps satisfy \Cref{assumption1,assumption2}.
        Let \((\beta_n)\) be the associated \(\beta\)-mixing coefficients (see \Cref{section:stochastic_corr}).
        Assume there exist \(C_\beta,c_\beta>0\) such that
        \[
            \beta_n \le C_\beta\,e^{-c_\beta n}
            \qquad\text{for all }n\in\mbN.
        \]
        Then there exist constants \(K_\beta,\kappa_\beta,p_\beta>0\) such that, for all Hermitian local observables \(\mbO_n,\mbO_m\) supported at sites \(n,m\in\mbZ\) with \(|n-m|\ge2\),
        \begin{multline}
        \label{eq:fast_beta_decay_final}
            \pr\!\left\{
                \omega:
                f^\omega(n,m)
                \le
                K_\beta\,\norm{\mbO_n}_\infty\,\norm{\mbO_m}_\infty\,e^{-\kappa_\beta|n-m|}
            \right\}
            \\
            \ge
            1-\exp\!\left\{
                -\,p_\beta\,\frac{|n-m|}{(\ln|n-m|)(\ln\ln|n-m|)}
            \right\}.
        \end{multline}
    \end{restatable}


\section{Proofs of \texorpdfstring{\Cref{thm:thermodynamic_limit}}{Theorem A}, \texorpdfstring{\Cref{thm:decay_random}}{Theorem B} and \texorpdfstring{\Cref{thm:TIcase}}{Theorem C}}
\label{sect:section2}


In this section we prove \Cref{thm:thermodynamic_limit}, \Cref{thm:decay_random}, and \Cref{thm:TIcase}.
We begin with a simple persistence observation for strict positivity.

    \begin{prop}
    \label{prop:ESP}
        Under \Cref{assumption1,assumption2}, the forward compositions
        \[
            \Phi^{(n)}_\omega
            :=
            \phi_{n-1}^\omega\circ\cdots\circ\phi_0^\omega
        \]
        are eventually strictly positive almost surely.
        In particular, there exists an almost surely finite random time \(n_*(\omega)\) such that
        \(\Phi^{(n)}_\omega\) is positivity improving for all \(n\ge n_*(\omega)\).
    \end{prop}
    \begin{proof}
        Let
        \[
            E_{\mathrm{ker}}
            :=
            \Bigl\{
                \omega\in\Omega:
                \ker{\phi_0^\omega}\cap\states=\emptyset
                \text{ and }
                \ker{(\phi_0^\omega)\adj}\cap\states=\emptyset
            \Bigr\}.
        \]
        By \Cref{assumption1}, \(\pr(E_{\mathrm{ker}})=1\).
        Since \(\phi_n^\omega=\phi_0^{\theta^n\omega}\) for every \(n\in\mbZ\), the set
        \[
            \widetilde E_{\mathrm{ker}}
            :=
            \bigcap_{j\in\mbZ}\theta^{-j}(E_{\mathrm{ker}})
        \]
        also has probability one.
        On \(\widetilde E_{\mathrm{ker}}\), none of the maps \(\phi_n^\omega\) nor \((\phi_n^\omega)\adj\) annihilates a state, for any \(n\in\mbZ\).
        
        By \Cref{assumption2}, there is a full-probability event \(E_*\subseteq\Omega\) such that for every \(\omega\in E_*\) there exists \(n_*(\omega)\in\mbN\) with \(\Phi^{(n_*(\omega))}_\omega\) strictly positive.
        Fix \(\omega\in E_*\cap \widetilde E_{\mathrm{ker}}\), and write \(n_*:=n_*(\omega)\).
        Let \(k\ge0\), and set
        \[
            \Psi_k
            :=
            \phi_{n_*+k-1}^\omega\circ\cdots\circ\phi_{n_*}^\omega,
            \qquad
            \Psi_0 := \mathrm{Id}_{\mbM_d},\qquad
            \qquad \text{so that}\qquad
            \Phi^{(n_*+k)}_\omega
            =
            \Psi_k\circ \Phi^{(n_*)}_\omega.
        \]
        We claim that \(\Phi^{(n_*+k)}_\omega\) is strictly positive for every \(k\ge0\).
        
        Let \(0\neq X\ge0\), and let \(\vartheta\in\states\) be arbitrary.
        Using the Hilbert--Schmidt pairing,
        \[
            \inner{\vartheta}{\Phi^{(n_*+k)}_\omega(X)}
            =
            \inner{(\Psi_k)\adj(\vartheta)}{\Phi^{(n_*)}_\omega(X)}.
        \]
        Since \(\Phi^{(n_*)}_\omega\) is positivity improving, the matrix \(\Phi^{(n_*)}_\omega(X)\) is strictly positive definite.
        In particular, it has strictly positive pairing with every nonzero positive semidefinite matrix.
        
        On the other hand, \((\Psi_k)\adj(\vartheta)\ge0\), and it is nonzero.
        Indeed, if \((\Psi_k)\adj(\vartheta)=0\), then after normalizing at the first step where needed, one would obtain a state annihilated by one of the adjoint transfer maps \((\phi_j^\omega)\adj\), contradicting \(\omega\in \widetilde E_{\mathrm{ker}}\).
        Hence
        \[
            \inner{(\Psi_k)\adj(\vartheta)}{\Phi^{(n_*)}_\omega(X)}>0.
        \]
        Since this holds for every \(\vartheta\in\states\), the matrix \(\Phi^{(n_*+k)}_\omega(X)\) must be strictly positive definite.
        As \(0\neq X\ge0\) was arbitrary, \(\Phi^{(n_*+k)}_\omega\) is positivity improving.
        Thus, on \(E_*\cap \widetilde E_{\mathrm{ker}}\), the map \(\Phi^{(n)}_\omega\) is strictly positive for all \(n\ge n_*(\omega)\).
        This proves the claim.
    \end{proof}


\subsection{Proof of \texorpdfstring{\Cref{thm:thermodynamic_limit}}{Theorem A}}
\label{subsection:proof_of_A}


We now prove \Cref{thm:thermodynamic_limit}.
The proof uses \Cref{lemma:superoperator_trace_cyclicity}, \Cref{lemma:deterministic_rank_one_reference}, and the quantitative refinement \Cref{lemma:boundary_states_quantitative}.
The point is to combine the two exterior blocks before applying any rank--one approximation.

\thermodynamic*

\begin{proof}
    Fix \(m\le n\).
    Let \(\Omega_0\subseteq\Omega\) be the full-measure \(\theta\)-invariant event furnished by \Cref{lemma:boundary_states_quantitative}.
    Fix \(\omega\in\Omega_0\), and suppress \(\omega\) from the notation whenever no confusion can arise.

    For \(N>\max\{|m|,|n|\}\), write
    \[
        \Phi_N^{\mathrm R}
        :=
        \phi_N\circ\cdots\circ\phi_{n+1},
        \qquad
        \Phi_N^{\mathrm L}
        :=
        \phi_{m-1}\circ\cdots\circ\phi_{-N},
    \]
    and
    \[
        \Phi_{[m,n]}
        :=
        \phi_n\circ\cdots\circ\phi_m.
    \]
    By the finite-volume transfer-operator identity,  whenever $\inner{\Psi_N^\omega}{\Psi_N^\omega}>0$, we have
    \begin{equation}
    \label{eq:ratio-TDL-new}
        \frac{\bra{\Psi_N^\omega}\,\mbO_{[m,n]}\,\ket{\Psi_N^\omega}}
             {\inner{\Psi_N^\omega}{\Psi_N^\omega}}
        =
        \frac{
            \Tr{
                \Phi_N^{\mathrm R}
                \circ
                \mcO_{[m,n]}
                \circ
                \Phi_N^{\mathrm L}
            }
        }{
            \Tr{
                \Phi_N^{\mathrm R}
                \circ
                \Phi_{[m,n]}
                \circ
                \Phi_N^{\mathrm L}
            }
        }.
    \end{equation}

    By cyclicity of the superoperator trace, \Cref{lemma:superoperator_trace_cyclicity}, we have
    \[
        \Tr{
            \Phi_N^{\mathrm R}
            \circ
            \mcO_{[m,n]}
            \circ
            \Phi_N^{\mathrm L}
        }
        =
        \Tr{
            \mcO_{[m,n]}
            \circ
            \Phi_N^{\mathrm L}
            \circ
            \Phi_N^{\mathrm R}
        },
    \]
    and
    \[
        \Tr{
            \Phi_N^{\mathrm R}
            \circ
            \Phi_{[m,n]}
            \circ
            \Phi_N^{\mathrm L}
        }
        =
        \Tr{
            \Phi_{[m,n]}
            \circ
            \Phi_N^{\mathrm L}
            \circ
            \Phi_N^{\mathrm R}
        }.
    \]
    Define the combined exterior map
    \[
        \Theta_N
        :=
        \Phi_N^{\mathrm L}
        \circ
        \Phi_N^{\mathrm R}.
    \]
    Then \eqref{eq:ratio-TDL-new} becomes
    \begin{equation}
    \label{eq:ratio-TDL-combined}
        \frac{\bra{\Psi_N^\omega}\,\mbO_{[m,n]}\,\ket{\Psi_N^\omega}}
             {\inner{\Psi_N^\omega}{\Psi_N^\omega}}
        =
        \frac{
            \Tr{
                \mcO_{[m,n]}
                \circ
                \Theta_N
            }
        }{
            \Tr{
                \Phi_{[m,n]}
                \circ
                \Theta_N
            }
        }.
    \end{equation}

    We shall now prove that a normalized version of \(\Theta_N\) converges to a rank--one map.
    Since \(\omega\in\Omega_0\), every finite transfer block appearing below maps states to nonzero positive matrices.
    Hence
    \[
        \ker{\Theta_N}\cap\states=\emptyset.
    \]
    Set
    \[
        s_N
        :=
        \tr{\Theta_N\adj(\mbI_D)}.
    \]
    By \Cref{lemma:deterministic_rank_one_reference}, applied to \(T=\Theta_N\), we have \(s_N>0\).
    Define
    \[
        \widehat\Theta_N
        :=
        \frac{\Theta_N}{s_N}.
    \]
    Since the same scalar \(s_N\) appears in the numerator and denominator of \eqref{eq:ratio-TDL-combined}, we obtain
    \begin{equation}
    \label{eq:ratio-TDL-combined-normalized}
        \frac{\bra{\Psi_N^\omega}\,\mbO_{[m,n]}\,\ket{\Psi_N^\omega}}
             {\inner{\Psi_N^\omega}{\Psi_N^\omega}}
        =
        \frac{
            \Tr{
                \mcO_{[m,n]}
                \circ
                \widehat\Theta_N
            }
        }{
            \Tr{
                \Phi_{[m,n]}
                \circ
                \widehat\Theta_N
            }
        }.
    \end{equation}

    Let
    \[
        \rho_*:=\frac{\mbI_D}{D}.
    \]
    Define
    \[
        r_N
        :=
        \Theta_N\proj\rho_*,
        \qquad
        \ell_N
        :=
        \Theta_N\adj\proj\rho_*.
    \]
    Equivalently,
    \[
        \ell_N
        =
        \frac{\Theta_N\adj(\mbI_D)}
             {\tr{\Theta_N\adj(\mbI_D)}}.
    \]
    Let
    \[
        \Xi_N(X)
        :=
        \tr{\ell_N X}\,r_N,
        \qquad
        X\in\mbM_D.
    \]
    By \Cref{lemma:deterministic_rank_one_reference},
    \begin{equation}
    \label{eq:theta-rank-one-error}
        \norm{\widehat\Theta_N-\Xi_N}_{1\to1}
        \le
        4\,\cnum{\Theta_N}.
    \end{equation}

    We next show that \(\cnum{\Theta_N}\to0\).
    By submultiplicativity of the contraction coefficient,
    \[
        \cnum{\Theta_N}
        =
        \cnum{
            \Phi_N^{\mathrm L}
            \circ
            \Phi_N^{\mathrm R}
        }
        \le
        \cnum{\Phi_N^{\mathrm L}}\,
        \cnum{\Phi_N^{\mathrm R}}.
    \]
    By \Cref{lemma:boundary_states_quantitative}\emph{(C)},
    \[
        \cnum{\Phi_N^{\mathrm R}}
        =
        \cnum{\phi_N\circ\cdots\circ\phi_{n+1}}
        \longrightarrow0,
    \]
    and
    \[
        \cnum{\Phi_N^{\mathrm L}}
        =
        \cnum{\phi_{m-1}\circ\cdots\circ\phi_{-N}}
        \longrightarrow0.
    \]
    Hence
    \begin{equation}
    \label{eq:theta-cnum-to-zero}
        \cnum{\Theta_N}\longrightarrow0.
    \end{equation}

    We now identify the limits of \(r_N\) and \(\ell_N\).
    Since
    \[
        r_N
        =
        \Theta_N\proj\rho_*
        =
        \Phi_N^{\mathrm L}
        \proj
        \left(
            \Phi_N^{\mathrm R}\proj\rho_*
        \right),
    \]
    and since \(\Phi_N^{\mathrm R}\proj\rho_*\in\states\), the uniform projective limit in \Cref{lemma:boundary_states_quantitative}\emph{(A)} gives
    \[
    \begin{aligned}
        \norm{r_N-Z_{m-1}(\omega)}_1
        &\le
        \sup_{\rho\in\states}
        \norm{
            \Phi_N^{\mathrm L}\proj\rho
            -
            Z_{m-1}(\omega)
        }_1
        \\
        &\longrightarrow0.
    \end{aligned}
    \]
    Thus
    \begin{equation}
    \label{eq:rN-limit}
        r_N\longrightarrow Z_{m-1}(\omega)
        \qquad
        \text{in }\norm{\,\cdot\,}_1.
    \end{equation}

    Similarly,
    \[
        \Theta_N\adj
        =
        \left(\Phi_N^{\mathrm R}\right)\adj
        \circ
        \left(\Phi_N^{\mathrm L}\right)\adj.
    \]
    Therefore
    \[
        \ell_N
        =
        \left(\Phi_N^{\mathrm R}\right)\adj
        \proj
        \left(
            \left(\Phi_N^{\mathrm L}\right)\adj\proj\rho_*
        \right).
    \]
    Since \(\left(\Phi_N^{\mathrm L}\right)\adj\proj\rho_*\in\states\), the uniform adjoint projective limit in \Cref{lemma:boundary_states_quantitative}\emph{(A)} gives
    \[
    \begin{aligned}
        \norm{\ell_N-Z'_{n+1}(\omega)}_1
        &\le
        \sup_{\sigma\in\states}
        \norm{
            \left(\Phi_N^{\mathrm R}\right)\adj\proj\sigma
            -
            Z'_{n+1}(\omega)
        }_1
        \\
        &\longrightarrow0.
    \end{aligned}
    \]
    Thus
    \begin{equation}
    \label{eq:ellN-limit}
        \ell_N\longrightarrow Z'_{n+1}(\omega)
        \qquad
        \text{in }\norm{\,\cdot\,}_1.
    \end{equation}

    Define the limiting rank--one map
    \[
        \Xi_\infty(X)
        :=
        \tr{Z'_{n+1}(\omega)X}\,Z_{m-1}(\omega),
        \qquad
        X\in\mbM_D.
    \]
    We claim that
    \[
        \Xi_N\longrightarrow\Xi_\infty
        \qquad
        \text{in }\norm{\,\cdot\,}_{1\to1}.
    \]
    Indeed, for \(\norm{X}_1\le1\), we have
    \[
    \begin{aligned}
        \norm{\Xi_N(X)-\Xi_\infty(X)}_1
        &\le
        \left|
            \tr{
                \left(
                    \ell_N-Z'_{n+1}(\omega)
                \right)X
            }
        \right|
        \norm{r_N}_1
        +
        \left|
            \tr{Z'_{n+1}(\omega)X}
        \right|
        \norm{
            r_N-Z_{m-1}(\omega)
        }_1
        \\
        &\le
        \norm{\ell_N-Z'_{n+1}(\omega)}_1
        +
        \norm{r_N-Z_{m-1}(\omega)}_1.
    \end{aligned}
    \]
    The right-hand side tends to zero by \eqref{eq:rN-limit} and \eqref{eq:ellN-limit}.
    Hence
    \begin{equation}
    \label{eq:XiN-to-Xiinfty}
        \Xi_N\longrightarrow\Xi_\infty
        \qquad
        \text{in }\norm{\,\cdot\,}_{1\to1}.
    \end{equation}
    Combining \eqref{eq:theta-rank-one-error}, \eqref{eq:theta-cnum-to-zero}, and \eqref{eq:XiN-to-Xiinfty}, we obtain
    \begin{equation}
    \label{eq:ThetaHat-to-Xiinfty}
        \widehat\Theta_N\longrightarrow\Xi_\infty
        \qquad
        \text{in }\norm{\,\cdot\,}_{1\to1}.
    \end{equation}

    Let \(S\in\mapspace\) be fixed.
    We claim that the linear functional
    \[
        E\mapsto \Tr{S\circ E}
    \]
    is continuous with respect to \(\norm{\,\cdot\,}_{1\to1}\).
    Indeed, let \(F_{ij}:=\ket{e_i}\bra{e_j}\), where \(\{e_i\}_{i=1}^D\) is an orthonormal basis of \(\mbC^D\).
    By the definition of the superoperator trace,
    \[
        \Tr{S\circ E}
        =
        \sum_{i,j=1}^D
        \tr{
            F_{ji}\,
            S\left(E(F_{ij})\right)
        }.
    \]
    Hence
    \[
    \begin{aligned}
        \left|\Tr{S\circ E}\right|
        &\le
        \sum_{i,j=1}^D
        \norm{F_{ji}}_\infty\,
        \norm{S\left(E(F_{ij})\right)}_1
        \\
        &\le
        \norm{S}_{1\to1}
        \sum_{i,j=1}^D
        \norm{E(F_{ij})}_1
        \\
        &\le
        \norm{S}_{1\to1}\,
        \norm{E}_{1\to1}
        \sum_{i,j=1}^D
        \norm{F_{ij}}_1
        \\
        &=
        D^2\,\norm{S}_{1\to1}\,\norm{E}_{1\to1}.
    \end{aligned}
    \]
    Therefore, if \(E_N\to E\) in \(\norm{\,\cdot\,}_{1\to1}\), then
    \begin{equation}
    \label{eq:supertrace-limit-general}
         \Tr{S\circ E_N}
        \longrightarrow
        \Tr{S\circ E}.
    \end{equation}
    Applying this with \(E_N=\widehat\Theta_N\) and \(E=\Xi_\infty\), and using \eqref{eq:ThetaHat-to-Xiinfty}, gives
    \[
        \Tr{S\circ\widehat\Theta_N}
        \longrightarrow
        \Tr{S\circ\Xi_\infty}.
    \]

    We compute the limiting trace.
    For \(X\in\mbM_D\),
    \[
        \left(S\circ\Xi_\infty\right)(X)
        =
        \tr{Z'_{n+1}(\omega)X}\,
        S\left(Z_{m-1}(\omega)\right).
    \]
    Using the rank--one trace identity
    \[
        \Tr{
            X\mapsto\tr{LX}R
        }
        =
        \tr{LR},
    \]
    we obtain
    \begin{equation}
    \label{eq:rank-one-trace-computation}
        \Tr{S\circ\Xi_\infty}
        =
        \tr{
            Z'_{n+1}(\omega)\,
            S\left(Z_{m-1}(\omega)\right)
        }.
    \end{equation}

    Applying \eqref{eq:supertrace-limit-general} and \eqref{eq:rank-one-trace-computation} with \(S=\mcO_{[m,n]}\), we get
    \begin{equation}
    \label{eq:numerator-limit-new}
        \Tr{
            \mcO_{[m,n]}
            \circ
            \widehat\Theta_N
        }
        \longrightarrow
        \tr{
            Z'_{n+1}(\omega)\,
            \mcO_{[m,n]}
            \left(
                Z_{m-1}(\omega)
            \right)
        }.
    \end{equation}
    Applying the same argument with \(S=\Phi_{[m,n]}\), we get
    \begin{equation}
    \label{eq:denominator-limit-new}
        \Tr{
            \Phi_{[m,n]}
            \circ
            \widehat\Theta_N
        }
        \longrightarrow
        \tr{
            Z'_{n+1}(\omega)\,
            \Phi_{[m,n]}
            \left(
                Z_{m-1}(\omega)
            \right)
        }.
    \end{equation}

    The limiting denominator is strictly positive.
    Indeed, \(Z'_{n+1}(\omega)\in\states^{\mathrm o}\) by \Cref{lemma:boundary_states_quantitative}.
    Moreover, \(\Phi_{[m,n]}\left(Z_{m-1}(\omega)\right)\ge0\), and repeated use of \Cref{assumption1} shows that this positive matrix is nonzero.
    Therefore
    \[
        \tr{
            Z'_{n+1}(\omega)\,
            \Phi_{[m,n]}
            \left(
                Z_{m-1}(\omega)
            \right)
        }
        >
        0.
    \]
    Moreover, by the finite-volume transfer-operator identity, cyclicity of the superoperator trace, and as \(\Theta_N=s_N\widehat\Theta_N\), we have that 
    \[
        \begin{aligned}
            \frac{\inner{\Psi_N^\omega}{\Psi_N^\omega}}{s_N}
            &=
            \frac{1}{s_N}
            \Tr{
                \Phi_N^{\mathrm R}
                \circ
                \Phi_{[m,n]}
                \circ
                \Phi_N^{\mathrm L}
            }
            \\
            &=
            \Tr{
                \Phi_{[m,n]}
                \circ
                \widehat\Theta_N
            }
            \longrightarrow
                    \tr{
                    Z'_{n+1}(\omega)\,
                    \Phi_{[m,n]}
                    \left(
                        Z_{m-1}(\omega)
                    \right)
                }
                >
                0.
        \end{aligned}
    \]
    Since \(s_N>0\), it follows that $\inner{\Psi_N^\omega}{\Psi_N^\omega}>0$ for all sufficiently large \(N\).
    Therefore, from \eqref{eq:ratio-TDL-combined-normalized}, \eqref{eq:numerator-limit-new}, and \eqref{eq:denominator-limit-new} it follows that
    \[
        \lim_{N\to\infty}
        \frac{\bra{\Psi_N^\omega}\,\mbO_{[m,n]}\,\ket{\Psi_N^\omega}}
             {\inner{\Psi_N^\omega}{\Psi_N^\omega}}
        =
        \frac{
            \tr{
                Z'_{n+1}(\omega)\,
                \mcO_{[m,n]}
                \left(
                    Z_{m-1}(\omega)
                \right)
            }
        }{
            \tr{
                Z'_{n+1}(\omega)\,
                \Phi_{[m,n]}
                \left(
                    Z_{m-1}(\omega)
                \right)
            }
        }.
    \]
    This is exactly \eqref{eq:TDL-functional}.
\end{proof}


\subsection{Dynamic Gauge Fixing}
\label{section:gauge_transform}


    One might impose at the outset that, for each \(n\), the local tensors \(\mcA_n=(A^{(n)}_1,\ldots,A^{(n)}_d)\) satisfy
    \[
        \sum_{i=1}^d
        (A_i^{(n)})\adj A_i^{(n)}
        =
        \mbI_D.
    \]
    This is equivalent to trace preservation of the associated transfer map \(\phi_n^\omega\). 
    This condition was not needed for the thermodynamic limit.
    Nevertheless, it is useful for the correlation estimates below to pass to a dynamically gauged cocycle whose transfer maps are CPTP.
    Following the dynamic gauge fixing introduced in \cite{Movassagh_2022} in the ergodic setting, we now describe this reduction.
    
    Assume \Cref{assumption1,assumption2}.
    Let
    \[
        \{Z_n(\omega)\}_{n\in\mbZ}\subset\states^{\mathrm o},
        \qquad
        \{Z_n'(\omega)\}_{n\in\mbZ}\subset\states^{\mathrm o}
    \]
    be the boundary families furnished by \Cref{lemma:boundary_states_main}.
    We choose these families equivariantly, so that
    \[
        Z_n(\omega)=Z_0(\theta^n\omega),
        \qquad
        Z_n'(\omega)=Z_0'(\theta^n\omega),
        \qquad
        n\in\mbZ,
    \]
    on the full-measure invariant event from \Cref{lemma:boundary_states_main}.
    This follows from the uniqueness of the projective limits.
    Let
    \[
        E_{\mathrm{ker}}
        :=
        \left\{
            \omega\in\Omega:
            \ker{\phi_0^\omega}\cap\states=\emptyset
            \text{ and }
            \ker{\left(\phi_0^\omega\right)\adj}\cap\states=\emptyset
        \right\}.
    \]
    By \Cref{assumption1}, \(\pr\left(E_{\mathrm{ker}}\right)=1\).
    Replacing the full-measure invariant event above by its intersection with
    \[
        \bigcap_{j\in\mbZ}\theta^{-j}E_{\mathrm{ker}},
    \]
    we may and shall assume that, for every \(n\in\mbZ\),
    \[
        \ker{\phi_n^\omega}\cap\states=\emptyset,
        \qquad
        \ker{\left(\phi_n^\omega\right)\adj}\cap\states=\emptyset.
    \]
    Here we note that if necessary we may further restrict to the full probability event $\cap_{j\in\mbZ}\theta^{-j} E_{\mathrm{ker}} \cap \Omega_0$ where $\Omega_0$ is the full probability event from \Cref{lemma:boundary_states_main}. 
    
    For each \(n\in\mbZ\), set
    \[
        \alpha_n(\omega)
        :=
        \tr{(\phi_n^\omega)\adj\left(Z_{n+1}'(\omega)\right)}.
    \]
    Since \(Z_{n+1}'(\omega)\in\states\), and since \((\phi_n^\omega)\adj\) does not annihilate states on the event under consideration, we have
    \[
        (\phi_n^\omega)\adj\left(Z_{n+1}'(\omega)\right)\ne0.
    \]
    Moreover, this matrix is positive semidefinite.
    Consequently,
    \[
        \alpha_n(\omega)>0.
    \]
    
    Define the gauge-transformed local tensors by
    \begin{equation}
    \label{eq:gauge_fix}
        B_i^{(n),\omega}
        :=
        \frac{1}{\sqrt{\alpha_n(\omega)}}\,
        \left(Z_{n+1}'(\omega)\right)^{1/2}
        A_i^{(n),\omega}
        \left(Z_n'(\omega)\right)^{-1/2},
        \qquad
        i=1,\ldots,d.
    \end{equation}
    Let
    \[
        \mcB_n^\omega
        :=
        \left(B_1^{(n),\omega},\ldots,B_d^{(n),\omega}\right),
    \]
    and let \(\widetilde\phi_n^\omega\) be the transfer map generated by \(\mcB_n^\omega\):
    \begin{equation}
    \label{eq:_transformed_trans_op}
        \widetilde\phi_n^\omega(X)
        :=
        \sum_{i=1}^d
        B_i^{(n),\omega}\,X\,(B_i^{(n),\omega})\adj.
    \end{equation}
    By construction, \(\widetilde\phi_n^\omega\) is completely positive.
    
    We now verify that \(\widetilde\phi_n^\omega\) is trace-preserving.
    For every \(\rho\in\mbM_D\),
    \[
    \begin{aligned}
        \tr{\widetilde\phi_n^\omega(\rho)}
        &=
        \frac{1}{\alpha_n(\omega)}
        \tr{
            \left(Z_{n+1}'(\omega)\right)^{1/2}
            \phi_n^\omega\left(
                \left(Z_n'(\omega)\right)^{-1/2}
                \rho
                \left(Z_n'(\omega)\right)^{-1/2}
            \right)
            \left(Z_{n+1}'(\omega)\right)^{1/2}
        }
        \\
        &=
        \frac{1}{\alpha_n(\omega)}
        \tr{
            Z_{n+1}'(\omega)\,
            \phi_n^\omega\left(
                \left(Z_n'(\omega)\right)^{-1/2}
                \rho
                \left(Z_n'(\omega)\right)^{-1/2}
            \right)
        }
        \\
        &=
        \frac{1}{\alpha_n(\omega)}
        \tr{
            (\phi_n^\omega)\adj\left(Z_{n+1}'(\omega)\right)
            \left(Z_n'(\omega)\right)^{-1/2}
            \rho
            \left(Z_n'(\omega)\right)^{-1/2}
        }.
    \end{aligned}
    \]
    By the dual cocycle relation from \Cref{lemma:boundary_states_main},
    \[
        (\phi_n^\omega)\adj\left(Z_{n+1}'(\omega)\right)
        =
        \tr{(\phi_n^\omega)\adj\left(Z_{n+1}'(\omega)\right)}
        Z_n'(\omega)
        =
        \alpha_n(\omega)Z_n'(\omega).
    \]
    Substituting this identity into the preceding display gives
    \[
    \begin{aligned}
        \tr{\widetilde\phi_n^\omega(\rho)}
        &=
        \tr{
            Z_n'(\omega)
            \left(Z_n'(\omega)\right)^{-1/2}
            \rho
            \left(Z_n'(\omega)\right)^{-1/2}
        }
        \\
        &=
        \tr{
            \left(Z_n'(\omega)\right)^{-1/2}
            Z_n'(\omega)
            \left(Z_n'(\omega)\right)^{-1/2}
            \rho
        }
        \\
        &=
        \tr{\rho}.
    \end{aligned}
    \]
    Thus each \(\widetilde\phi_n^\omega\) is trace-preserving.
    Since each \(\widetilde\phi_n^\omega\) is also completely positive, it is CPTP.
       
    \begin{figure}
        \centering
        \begin{tikzpicture}
        \tikzstyle{vertical} = [thin, blue]
            \draw[very thin] (-6,0) --(-5,0); 
            \draw[very thin] (-6,3) --(-5,3); 
            \draw[very thin] (-5,0) --(-2,0); 
            \draw[very thin] (-5,3) --(-2,3); 
            \draw[very thin] (-2,0) --(-1,0); 
            \draw[very thin] (-2,3) --(-1,3); 
            \draw[very thin] (6,0) --(7,0); 
            \draw[very thin] (6,3) --(7,3); 
            \draw[very thin] (6,0) --(2,0); 
            \draw[very thin] (6,3) --(2,3); 
            \draw[very thin] (2,0) --(1,0); 
            \draw[very thin] (2,3) --(1,3); 
            \node[tensor] (a) at (-5,0) {$\mathcal{A}_{n}$};
            \node[tensor] (b) at (-2,0) {$\mathcal{A}_{n-1}$};
            \node at (0.5,0) {$\ldots$};
            \node[tensor] (c) at (3,0) {$\mathcal{A}_{m+1}$};
            \node[tensor] (d) at (6,0) {$\mathcal{A}_{m}$};
            \node[tensorT] (e) at (-5,3) {$\overline{\mathcal{A}}_{n}$};
            \node[tensorT] (f) at (-2,3) {$\overline{\mathcal{A}}_{n-1}$};
            \node at (0.5,3) {$\ldots$};
            \node[tensorT] (g) at (3,3) {$\overline{\mathcal{A}}_{m+1}$};
            \node[tensorT] (h) at (6,3) {$\overline{\mathcal{A}}_{m}$};
            \begin{scope}[thick, black, decoration={
                markings,
                mark=at position 0.5 with {\arrow[newred]{<}}
                }] 
                \draw[postaction={decorate}]  (-7,0) -- (-6,0);
                \draw[postaction={decorate}]  (-4,0) -- (-3,0);
                \draw[postaction={decorate}]  (-1,0) -- (0,0);
                \draw[postaction={decorate}]  (1,0) -- (2,0);
                \draw[postaction={decorate}]  (4,0) -- (5,0);
                \draw[postaction={decorate}]  (7,0) -- (8,0);
                \draw[postaction={decorate}]  (-7,3) -- (-6,3);
                \draw[postaction={decorate}]  (-4,3) -- (-3,3);
                \draw[postaction={decorate}]  (-1,3) -- (0,3);
                \draw[postaction={decorate}]  (1,3) -- (2,3);
                \draw[postaction={decorate}]  (4,3) -- (5,3);
                \draw[postaction={decorate}]  (7,3) -- (8,3);
            \end{scope}
            \draw[vertical] (a) -- (e);
            \draw[vertical] (b) -- (f);
            \draw[vertical] (c) -- (g);
            \draw[vertical] (d) -- (h);
            \draw[fill=newyellow!40, rounded corners=0.2cm] (-5.5, 1) rectangle (6.5, 2);
            \draw (0.5,1.5) node[] {$\mbO_{[m,n]}$};
            \draw[blue!50, fill=newgreen!50] (-6,0) circle(0.25);
            \draw (-6,-0.5) node[] {${M_{n+1}}$};
            \draw[blue!50, fill=newgreen!50] (-3,0) circle(0.25);
            \draw (-3,-0.5) node[] {${M_{n}}$};
            \draw[blue!50, fill=newgreen!50] (2,0) circle(0.25);
            \draw (2,-0.5) node[] {${M_{m+2}}$};
            \draw[blue!50, fill=newgreen!50] (5,0) circle(0.25);
            \draw (5,-0.5) node[] {${M_{m+1}}$};
            \draw[blue!50, fill=newred!50] (-4,0) circle(0.25);
            \draw (-4,0.5) node[] {${M^{-1}_{n}}$};
            \draw[blue!50, fill=newred!50] (-1,0) circle(0.25);
            \draw (-1,0.5) node[] {${M^{-1}_{n-1}}$};
            \draw[blue!50, fill=newred!50] (4,0) circle(0.25);
            \draw (4,0.5) node[] {${M^{-1}_{m+1}}$};
            \draw[blue!50, fill=newred!50] (7,0) circle(0.25);
            \draw (7,0.5) node[] {${M^{-1}_{m}}$};
            \draw[blue!50, fill=newgreen!50] (-6,3) circle(0.25);
            \draw[blue!50, fill=newgreen!50] (-3,3) circle(0.25);
            \draw[blue!50, fill=newgreen!50] (2,3) circle(0.25);
            \draw[blue!50, fill=newgreen!50] (5,3) circle(0.25);
            \draw[blue!50, fill=newred!50] (-4,3) circle(0.25);
            \draw[blue!50, fill=newred!50] (-1,3) circle(0.25);
            \draw[blue!50, fill=newred!50] (4,3) circle(0.25);
            \draw[blue!50, fill=newred!50] (7,3) circle(0.25);
        \end{tikzpicture}
    \caption{Dynamic Gauge Fixing. Here $M_n = \sqrt{Z'_n}$.}
    \label{fig:gauge}
    \end{figure}
    
    If the original transfer maps \((\phi_n^\omega)\) are already trace-preserving, then one may choose
    \[
        Z_n'(\omega)=\frac{\mbI_D}{D}
        \qquad(n\in\mbZ).
    \]
    With this choice, \eqref{eq:gauge_fix} gives \(B_i^{(n),\omega}=A_i^{(n),\omega}\). Thus the gauge transform is trivial in the CPTP case.
    
    For later use, introduce the positive conjugation maps
    \[
        M_i^\omega(X):=\sqrt{Z_i'(\omega)}\,X\,\sqrt{Z_i'(\omega)},
        \qquad
        N_i^\omega(X):=\bigl(Z_i'(\omega)\bigr)^{-1/2}X\bigl(Z_i'(\omega)\bigr)^{-1/2}.
    \]
    Then
    \[
        \widetilde\phi_n^\omega
        =
        \alpha_n(\omega)^{-1}\,
        M_{n+1}^\omega\circ \phi_n^\omega\circ N_n^\omega.
    \]
    Consequently, for every \(a\le b\),
    \begin{equation}
    \label{eq:Phi-tilde-conj}
        \widetilde\Phi_{[a,b]}^\omega
        :=
        \widetilde\phi_b^\omega\circ\cdots\circ\widetilde\phi_a^\omega
        =
        \left(\prod_{k=a}^b \alpha_k(\omega)^{-1}\right)
        M_{b+1}^\omega\circ \Phi_{[a,b]}^\omega\circ N_a^\omega.
    \end{equation}

    \begin{prop}
    \label{prop:gauge_basic}
        Assume \Cref{assumption1,assumption2}. Then the gauge-transformed cocycle \((\widetilde\phi_n^\omega)_{n\in\mbZ}\) is strictly stationary and satisfies \Cref{assumption1,assumption2}. 
        Moreover, its dual boundary states are
        \[
            \widetilde Z_n'(\omega)=\frac{\mbI_D}{D},
            \qquad n\in\mbZ,
        \]
        and its right boundary states are given by
        \begin{equation}
        \label{eq:widetilde_z_formula}
            \widetilde Z_k(\omega)
            =
            \frac{
                \sqrt{Z_{k+1}'(\omega)}\,Z_k(\omega)\,\sqrt{Z_{k+1}'(\omega)}
            }{
                \tr{Z_{k+1}'(\omega)\,Z_k(\omega)}
            }.
        \end{equation}
    \end{prop}
    \begin{proof}
        The equivariance of the boundary families implies that
        \[
            B_i^{(n),\omega}
            =
            B_i^{(0),\theta^n\omega}
        \]
        on the invariant full-measure event under consideration.
        After modifying the tensors arbitrarily on the complement of this event, the gauged local tensors may be realized as a strictly stationary sequence.
        
        We next verify the assumptions. Since \(M_i^\omega\) and \(N_i^\omega\) are invertible positive conjugations, \eqref{eq:Phi-tilde-conj} implies that \(\widetilde\Phi_{[a,b]}^\omega\) is strictly positive if and only if \(\Phi_{[a,b]}^\omega\) is strictly positive. Thus \Cref{assumption2} passes from \((\phi_n^\omega)\) to \((\widetilde\phi_n^\omega)\).
        
        For \Cref{assumption1}, let \(\rho\in\states\). Since \(N_n^\omega(\rho)\) is nonzero and positive,
        \[
            \widetilde\phi_n^\omega(\rho)=0
            \quad\Longleftrightarrow\quad
            \phi_n^\omega\bigl(N_n^\omega(\rho)\bigr)=0.
        \]
        After normalizing \(N_n^\omega(\rho)\), this would contradict \Cref{assumption1} for \(\phi_n^\omega\). Hence
        \[
            \ker{\widetilde\phi_n^\omega}\cap\states=\emptyset.
        \]
        Similarly,
        \[
            (\widetilde\phi_n^\omega)\adj
            =
            \alpha_n(\omega)^{-1}\,
            N_n^\omega\circ(\phi_n^\omega)\adj\circ M_{n+1}^\omega,
        \]
        and since \(M_{n+1}^\omega(\rho)\) is nonzero and positive whenever \(\rho\in\states\), \Cref{assumption1} for \((\phi_n^\omega)\adj\) implies
        \[
            \ker{(\widetilde\phi_n^\omega)\adj}\cap\states=\emptyset.
        \]
        Thus, the transformed cocycle satisfies \Cref{assumption1}.

        Since \(\widetilde\phi_n^\omega\) is trace-preserving, its adjoint is unital:
        \[
            (\widetilde\phi_n^\omega)\adj(\mbI_D)=\mbI_D.
        \]
        Hence
        \[
            (\widetilde\phi_n^\omega)\adj\proj\frac{\mbI_D}{D}
            =
            \frac{\mbI_D}{D}.
        \]
        Applying the dual projective limit from \Cref{lemma:boundary_states_main} to the transformed cocycle, and choosing the initial state \(\mbI_D/D\), gives
        \[
            \widetilde Z_n'(\omega)=\frac{\mbI_D}{D}.
        \]
        
        It remains to identify \(\widetilde Z_k\). Fix \(k\in\mbZ\). From \eqref{eq:Phi-tilde-conj}, for every \(\rho\in\states\),
        \[
            \widetilde\Phi_{[-N,k]}^\omega\proj\rho
            =
            M_{k+1}^\omega\proj
            \left(
                \Phi_{[-N,k]}^\omega\proj
                \bigl(N_{-N}^\omega\proj\rho\bigr)
            \right).
        \]
        Since \(N_{-N}^\omega\proj\rho\in\states\), \Cref{lemma:boundary_states_main} gives
        \[
            \Phi_{[-N,k]}^\omega\proj
            \bigl(N_{-N}^\omega\proj\rho\bigr)
            \longrightarrow
            Z_k(\omega)
        \]
        as \(N\to\infty\), uniformly in \(\rho\in\states\). By continuity of the projective action of \(M_{k+1}^\omega\),
        \[
            \widetilde\Phi_{[-N,k]}^\omega\proj\rho
            \longrightarrow
            M_{k+1}^\omega\proj Z_k(\omega)
            =
            \frac{
                \sqrt{Z_{k+1}'(\omega)}\,Z_k(\omega)\,\sqrt{Z_{k+1}'(\omega)}
            }{
                \tr{Z_{k+1}'(\omega)\,Z_k(\omega)}
            }.
        \]
        By uniqueness of the projective limit for the transformed cocycle, this limit is \(\widetilde Z_k(\omega)\). This proves \eqref{eq:widetilde_z_formula}.
    \end{proof}

    Let \(\widetilde{\mcO}_{[m,n]}^\omega\) be the observable-dependent transfer map obtained from \(\mcB_m^\omega,\ldots,\mcB_n^\omega\).
    Dynamic gauge fixing preserves the infinite-volume expectation functional in the following sense.

    \begin{prop}
        \label{prop:gauge_thermo}
        Assume \Cref{assumption1,assumption2}.
        Let \(\mbO_{[m,n]}\) be a local observable supported on \([m,n]\), and denote by \(\ket{\widetilde\Psi_N^\omega}\) the periodic MPS on \(2N+1\) sites obtained from the gauge-transformed tensors \((\mcB_k)_{k\in\mbZ}\). 
        Then, for \(\pr\)-almost every \(\omega\),
        \[
            \lim_{N\to\infty}
            \frac{\bra{\widetilde\Psi_N^\omega}\,\mbO_{[m,n]}\,\ket{\widetilde\Psi_N^\omega}}
                 {\inner{\widetilde\Psi_N^\omega}{\widetilde\Psi_N^\omega}}
            =
            \lim_{N\to\infty}
            \frac{\bra{\Psi_N^\omega}\,\mbO_{[m,n]}\,\ket{\Psi_N^\omega}}
                 {\inner{\Psi_N^\omega}{\Psi_N^\omega}}
            =
            \mcT_\omega(\mbO_{[m,n]}).
        \]
        Here, both finite-volume norms appearing above are strictly positive for all sufficiently large \(N\). 
        Moreover,
        \begin{equation}
        \label{eq:TDL-gauge}
            \mcT_\omega(\mbO_{[m,n]})
            =
            \frac{1}{\tr{Z_m'(\omega)\,Z_{m-1}(\omega)}}\,
            \tr{
                \widetilde{\mcO}_{[m,n]}^\omega
                \!\left(
                    \sqrt{Z_m'(\omega)}\,Z_{m-1}(\omega)\,\sqrt{Z_m'(\omega)}
                \right)
            }.
        \end{equation}
    \end{prop}
    \begin{proof}
        Fix \(\omega\) in the full-measure event on which \Cref{lemma:boundary_states_main,thm:thermodynamic_limit} hold for the original cocycle and on which \Cref{prop:gauge_basic} holds for the transformed cocycle.
        
        By \Cref{prop:gauge_basic}, the transformed cocycle satisfies \Cref{assumption1,assumption2}.
        Therefore \Cref{thm:thermodynamic_limit} applies to the transformed MPS.
        Since \(\widetilde Z_{n+1}'=\mbI_D/D\), \Cref{thm:thermodynamic_limit} gives
        \[
            \widetilde{\mcT}_\omega(\mbO_{[m,n]})
            =
            \frac{
                \tr{
                    \frac{\mbI_D}{D}\,
                    \widetilde{\mcO}_{[m,n]}^\omega
                    \left(\widetilde Z_{m-1}(\omega)\right)
                }
            }{
                \tr{
                    \frac{\mbI_D}{D}\,
                    \widetilde\Phi_{[m,n]}^\omega
                    \left(\widetilde Z_{m-1}(\omega)\right)
                }
            }.
        \]
        Since \(\widetilde\Phi_{[m,n]}^\omega\) is trace-preserving and \(\widetilde Z_{m-1}(\omega)\in\states\), the denominator is \(1/D\).
        Therefore
        \[
            \widetilde{\mcT}_\omega(\mbO_{[m,n]})
            =
            \tr{
                \widetilde{\mcO}_{[m,n]}^\omega
                \left(\widetilde Z_{m-1}(\omega)\right)
            }.
        \]
        Using \eqref{eq:widetilde_z_formula} at \(k=m-1\), we have
        \[
            \widetilde Z_{m-1}(\omega)
            =
            \frac{
                \sqrt{Z_m'(\omega)}\,Z_{m-1}(\omega)\,\sqrt{Z_m'(\omega)}
            }{
                \tr{Z_m'(\omega)\,Z_{m-1}(\omega)}
            }.
        \]
        Hence
        \[
            \widetilde{\mcT}_\omega(\mbO_{[m,n]})
            =
            \frac{1}{\tr{Z_m'(\omega)\,Z_{m-1}(\omega)}}\,
            \tr{
                \widetilde{\mcO}_{[m,n]}^\omega
                \!\left(
                    \sqrt{Z_m'(\omega)}\,Z_{m-1}(\omega)\,\sqrt{Z_m'(\omega)}
                \right)
            }.
        \]
        
        It remains to identify this expression with \(\mcT_\omega(\mbO_{[m,n]})\). 
        Iterating the definition of the gauge transform gives
        \begin{equation}
        \label{eq:gauge-block-O-clean}
            \widetilde{\mcO}_{[m,n]}^\omega(X)
            =
            \frac{1}{\prod_{k=m}^n \alpha_k(\omega)}
            \bigl(Z_{n+1}'(\omega)\bigr)^{1/2}
            \mcO_{[m,n]}^\omega\!\left(
                \bigl(Z_m'(\omega)\bigr)^{-1/2}
                X
                \bigl(Z_m'(\omega)\bigr)^{-1/2}
            \right)
            \bigl(Z_{n+1}'(\omega)\bigr)^{1/2}.
        \end{equation}
        Similarly,
        \begin{equation}
        \label{eq:gauge-block-phi-clean}
            \widetilde\Phi_{[m,n]}^\omega(X)
            =
            \frac{1}{\prod_{k=m}^n \alpha_k(\omega)}
            \bigl(Z_{n+1}'(\omega)\bigr)^{1/2}
            \Phi_{[m,n]}^\omega\!\left(
                \bigl(Z_m'(\omega)\bigr)^{-1/2}
                X
                \bigl(Z_m'(\omega)\bigr)^{-1/2}
            \right)
            \bigl(Z_{n+1}'(\omega)\bigr)^{1/2}.
        \end{equation}
        
        Set
        \[
            X_m(\omega):=
            \sqrt{Z_m'(\omega)}\,Z_{m-1}(\omega)\,\sqrt{Z_m'(\omega)}.
        \]
        Then
        \[
            \bigl(Z_m'(\omega)\bigr)^{-1/2}
            X_m(\omega)
            \bigl(Z_m'(\omega)\bigr)^{-1/2}
            =
            Z_{m-1}(\omega).
        \]
        Thus, by \eqref{eq:gauge-block-O-clean},
        \[
            \tr{\widetilde{\mcO}_{[m,n]}^\omega(X_m(\omega))}
            =
            \frac{1}{\prod_{k=m}^n \alpha_k(\omega)}
            \tr{
                Z_{n+1}'(\omega)\,
                \mcO_{[m,n]}^\omega(Z_{m-1}(\omega))
            }.
        \]
        Likewise, by \eqref{eq:gauge-block-phi-clean},
        \[
            \tr{\widetilde\Phi_{[m,n]}^\omega(X_m(\omega))}
            =
            \frac{1}{\prod_{k=m}^n \alpha_k(\omega)}
            \tr{
                Z_{n+1}'(\omega)\,
                \Phi_{[m,n]}^\omega(Z_{m-1}(\omega))
            }.
        \]
        Since \(\widetilde\Phi_{[m,n]}^\omega\) is trace-preserving,
        \[
            \tr{\widetilde\Phi_{[m,n]}^\omega(X_m(\omega))}
            =
            \tr{X_m(\omega)}
            =
            \tr{Z_m'(\omega)\,Z_{m-1}(\omega)}.
        \]
        Combining the last three displays gives
        \[
            \frac{1}{\tr{Z_m'(\omega)\,Z_{m-1}(\omega)}}
            \tr{\widetilde{\mcO}_{[m,n]}^\omega(X_m(\omega))}
            =
            \frac{
                \tr{Z_{n+1}'(\omega)\,\mcO_{[m,n]}^\omega(Z_{m-1}(\omega))}
            }{
                \tr{Z_{n+1}'(\omega)\,\Phi_{[m,n]}^\omega(Z_{m-1}(\omega))}
            }.
        \]
        The right-hand side is \(\mcT_\omega(\mbO_{[m,n]})\) by \Cref{thm:thermodynamic_limit}. Therefore
        \[
            \widetilde{\mcT}_\omega(\mbO_{[m,n]})
            =
            \mcT_\omega(\mbO_{[m,n]}),
        \]
        which proves both the equality of the thermodynamic limits and \eqref{eq:TDL-gauge}.
    \end{proof}

    \begin{remark}
    \label{rem:gauge_replacement}
        For the transformed cocycle, the rank-one maps take the simple replacement form
        \[
            \widetilde\Xi_{[m,n]}^\omega(X)
            =
            \tr{X}\,\frac{\widetilde Z_n(\omega)}{D}.
        \]
        Indeed, \(\widetilde Z_m'(\omega)=\mbI_D/D\).
    \end{remark}

    \begin{figure}[hbt!]
        \hspace{-1cm}
        \centering{
            \scalebox{0.8}{
                \begin{tikzpicture}
                \tikzstyle{vertical} = [thin, blue]
                    \node[tensor] (a) at (-5,0) {$\mathcal{A}_{n}$};
                    \node[tensor] (b) at (-2,0) {$\mathcal{A}_{n-1}$};
                    \node at (0.5,0) {$\ldots$};
                    \node[tensor] (c) at (3,0) {$\mathcal{A}_{m+1}$};
                    \node[tensor] (d) at (6,0) {$\mathcal{A}_{m}$};    
                    \node[tensorT] (e) at (-5,3) {$\overline{\mathcal{A}}_{n}$};
                    \node[tensorT] (f) at (-2,3) {$\overline{\mathcal{A}}_{n-1}$};
                    \node at (0.5,3) {$\ldots$};
                    \node[tensorT] (g) at (3,3) {$\overline{\mathcal{A}}_{m+1}$};
                    \node[tensorT] (h) at (6,3) {$\overline{\mathcal{A}}_{m}$};
                    \draw[vertical] (a) -- (e);
                    \draw[vertical] (b) -- (f);
                    \draw[vertical] (c) -- (g);
                    \draw[vertical] (d) -- (h);
                    \draw (a) -- (b);
                    \draw (c) -- (d);
                    \draw (e) -- (f);
                    \draw (g) - -(h);
                    \draw (a) -- (-7,0);
                    \draw (d) -- (8,0);
                    \draw (e) -- (-7,3);
                    \draw (h) -- (8,3);
                    \draw (b) -- (-1,0);
                    \draw (f) -- (-1,3);
                    \draw (c) -- (2,0);
                    \draw (g) -- (2,3);
                    \draw (-7,0) -- (-7,3);
                    \draw (8,0) -- (8,3);
                    \draw[blue!50, fill=newgreen!50] (-7,1.5) circle(0.5);
                    \draw (-7,1.5) node[] {$Z'_{n+1}$};
                    \draw[blue!50, fill=newred!50] (8,1.5) circle(0.5);
                    \draw (8,1.5) node[] {$Z_{m-1}$};
                    \draw[fill=newyellow!40, rounded corners=0.2cm] (-5.5, 1) rectangle (6.5, 2);
                    \draw (0.5,1.5) node[] {$\mbO_{[m,n]}$}; 
                \end{tikzpicture}
            }\\
            \hspace{-0.8cm}=\\ \hspace{0.3cm}
            \scalebox{0.8}{
                \begin{tikzpicture}
                \tikzstyle{vertical} = [thin, blue]
                    \node[tensor] (a) at (-5,0) {$\mathcal{A}_{n}$};
                    \node[tensor] (b) at (-2,0) {$\mathcal{A}_{n-1}$};
                    \node at (0.5,0) {$\ldots$};
                    \node[tensor] (c) at (3,0) {$\mathcal{A}_{m+1}$};
                    \node[tensor] (d) at (6,0) {$\mathcal{A}_{m}$};    
                    \node[tensorT] (e) at (-5,3) {$\overline{\mathcal{A}}_{n}$};
                    \node[tensorT] (f) at (-2,3) {$\overline{\mathcal{A}}_{n-1}$};
                    \node at (0.5,3) {$\ldots$};
                    \node[tensorT] (g) at (3,3) {$\overline{\mathcal{A}}_{m+1}$};
                    \node[tensorT] (h) at (6,3) {$\overline{\mathcal{A}}_{m}$};
                    \draw[vertical] (a) -- (e);
                    \draw[vertical] (b) -- (f);
                    \draw[vertical] (c) -- (g);
                    \draw[vertical] (d) -- (h);
                    \draw (a) -- (b);
                    \draw (c) -- (d);
                    \draw (e) -- (f);
                    \draw (g) - -(h);
                    \draw (a) -- (-7,0);
                    \draw (d) -- (8,0);
                    \draw (e) -- (-7,3);
                    \draw (h) -- (8,3);
                    \draw (b) -- (-1,0);
                    \draw (f) -- (-1,3);
                    \draw (c) -- (2,0);
                    \draw (g) -- (2,3);
                    \draw (-7,0) -- (-7,3);
                    \draw (8,0) -- (8,3);
                    \draw[blue!50, fill=newred!50] (8,1.5) circle(0.5);
                    \draw (8,1.5) node[] {$Z_{m-1}$};
                    \draw[fill=newyellow!40, rounded corners=0.2cm] (-5.5, 1) rectangle (6.5, 2);
                    \draw (0.5,1.5) node[] {$\mbO_{[m,n]}$};
                    \draw[blue!50, fill=newgreen!50] (8,0.625) circle(0.25);
                    \draw[blue!50, fill=newgreen!50] (8,2.375) circle(0.25);
                    \draw[blue!50, fill=newgreen!50] (-6,0) circle(0.25);
                    \draw[blue!50, fill=newgreen!50] (-3,0) circle(0.25);
                    \draw[blue!50, fill=newgreen!50] (2,0) circle(0.25);
                    \draw[blue!50, fill=newgreen!50] (5,0) circle(0.25);
                    \draw[blue!50, fill=newred!50] (-4,0) circle(0.25);
                    \draw[blue!50, fill=newred!50] (-1,0) circle(0.25);
                    \draw[blue!50, fill=newred!50] (4,0) circle(0.25);
                    \draw[blue!50, fill=newred!50] (7,0) circle(0.25);
                    \draw[blue!50, fill=newgreen!50] (-6,3) circle(0.25);
                    \draw[blue!50, fill=newgreen!50] (-3,3) circle(0.25);
                    \draw[blue!50, fill=newgreen!50] (2,3) circle(0.25);
                    \draw[blue!50, fill=newgreen!50] (5,3) circle(0.25);
                    \draw[blue!50, fill=newred!50] (-4,3) circle(0.25);
                    \draw[blue!50, fill=newred!50] (-1,3) circle(0.25);
                    \draw[blue!50, fill=newred!50] (4,3) circle(0.25);
                    \draw[blue!50, fill=newred!50] (7,3) circle(0.25);
                    \draw (-6,0.5) node[] {$\sqrt{Z'_{n+1}}$};
                    \draw (-6,2.5) node[] {$\sqrt{Z'_{n+1}}$};
                    \draw (7,-0.5) node[] {$({Z'_{m}})^{-1/2}$};
                    \draw (7,3.5) node[] {$({Z'_{m}})^{-1/2}$};
                    \draw (9,0.625) node[] {$\sqrt{Z'_m}$};
                    \draw (9,2.375) node[] {$\sqrt{Z'_m}$};
                \end{tikzpicture}
            }\\
        }
        \caption{Tensor-network illustration of the gauge identity used in the proof of \Cref{prop:gauge_thermo}. The scalar factors \(\alpha_k^{-1/2}\) are omitted because their double-layer contribution cancels in the normalized expectation.}
        \label{fig:proof_of_cor}
    \end{figure}
  

\subsection{Almost Sure Exponential Decay -- \texorpdfstring{\Cref{thm:decay_random}}{Theorem B}}
\label{subsec:proof_decay_random}


    Equipped with \Cref{thm:thermodynamic_limit} and \Cref{prop:gauge_thermo}, we now derive the basic correlation estimate that underlies all subsequent decay results.

    The role of the dynamic gauge fixing is to reduce the problem to a CPTP cocycle.
    At this stage, we also introduce the contraction coefficient \(\cnum{\,\cdot\,}\), whose definition and basic properties are recalled in \Cref{dfn:cnum,lemma:cnum_properties}.
    Although \(\cnum{\,\cdot\,}\) did not appear in the main-results section, it now enters as a quantitative proof device.
    In particular, we invoke the quantitative refinement of \Cref{lemma:boundary_states_main}, proved later in \Cref{section:appen_1}.

    We first record the corresponding quantitative rank--one approximation for the gauged middle block.

    \begin{lemma}
    \label{lemma:gauged_rank_one_quant}
        Assume \Cref{assumption1,assumption2}. Let \((\widetilde\phi_n^\omega)_{n\in\mbZ}\) be the gauge-transformed cocycle from \Cref{section:gauge_transform}. 
       Then, for \(\pr\)-almost every \(\omega\), and simultaneously for every \(m\le n\), 
        \begin{equation}
        \label{eq:gauged_rank_one_quant}
            \norm{
                \frac{1}{D}\,\widetilde\Phi_{[m,n]}^\omega
                -
                \widetilde\Xi_{[m,n]}^\omega
            }_{1\to1}
            \le
            8\,\cnum{\widetilde\Phi_{[m,n]}^\omega},
        \end{equation}
        where
        \[
            \widetilde\Xi_{[m,n]}^\omega(X)
            =
            \tr{X}\,\frac{\widetilde Z_n(\omega)}{D}.
        \]
        Moreover,
        \begin{equation}
        \label{eq:cnum_gauge_compare}
            \cnum{\widetilde\Phi_{[m,n]}^\omega}
            \le
            \cnum{\Phi_{[m,n]}^\omega}.
        \end{equation}
    \end{lemma}
    \begin{proof}
        By \Cref{prop:gauge_basic}, the gauge-transformed cocycle satisfies \Cref{assumption1,assumption2}, and its dual boundary states are \(\widetilde Z_n'=\mbI_D/D\).
        Applying the quantitative refinement of \Cref{lemma:boundary_states_main} from \Cref{section:appen_1} to the transformed cocycle yields
        \[
            \norm{
                \frac{\widetilde\Phi_{[m,n]}^\omega}
                     {\tr{(\widetilde\Phi_{[m,n]}^\omega)\adj(\mbI_D)}}
                -
                \widetilde\Xi_{[m,n]}^\omega
            }_{1\to1}
            \le
            8\,\cnum{\widetilde\Phi_{[m,n]}^\omega}.
        \]
        Since \(\widetilde\Phi_{[m,n]}^\omega\) is trace-preserving, its adjoint is unital, and therefore
        \[
            \tr{(\widetilde\Phi_{[m,n]}^\omega)\adj(\mbI_D)}
            =
            \tr{\mbI_D}
            =
            D.
        \]
        This gives \eqref{eq:gauged_rank_one_quant}.
        To prove \eqref{eq:cnum_gauge_compare}, recall from \eqref{eq:Phi-tilde-conj} that
        \[
            \widetilde\Phi_{[m,n]}^\omega
            =
            c_\omega\,
            M_{n+1}^\omega\circ
            \Phi_{[m,n]}^\omega\circ
            N_m^\omega
        \]
        for some scalar \(c_\omega>0\).
        Since \(\cnum{\lambda\Psi}=\cnum{\Psi}\) for \(\lambda>0\), and since the projective actions of the invertible positive conjugations \(M_{n+1}^\omega\) and \(N_m^\omega\) are nonexpansive for the projective metric, submultiplicativity yields
        \[
            \cnum{\widetilde\Phi_{[m,n]}^\omega}
            \le
            \cnum{\Phi_{[m,n]}^\omega}.
        \]
    \end{proof}

    The previous lemma controls the long transfer block between the two observables.
    We next record a simple bound on the one-site gauged observable transfer maps.

    \begin{lemma}
    \label{lem:otilde_bd}
        For every single-site observable \(\mbO_k\), the corresponding gauge-transformed observable transfer map satisfies
        \[
            \norm{\widetilde{\mcO}_k}_{1\to1}
            \le
            \norm{\mbO_k}_\infty.
        \]
    \end{lemma} 
    \begin{proof}
        Using the dual characterization of the trace norm, we have
        \[
            \norm{\widetilde{\mcO}_k}_{1\to1}
            =
            \sup_{\norm{X}_1=1}
            \sup_{\norm{Y}_\infty\le1}
            \left|\tr{Y\adj\widetilde{\mcO}_k(X)}\right|.
        \]
        Writing \(\mcB_k=(B_1,\ldots,B_d)\), we obtain
        \[
            \tr{Y\adj\widetilde{\mcO}_k(X)}
            =
            \tr{
                \left[
                    \sum_{p,q=1}^d
                    (\mbO_k)_{p,q}\,
                    B_p\adj Y\adj B_q
                \right]X
            }.
        \]
        Hence
        \[
            \left|\tr{Y\adj\widetilde{\mcO}_k(X)}\right|
            \le
            \left\|
                \sum_{p,q=1}^d
                (\mbO_k)_{p,q}\,
                B_p\adj Y\adj B_q
            \right\|_\infty
            \norm{X}_1.
        \]
        Set
        \[
            M_Y
            :=
            \sum_{p,q=1}^d
            (\mbO_k)_{p,q}\,
            B_p\adj Y\adj B_q.
        \]
        For unit vectors \(u,v\in\mbC^D\),
        \[
            \inner{u}{M_Yv}
            =
            \sum_{p,q=1}^d
            (\mbO_k)_{p,q}\,
            \inner{YB_pu}{B_qv}.
        \]
        Let \(a=(YB_pu)_p\) and \(b=(B_qv)_q\), viewed in the direct-sum Hilbert space \((\mbC^D)^{\oplus d}\). Then
        \[
            \sum_{p,q=1}^d
            (\mbO_k)_{p,q}\,
            \inner{YB_pu}{B_qv}
            =
            \inner{a}{(\mbO_k\otimes\mbI_D)b}.
        \]
        By Cauchy--Schwarz,
        \[
            |\inner{u}{M_Yv}|
            \le
            \norm{\mbO_k}_\infty\,\norm{a}\,\norm{b}.
        \]
        Since \(\widetilde\phi_k\) is trace-preserving, \(\sum_i B_i\adj B_i=\mbI_D\). Therefore
        \[
            \norm{a}^2
            =
            \sum_p \norm{YB_pu}^2
            \le
            \norm{Y}_\infty^2
            \sum_p \inner{u}{B_p\adj B_p u}
            \le
            1,
        \]
        and
        \[
            \norm{b}^2
            =
            \sum_q \norm{B_qv}^2
            =
            \sum_q \inner{v}{B_q\adj B_q v}
            =
            1.
        \]
        Thus \(|\inner{u}{M_Yv}|\le \norm{\mbO_k}_\infty\). Taking the supremum over unit \(u,v\), and then over \(\norm{Y}_\infty\le1\), gives
        \[
            \norm{\widetilde{\mcO}_k}_{1\to1}
            \le
            \norm{\mbO_k}_\infty.
        \]
    \end{proof}

    We can now combine \Cref{prop:gauge_thermo,prop:gauge_basic,lemma:gauged_rank_one_quant,lem:otilde_bd} to obtain the basic almost-sure bound on connected correlations.

    \begin{lemma}
        \label{lemma:main_bound}
        Assume \Cref{assumption1,assumption2}.
        Fix sites \(m<n\) with \(n-m\ge2\), and let \(\mbO_m,\mbO_n\) be local observables supported at \(m\) and \(n\), respectively.
        Then
        \begin{equation}
        \label{thm:decay_ineq_2}
            f^\omega(n,m)
            \le
            8D\,\norm{\mbO_n}_\infty\,\norm{\mbO_m}_\infty\,
            \cnum{\phi_{n-1}^\omega\circ\cdots\circ\phi_{m+1}^\omega}
        \end{equation}
        for \(\pr\)-almost every \(\omega\).
    \end{lemma}

\begin{proof}
        By \Cref{prop:gauge_thermo}, \Cref{prop:gauge_basic}, and the definition of the gauge-transformed observable transfer maps,
        \[
            \mcT_\omega(\mbO_n\mbO_m)
            =
            \tr{
                \widetilde{\mcO}_n
                \circ
                \widetilde\Phi_{[m+1,n-1]}^\omega
                \circ
                \widetilde{\mcO}_m(\widetilde Z_{m-1})
            },
        \]
        where
        \[
            \widetilde\Phi_{[m+1,n-1]}^\omega
            =
            \widetilde\phi_{n-1}^\omega\circ\cdots\circ\widetilde\phi_{m+1}^\omega.
        \]
        By \Cref{rem:gauge_replacement},
        \[
            \widetilde\Xi_{[m+1,n-1]}^\omega(X)
            =
            \tr{X}\,\frac{\widetilde Z_{n-1}(\omega)}{D}.
        \]
        Thus \(D\,\widetilde\Xi_{[m+1,n-1]}^\omega(X)=\tr{X}\,\widetilde Z_{n-1}(\omega)\), and hence
        \[
            \tr{
                \widetilde{\mcO}_n
                \circ
                D\,\widetilde\Xi_{[m+1,n-1]}^\omega
                \circ
                \widetilde{\mcO}_m(\widetilde Z_{m-1})
            }
            =
            \mcT_\omega(\mbO_n)\,\mcT_\omega(\mbO_m).
        \]
        Therefore
        \[
            f^\omega(n,m)
            =
            \left|
                \tr{
                    \widetilde{\mcO}_n
                    \circ
                    \left(
                        \widetilde\Phi_{[m+1,n-1]}^\omega
                        -
                        D\,\widetilde\Xi_{[m+1,n-1]}^\omega
                    \right)
                    \circ
                    \widetilde{\mcO}_m(\widetilde Z_{m-1})
                }
            \right|.
        \]
        Using \(|\tr{X}|\le\norm{X}_1\), we obtain
        \[
            f^\omega(n,m)
            \le
            \norm{\widetilde{\mcO}_n}_{1\to1}\,
            \norm{
                \widetilde\Phi_{[m+1,n-1]}^\omega
                -
                D\,\widetilde\Xi_{[m+1,n-1]}^\omega
            }_{1\to1}\,
            \norm{\widetilde{\mcO}_m}_{1\to1}.
        \]
        By \Cref{lemma:gauged_rank_one_quant},
        \[
            \norm{
                \widetilde\Phi_{[m+1,n-1]}^\omega
                -
                D\,\widetilde\Xi_{[m+1,n-1]}^\omega
            }_{1\to1}
            \le
            8D\,\cnum{\widetilde\Phi_{[m+1,n-1]}^\omega}.
        \]
        By \Cref{lem:otilde_bd},
        \[
            \norm{\widetilde{\mcO}_n}_{1\to1}\le\norm{\mbO_n}_\infty,
            \qquad
            \norm{\widetilde{\mcO}_m}_{1\to1}\le\norm{\mbO_m}_\infty.
        \]
        Finally, \eqref{eq:cnum_gauge_compare} gives
        \[
            \cnum{\widetilde\Phi_{[m+1,n-1]}^\omega}
            \le
            \cnum{\Phi_{[m+1,n-1]}^\omega}
            =
            \cnum{\phi_{n-1}^\omega\circ\cdots\circ\phi_{m+1}^\omega}.
        \]
        Combining the preceding estimates yields \eqref{thm:decay_ineq_2}.
    \end{proof}

    We are now ready to prove \Cref{thm:decay_random}.

    \decayrandom*
    \begin{proof}
        By \Cref{lemma:boundary_states_quantitative}(C), there exists a \(\theta\)-invariant random variable
        \[
            \xi:\Omega\to[-\infty,0)
        \]
        such that, for every fixed \(x\in\mbZ\),
        \begin{equation}
        \label{eq:forward_xi_limit}
            \lim_{k\to\infty}
            \frac{1}{k}
            \ln \cnum{\phi_{x+k-1}^\omega\circ\cdots\circ\phi_x^\omega}
            =
            \xi(\omega),
        \end{equation}
        and
        \begin{equation}
        \label{eq:backward_xi_limit}
            \lim_{k\to\infty}
            \frac{1}{k}
            \ln \cnum{\phi_x^\omega\circ\cdots\circ\phi_{x-k+1}^\omega}
            =
            \xi(\omega),
        \end{equation}
        on a \(\theta\)-invariant full-measure event.

        By intersecting with the full-measure event on which \Cref{lemma:main_bound} holds, we may choose a full-measure event \(\Omega_*\subseteq\Omega\) on which \eqref{eq:forward_xi_limit}, \eqref{eq:backward_xi_limit}, and \Cref{lemma:main_bound} all hold simultaneously for every \(x\in\mbZ\).

        Set
        \[
            \mu(\omega)
            :=
            \frac{e^{\xi(\omega)}+1}{2}.
        \]
        Then
        \[
            \frac12 \le \mu(\omega) < 1
            \qquad\text{for }\pr\text{-almost every }\omega.
        \]
        For each \(x\in\mbZ\), define
        \begin{equation}
        \label{eq:Nxplus_def}
            N_x^+(\omega)
            :=
            \min\Bigl\{
                N\in\mbN:
                \cnum{\phi_{x+k-1}^\omega\circ\cdots\circ\phi_x^\omega}
                \le
                \mu(\omega)^k
                \text{ for all }k\ge N
            \Bigr\},
        \end{equation}
        and
        \begin{equation}
        \label{eq:Nxminus_def}
            N_x^-(\omega)
            :=
            \min\Bigl\{
                N\in\mbN:
                \cnum{\phi_x^\omega\circ\cdots\circ\phi_{x-k+1}^\omega}
                \le
                \mu(\omega)^k
                \text{ for all }k\ge N
            \Bigr\}.
        \end{equation}
        By \eqref{eq:forward_xi_limit} and \eqref{eq:backward_xi_limit}, these are almost surely finite random variables.

        Fix \(x\in\mbZ\), and fix \(\omega\in\Omega_*\).
        Let \(m<x<n\) with \(n-m\ge2\), and let \(\mbO_m,\mbO_n\) be local observables supported at \(m\) and \(n\), respectively.
        By \Cref{lemma:main_bound},
        \begin{equation}
        \label{eq:start_decay_random_proof}
            f^\omega(n,m)
            \le
            8D\,\norm{\mbO_n}_\infty\,\norm{\mbO_m}_\infty\,
            \cnum{\phi_{n-1}^\omega\circ\cdots\circ\phi_{m+1}^\omega}.
        \end{equation}

        We now split the middle block at the site \(x\).
        With the convention that an empty composition is the identity map and has contraction coefficient \(1\), submultiplicativity of \(\cnum{\,\cdot\,}\) yields
        \[
            \cnum{\phi_{n-1}^\omega\circ\cdots\circ\phi_{m+1}^\omega}
            \le
            \cnum{\phi_{n-1}^\omega\circ\cdots\circ\phi_{x+1}^\omega}\,
            \cnum{\phi_{x-1}^\omega\circ\cdots\circ\phi_{m+1}^\omega}.
        \]

        Consider first the right block.
        If \(n-x-1\ge N_{x+1}^+(\omega)\), then by definition of \(N_{x+1}^+(\omega)\),
        \[
            \cnum{\phi_{n-1}^\omega\circ\cdots\circ\phi_{x+1}^\omega}
            \le
            \mu(\omega)^{\,n-x-1}.
        \]
        If \(n-x-1< N_{x+1}^+(\omega)\), then trivially
        \[
            \cnum{\phi_{n-1}^\omega\circ\cdots\circ\phi_{x+1}^\omega}
            \le 1
            \le
            \mu(\omega)^{-N_{x+1}^+(\omega)}\,\mu(\omega)^{\,n-x-1}.
        \]
        Thus, in all cases,
        \begin{equation}
        \label{eq:right_block_bound}
            \cnum{\phi_{n-1}^\omega\circ\cdots\circ\phi_{x+1}^\omega}
            \le
            \mu(\omega)^{-N_{x+1}^+(\omega)}\,\mu(\omega)^{\,n-x-1}.
        \end{equation}
        
        Similarly, for the left block, if \(x-m-1\ge N_{x-1}^-(\omega)\), then
        \[
            \cnum{\phi_{x-1}^\omega\circ\cdots\circ\phi_{m+1}^\omega}
            \le
            \mu(\omega)^{\,x-m-1},
        \]
        while if \(x-m-1< N_{x-1}^-(\omega)\), then
        \[
            \cnum{\phi_{x-1}^\omega\circ\cdots\circ\phi_{m+1}^\omega}
            \le 1
            \le
            \mu(\omega)^{-N_{x-1}^-(\omega)}\,\mu(\omega)^{\,x-m-1}.
        \]
        Hence,
        \begin{equation}
        \label{eq:left_block_bound}
            \cnum{\phi_{x-1}^\omega\circ\cdots\circ\phi_{m+1}^\omega}
            \le
            \mu(\omega)^{-N_{x-1}^-(\omega)}\,\mu(\omega)^{\,x-m-1}.
        \end{equation}
        
        Combining \eqref{eq:start_decay_random_proof}, \eqref{eq:right_block_bound}, and \eqref{eq:left_block_bound}, we obtain
        \[
            f^\omega(n,m)
            \le
            8D\,\norm{\mbO_n}_\infty\,\norm{\mbO_m}_\infty\,
            \mu(\omega)^{-N_{x+1}^+(\omega)-N_{x-1}^-(\omega)}\,
            \mu(\omega)^{\,n-m-2}.
        \]
        Set
        \[
            T_x(\omega)
            :=
            \max\{N_{x+1}^+(\omega),\,N_{x-1}^-(\omega)\}.
        \]
        Then
        \[
            f^\omega(n,m)
            \le
            8D\,\norm{\mbO_n}_\infty\,\norm{\mbO_m}_\infty\,
            \mu(\omega)^{-2T_x(\omega)}\,
            \mu(\omega)^{\,n-m-2}.
        \]
        Since \(\mu(\omega)^{\,n-m-2}=\mu(\omega)^{-2}\mu(\omega)^{\,n-m}\), this becomes
        \[
            f^\omega(n,m)
            \le
            g_x(\omega)\,\norm{\mbO_n}_\infty\,\norm{\mbO_m}_\infty\,
            e^{-\alpha(\omega)(n-m)},
        \]
        where
        \[
            \alpha(\omega):=-\ln \mu(\omega)>0,
            \qquad
            g_x(\omega):=8D\,\mu(\omega)^{-2T_x(\omega)-2}.
        \]
        This proves \eqref{eq:exp-decay}.

        It remains to verify the two additional assertions.
        If \((\mcA_n)_{n\in\mbZ}\) is stationary ergodic, then the shift \(\theta\) is ergodic, and since \(\xi\) is \(\theta\)-invariant, it follows that \(\xi\) is almost surely constant. 
        Hence \(\mu\) and therefore \(\alpha\) may be chosen deterministic. This applies in particular in the i.i.d.\ case.
        If \((\mcA_n)_{n\in\mbZ}\) is random homogeneous, then for each fixed \(\omega\) we have
        \[
            \phi_n^\omega=\phi_0^\omega
            \qquad\text{for all }n\in\mbZ.
        \]
        Therefore the quantities
        \[
            \cnum{\phi_{x+k-1}^\omega\circ\cdots\circ\phi_x^\omega},
            \qquad
            \cnum{\phi_x^\omega\circ\cdots\circ\phi_{x-k+1}^\omega}
        \]
        do not depend on \(x\), and hence \(N_x^+(\omega)\), \(N_x^-(\omega)\), and \(T_x(\omega)\) may all be chosen independent of \(x\). 
        Consequently, the prefactor \(g_x\) may be replaced by a single almost surely finite random variable \(g\).
        This completes the proof.
    \end{proof}


\subsection{Random TI-MPS Case -- \texorpdfstring{\Cref{thm:TIcase}}{Theorem C}} 


    Before we state and prove the \Cref{thm:TIcase}, we define 
    \[
        \tau(\omega) := \inf \{n : \phi_{n-1}^\omega\circ\ldots\circ \phi_0^\omega \text{ is strictly positive }\},
    \]
    with the convention $\inf \, \emptyset = +\infty$, and set
    \[
    f(n) = \pr\{\tau > n\}.
    \]
    Since we are in the random TI case, after removing a null set we may simply write
    \[
        \phi_n^\omega=\phi^\omega
        \qquad\text{for every }n\in\mbZ.
    \]
    Under the \Cref{assumption2} we have that $\tau$ is almost surely finite and that $f(n) \to 0$ as $n\to \infty$.
    Furthermore for given $m,n \in \mbN$ we define 
    \begin{equation}
        \zeta_{n}(m) = \sum_{t=1}^n \pr\left(\{\tau=t\} \cap \{\cnum{\phi_{t-1}\circ\ldots\circ\phi_0} > 1-1/m\}\right)
    \end{equation}
    note that for a fixed $n\in\mbN$ $\zeta_n(m)$ is non-increasing in $m$ and $\lim_{m\to\infty} \zeta_{n}(m) = 0$. 
    
    With these definitions stated, we are now ready to prove \Cref{thm:TIcase}

    \TICase*

        \begin{proof}
            Fix $\epsilon\in(0,1)$. 
            Choose $b\in\mathbb N$ and $u\in\mathbb N$, $u\ge2$, such that
            \[
                f(b)\ \le\ \frac{\epsilon}{2},
                    \qquad
                \zeta_b(u)\ \le\ \frac{\epsilon}{2}.
            \]
        Set $\lambda_1:=(1-\frac{1}{u})^{1/b}\in(0,1)$ and $K(\epsilon):=8D\,\lambda_1^{-(b+1)}$.
        
        Let $m,n\in\mathbb Z$ with $2\le |n-m|$ and, w.l.o.g., $m<n$. Define the event
        \[
          E\ :=\ \{\tau\le b\}\ \cap\ \big\{\cnum{(\phi^\omega)^b}\le \lambda_1^b\big\}.
        \]
        We claim $\Pr(E)\ge 1-\epsilon$. 
        Indeed,
        \begin{align*}
            E^c 
            &=
            \{
                \tau > b
            \} 
            \cup 
            \left(
            \{\tau \le b\} 
            \cap 
            \{\cnum{(\phi^\omega)^b}  > \lambda_1^b \}
            \right)\,.
        \end{align*}
        Thus 
        \[
            \Pr(E^c)
            \ \le\ f(b) + \sum_{t=1}^b \Pr\Big(\{\tau=t\}\cap \big\{\cnum{(\phi^\omega)^b}>\lambda_1^b\big\}\Big).
        \]
        But for $t\in \{1, \ldots , b\}$ we have that, by submultiplicativity of $\cnum{\,\cdot\,}$, $\cnum{(\phi^\omega)^t} \ge \cnum{(\phi^\omega)^b} > 1-1/u$. 
        Therefore 
        \[
            \Pr(E^c)\ \le\ f(b)+\sum_{t=1}^b \Pr\big(\{\tau=t\}\cap\{\cnum{(\phi^\omega)^t}>1-\tfrac{1}{u}\}\big)
            \ =\ f(b)+\zeta_b(u)\ \le\ \epsilon.
        \]

        Now for $2 \le (n-m)$ we have that there is some $k\in\mbN_0$ so that $0 \le kb \le n-m-2< (k+1)b$.
        If $k= 0$ we trivially have from \Cref{lemma:main_bound} that, with probability $1$,
        \begin{align*}
           f^{\omega}(n,m) 
           &\leq 8 D \norm{\mbO_n}_\infty\norm{\mbO_m}_\infty\,,\\
           &\leq 8 D \norm{\mbO_n}_\infty\norm{\mbO_m}_\infty \left(1-\dfrac{1}{u}\right)^{\dfrac{n-m-2}{b} + \dfrac{1}{b} - 1}\,,\\
           &= 8 D \norm{\mbO_n}_\infty\norm{\mbO_m}_\infty \lambda_1^{-1 -b}\lambda_1^{n-m}\,.
        \end{align*}
        Now for $k\ge 1$, on the event $E$,
        \[
            \cnum{\phi^\omega_{n-1}\circ\cdots\circ\phi^\omega_{m+1}}
            \ =\ \cnum{(\phi^\omega)^{\,n-m-1}}
            \ \le\ \big(\cnum{(\phi^\omega)^b}\big)^{\lfloor (n-m-2)/b\rfloor}
            \ \le\ \lambda_1^{\,n-m}\,\lambda_1^{-(b+1)},
        \]
        where we have used that $k = \floor{(n-m-2)/b}$ and that $kb \ge (n-m) - (b+1)$.
        Thus by \Cref{lemma:main_bound},
        \[
            f^\omega(n,m)
            \ \le\ 8 D\,\norm{\mbO_n}_\infty\,\norm{\mbO_m}_\infty\;\cnum{\phi^\omega_{n-1}\circ\cdots\circ\phi^\omega_{m+1}}
            \ \le\ K(\epsilon)\,\norm{\mbO_n}_\infty\,\norm{\mbO_m}_\infty\;\lambda_1^{\,n-m}
        \]
        on the event $E$. Here $K(\epsilon) = 8 D\lambda_1^{-(1+b)}$.
        Considering all cases, we obtain ith probability at least $\pr(E) \ge 1 - \epsilon$, it must be the case that
        \[
            f^\omega(n,m)  \ \le\ K(\epsilon)\,\norm{\mbO_n}_\infty\,\norm{\mbO_m}_\infty\;\exp{-\lambda(\epsilon) |n-m|}
        \]
        where $\lambda(\epsilon) = -\ln{\lambda_1}> 0$.
        \end{proof}


\section{Stochastically Generated MPS with Decaying Stochastic Correlations}
\label{section:stochastic_corr}


    In this section we provide the proof of \Cref{thm:IID}, \Cref{thm:decay_cor}, \Cref{thm:fast_decay} and \Cref{thm:fast_beta_decay}. 
    First, we start with the extreme case where there is no stochastic correlation between sites. 
    

\subsection{Random MPS with IID Sampling}
\label{sect:IID}


    In the extreme case where the local tensor sequence $(\mcA_n)_{n\in\mbZ}$ is independent and identically distributed (IID), the spatial maximal correlation profile satisfies $\rho_n = 0$, for all $n\geq 1$ (and similar for all other mixing coefficients). 
    This strong stochastic de-correlation structure permits a significantly sharper decay bound for the stochastic expectation of the contraction coefficient of the associated  $n$-fold transfer operator $\Phi^{(n)}$. Specifically, the decay is exponential. 

    First, we need the following small result 

    \begin{prop}
    \label{prop:expectation_to_0}
    Under Assumption~\ref{assumption1}-\ref{assumption2}, $\lim_{n\to\infty}\,\mathbb{E}\bigl[\cnum{\Phi^{(n)}}\bigr]=0$.
    \end{prop}
        \begin{proof}
            By \Cref{lemma:boundary_states_quantitative}\emph{(C)}, applied  with \(x=0\), we get \(\cnum{\Phi^{(n)}}\to 0\) almost surely.
            Moreover,  $0\leq\cnum{\Phi^{(n)}}\leq1$, almost surely.
            The conclusion therefore follows from the dominated convergence theorem. 
        \end{proof}

    \begin{lemma}
    \label{lemma:C_decay_exp}
         Suppose the sequence $(\mcA_n)_{n\in\mbZ}$ of random local tensors is sampled in an IID fashion and that \Cref{assumption1} and \Cref{assumption2} hold. Then there exist constants $C_{\text{iid}}>0$, $\delta>0$ such that
         \begin{equation}
             \avg{\cnum{\Phi^{(n)}}} \leq C_{\text{iid}} e^{-\delta n}  \text{for all } n \in \mbN\, .
         \end{equation}
    \end{lemma}
        \begin{proof}
            From Assumptions~\ref{assumption1} and \ref{assumption2} we must have that there is some $N_0$ such that $\avg{\cnum{\Phi^{(n)}}} \leq 1/2$ for all $n\geq N_0$ (see \Cref{prop:expectation_to_0}). 
            Now fix $\delta = \ln(1/2)/(-2N_0)$ so that $\delta>0$. 
            Now for given $n\geq 2N_0$ there must exist some $i\in\mbN$ such that 
            \[n \in [2iN_0, 2(i+1)N_0)\ .\]
            Using submultiplicativity of the contraction coefficient, we obtain that
            \[
                \cnum{\Phi^{(n)}} \leq  \cnum{\Phi^{(2iN_0)}} \leq \prod_{k=1}^i \cnum{\phi_{2kN_0-1}\circ\ldots \circ \phi_{2(k-1)N_0}}\, ,
            \]
            where each block $\cnum{\phi_{2kN_0-1}\circ\ldots \circ \phi_{2(k-1)N_0}}$ is independent and identically distributed with the same distribution as $\Phi^{2N_0}$. Now using stationarity and independence of the sequence $(\phi_n)_n$, we obtain that
            \[
                 \avg{\cnum{\Phi^{(n)}}} \leq \avg{\cnum{\Phi^{(2N_0)}}}^i  \leq \avg{\cnum{\Phi^{(N_0)}}}^{2i} \leq \dfrac{1}{2^{2i}}= e^{-4iN_0\delta} \leq e^{-n\delta}
            \]
            where the last inequality uses that $n< 2(i+1)N_0 \le 4iN_0$. Since the bound above holds for all $n\geq 2N_0$ we must have that there is some $C_{\text{iid}}>0$ such that
            \[
                 \avg{\cnum{\Phi^{(n)}}} \leq C_{\text{iid}} e^{-n\delta} \quad \text{for all } n \in \mbN\, .
            \]
        \end{proof}

    We are now ready to prove \Cref{thm:IID}:
    
    \IID*
    \begin{proof}
        We may therefore assume without loss of generality that \(m<n\).
        First note from \Cref{lemma:main_bound} we have that 
        $f(n,m) \leq 8D \norm{\mbO_n}_\infty\norm{\mbO_m}_\infty \cnum{\phi_{n-1} \circ\ldots \phi_{m+1}}$, almost surely.
        Also from \Cref{lemma:C_decay_exp} we have the existence of $\delta>0$ and a constant $C_{\text{iid}}>0$ such that $\avg{\cnum{\Phi^{(k)}}} \leq C_{\text{iid}}e^{-\delta k}$ for all $k\in\mbN$. 
        We also have that
        \[
            \pr\{\omega: \cnum{\Phi^{(n-m-1)}} \leq C_{\text{iid}} e^{\delta}e^{-(1/2)\delta (n-m)}\}
            = 
            \pr\{\omega: \cnum{\phi_{n-1}\circ\ldots\circ\phi_{m+1}} \leq C_{\text{iid}} e^{\delta} e^{-(1/2)\delta (n-m)}\,\}
        \]
        by strict stationarity.
        However, from Markov's inequality, we have that 
        \[
            \pr\{\omega: \cnum{\Phi^{(n-m-1)}} \leq C_{\text{iid}} e^\delta e^{-(1/2)\delta (n-m)}\} 
            \geq 
            1 
            - 
            \dfrac{\avg{\cnum{\Phi^{(n-m-1)}}}}{C_{\text{iid}}e^{\delta} e^{-(1/2)\delta (n-m)}} \geq 1 - e^{-(1/2)\delta(n-m)}
        \]
        Therefore, we have that for $C_{\text{pr}} := 8DC_{\text{iid}}e^{\delta} $ and $\beta = \delta/2$
        \[
            \pr\{\omega:f^\omega(n,m) \leq C_{\text{pr}}\norm{\mbO_m}_\infty\norm{\mbO_n}_\infty e^{-\beta (n-m)}\} \geq  1 - e^{-\beta(n-m)}\, ,
        \]
        concluding the proof.
    \end{proof}
    

\subsection{Mixing Coefficients for Random Systems}
\label{subsec:mixing_coeffs}

    In a system described by random objects, there are various \emph{measures} of stochastic correlation/decorrelation given by so-called mixing coefficients. 
    Five such classical mixing coefficients are described in this section. 
    Before we discuss the mixing coefficients, we start with the stochastic correlation between two random variables. 
    Let $(\Omega,\mcF,\Pr)$ be a probability space. 
    For real–valued random variables $X,Y$ on $(\Omega,\mcF,\Pr)$, the \emph{maximal stochastic correlation} is
    \begin{equation}
    \label{eq:rho_1}
          \rho(X,Y)
          \ :=\
          \sup\bigl\{\,|\corr{f(X)}{g(Y)}|\ :\ f,g\in L^2\ \bigr\}.
    \end{equation}
    It is classical that $X\!\perp\!\!\!\perp Y$ (i.e. $X$ and $Y$ are independent)  if and only if $\rho(X,Y)=0$. 
    The notion extends to random elements with values in measurable spaces. Writing
    \[
        \mcA:=\sigma(X),\qquad \mcB:=\sigma(Y),
    \]
    where $\sigma(\,\cdot\,)$ denotes the generated $\sigma$-algebra.
    One has the equivalent formulation
    \begin{equation}
    \label{eq:rho_2}
        \rho(X,Y)=\rho(\mcA,\mcB)
        \ :=\
        \sup\bigl\{\,|\corr{U}{V}|\ :\ U\in L^2(\mcA),\ V\in L^2(\mcB) \, \bigr\}.
    \end{equation}

    \medskip

    Beyond $\rho$ there are standard mixing coefficients for pairs of sub–$\sigma$–algebras $\mcA,\mcB\subset\mcF$:
    \begin{align*}
          \psi(\mcA,\mcB)
                &:= \sup\Bigl\{\Bigl|1-\frac{\Pr(A\cap B)}{\Pr(A)\Pr(B)}\Bigr|:\ A\in\mcA,\ B\in\mcB,\ \Pr(A),\Pr(B)>0\Bigr\},\\
          \varphi(\mcA,\mcB)
                &:= \sup\bigl\{\,|\Pr(B\mid A)-\Pr(B)|:\ A\in\mcA,\ \Pr(A)>0,\ B\in\mcB\,\bigr\},\\
          \alpha(\mcA,\mcB)
                &:= \sup\bigl\{\,|\Pr(A\cap B)-\Pr(A)\Pr(B)|:\ A\in\mcA,\ B\in\mcB\,\bigr\},\\
          \beta(\mcA,\mcB)
                &:= \sup\biggl\{\frac12\sum_{i=1}^{I}\sum_{j=1}^{J}\bigl|\Pr(A_i\cap B_j)-\Pr(A_i)\Pr(B_j)\bigr|\biggr\},
    \end{align*}
    where the last supremum is over all finite partitions $\{A_i\}_{i=1}^I\subset\mcA$ and $\{B_j\}_{j=1}^J\subset\mcB$.
    These satisfy the (one–sided) hierarchy (see \cite{bradley2005basic,bradley2007introduction})
    \begin{equation}
    \label{mixing_hierarchy}
        \rho(\mcA,\mcB)\le \psi(\mcA,\mcB),
        \qquad
        \rho(\mcA,\mcB)\le 2\varphi(\mcA,\mcB)^{1/2},
        \qquad
        \alpha(\mcA,\mcB)\le \tfrac12\,\beta(\mcA,\mcB).
    \end{equation}

    Now these mixing coefficients generalize further to a random system described by a sequence of random objects $(\mcA_n)_{n\in\mbZ}$ (for our case, the random local tensors) as follows:
    Define the forward and backward filtrations
    \begin{equation}
    \label{eq:sigma_algebras_forward_backward}
        \mcF_k:=\sigma(\mcA_n: n\le k),
        \qquad
        \mcF^{k}:=\sigma(\mcA_n: n\ge k).
    \end{equation}
    The \emph{maximal spatial $\rho$–profile} is
    \begin{equation}
    \label{eq:rho_spatial}
        \rho_n\ :=\ \sup_{k\in\mbZ}\ \rho(\mcF_k,\mcF^{k+n})\qquad(n\in\mbN).
    \end{equation}
    We say the sequence is \emph{$\rho$–mixing} if $\rho_n\to0$ as $n\to\infty$ (cf.~\cite{kolmogorov1960strong}). 
    Analogously define spatial profiles $\psi_n,\varphi_n,\alpha_n,\beta_n$ by replacing $\rho$ in \eqref{eq:rho_spatial} with the corresponding coefficient.

    \begin{remark}
    \label{rem:stationary_profile}
    If $(\mcA_n)_{n\in\mbZ}$ is strictly stationary (our standing assumption), then $\rho(\mcF_k,\mcF^{k+n})$ does not depend on $k$, hence
    \(
        \rho_n=\rho(\mcF_0,\mcF^{n})\ \ \text{and similarly for }\psi_n,\varphi_n,\alpha_n,\beta_n.
    \)
    Thus, the ``$\sup_k$'' in \eqref{eq:rho_spatial} can be dropped without changing the value.
    \end{remark}

    With these definitions in place, we are now ready to prove the remaining results. 
    We start with the following practical lemma.

    \begin{lemma}[{\cite[\S1.2 Thm.~3]{doukhan2012mixing}}, {\cite[Thm.~3.8]{bradley2007introduction}}]
    \label{lemma:covariance_bounds}
        If $\mcA,\mcB\subset\mcF$ are sub–$\sigma$–algebras, $U\in L^2(\mcA)$ and $V\in L^2(\mcB)$, then
        \[
            |\cov{U}{V}|\ \le\ \rho(\mcA,\mcB)\,\|U\|_{L^2}\,\|V\|_{L^2}\,.
        \]
    \end{lemma}

    The following lemma shall be helpful in the subsequent analysis. 

    \begin{lemma}
    \label{lemma:covariance_recursive_bound}
        Let $(\mcA_n)_{n\in\mbZ}$ be a strictly stationary sequence of local tensors satisfying \Cref{assumption1}. 
        Then we have that 
        \begin{equation}
        \label{eq:covariance_recursive_bound}
              \avg{\cnum{\Phi^{(i+q+r)}}}
              \ \le\
              \avg{\cnum{\Phi^{(i)}}}\,\avg{\cnum{\Phi^{(r)}}}
              \ +\
              \rho_q\,\sqrt{\avg{\cnum{\Phi^{(i)}}}\,\avg{\cnum{\Phi^{(r)}}}}
              \qquad\text{for all }i,q,r\in\mbN.
        \end{equation}
    \end{lemma}
        \begin{proof}
            Fix $i,q,r\in\mbN$. By submultiplicativity and the bound $\cnum{\,\cdot\,}\le 1$,
            \[
                \cnum{\Phi^{(i+q+r)}}
                \ \le\
                \cnum{\Phi^{(i)}}\,\cnum{\phi_{i+q+r-1}\circ\cdots\circ\phi_{i+q}},
            \]
            hence
            \[
                \avg{\cnum{\Phi^{(i+q+r)}}}
                \ \le\
                \avg{\,\cnum{\Phi^{(i)}}\,\cnum{\phi_{i+q+r-1}\circ\cdots\circ\phi_{i+q}}\,}.
            \]
            Here $\cnum{\Phi^{(i)}}$ is $\mcF_{i-1}$–measurable, and hence $\mcF_i$–measurable, while $\cnum{\phi_{i+q+r-1}\circ\cdots\circ\phi_{i+q}}$ is $\mcF^{\,i+q}$–measurable. 
            Applying \Cref{lemma:covariance_bounds} with $\mcA=\mcF_i$ and $\mcB=\mcF^{\,i+q}$, and with $U=\cnum{\Phi^{(i)}}$ and $V=\cnum{\phi_{i+q+r-1}\circ\cdots\circ\phi_{i+q}}$, and using $(\cnum{\,\cdot\,})^2\le \cnum{\,\cdot\,}$, we obtain
            \begin{multline}
            \label{eq:covariance_recursive_bound_precise}
                \avg{\cnum{\Phi^{(i+q+r)}}}\\
                \ \le\
                \avg{\cnum{\Phi^{(i)}}}\,\avg{\cnum{\phi_{i+q+r-1}\circ\cdots\circ\phi_{i+q}}}
                \ +\
                \rho_q\,\sqrt{\avg{\cnum{\Phi^{(i)}}}\,\avg{\cnum{\phi_{i+q+r-1}\circ\cdots\circ\phi_{i+q}}}}.
            \end{multline}
            By \emph{strict stationarity} of $(\mcA_n)_{n\in\mbZ}$, the distribution of any length–$r$ block is translation invariant, so
            \[
                \avg{\cnum{\phi_{i+q+r-1}\circ\cdots\circ\phi_{i+q}}}\ =\ \avg{\cnum{\Phi^{(r)}}}.
            \]
            Thus \eqref{eq:covariance_recursive_bound_precise} reduces to
            \begin{equation}
                  \avg{\cnum{\Phi^{(i+q+r)}}}
                  \ \le\
                  \avg{\cnum{\Phi^{(i)}}}\,\avg{\cnum{\Phi^{(r)}}}
                  \ +\
                  \rho_q\,\sqrt{\avg{\cnum{\Phi^{(i)}}}\,\avg{\cnum{\Phi^{(r)}}}},
            \end{equation}
            proving \eqref{eq:covariance_recursive_bound}. 
        \end{proof}


\subsection{Random MPS with Vanishing Spatial Stochastic Correlations}
\label{section:rho_to_0}


    In scenarios where the maximal spatial stochastic correlation $\rho_n$ (as defined in \eqref{eq:rho_spatial}) decays to zero as $n\to\infty$, we obtain quantitative decay estimates on the stochastic expectation of the contraction coefficient $\cnum{\Phi^{(n)}}$. 
    Specifically, the decay can be shown to be faster than any prescribed polynomial rate:

    \begin{lemma}
    \label{lemma:C_decay_in_poly}
          Let $(\mcA_n)_{n\in\mbZ}$ be a \emph{strictly stationary} bi--infinite sequence of random local tensors satisfying Assumptions~\ref{assumption1}-\ref{assumption2}, and suppose the maximal spatial stochastic correlation $\rho_n\to 0$ as $n\to\infty$. 
          Then for each $p\in\mbN$ there exists a constant $C_p>0$ such that 
          \begin{equation}
          \label{eq:poly_decay_expectation}
                \avg{\cnum{\Phi^{(n)}}}\ \le\ \frac{C_p}{n^p}
            \qquad\text{for all }n\in\mbN\, .
          \end{equation}
    \end{lemma}

        \begin{proof}
            Since $\rho_n\to0$ and, by \Cref{prop:expectation_to_0}, $\avg{\cnum{\Phi^{(n)}}}\to0$, for any given $p\in\mbN$ there exists $N_0\in\mbN$ such that
            \begin{equation}
            \label{eq:choice_N0}
                \rho_{N_0}+\avg{\cnum{\Phi^{(N_0)}}}\ \le\ 4^{-p}.
            \end{equation}
            Apply \eqref{eq:covariance_recursive_bound} with $i=q=r=N_0$ to get
            \begin{equation}
            \label{eq:first_bootstrap}
                \avg{\cnum{\Phi^{(3N_0)}}}
                \ \le\
                \avg{\cnum{\Phi^{(N_0)}}}^2+\rho_{N_0}\,\avg{\cnum{\Phi^{(N_0)}}}
                \ \le\
                4^{-p}\,\avg{\cnum{\Phi^{(N_0)}}}.
            \end{equation}
            Next, setting $i=r=3N_0$ and $q=N_0$,
            \[
                  \avg{\cnum{\Phi^{(7N_0)}}}
                  \ \le\
                  \avg{\cnum{\Phi^{(3N_0)}}}^2+\rho_{N_0}\,\avg{\cnum{\Phi^{(3N_0)}}}.
            \]
            Using \eqref{eq:first_bootstrap} and \eqref{eq:choice_N0}, together with $\avg{\cnum{\Phi^{(N_0)}}}\ge \avg{\cnum{\Phi^{(3N_0)}}}$, yields
            \begin{equation}
            \label{eq:second_bootstrap}
                  \avg{\cnum{\Phi^{(7N_0)}}}
                  \ \le\
                  4^{-2p}\,\avg{\cnum{\Phi^{(N_0)}}}.
            \end{equation}

            Define \(M_k:=(2^k-1)N_0\) for \(k\geq1\). 
            Since \(M_{k+1}=2M_k+N_0\), \eqref{eq:covariance_recursive_bound}  with \(i=r=M_k\) and \(q=N_0\) gives
            \begin{align*}
                \avg{\cnum{\Phi^{(M_{k+1})}}}
                &\leq
                \left(
                    \avg{\cnum{\Phi^{(M_k)}}}
                    +
                    \rho_{N_0}
                \right)
                \avg{\cnum{\Phi^{(M_k)}}}
                \\
                &\leq
                \left(
                    \avg{\cnum{\Phi^{(N_0)}}}
                    +
                    \rho_{N_0}
                \right)
                \avg{\cnum{\Phi^{(M_k)}}}
                \\
                &\leq
                4^{-p}\,
                \avg{\cnum{\Phi^{(M_k)}}}.
            \end{align*}
            Here the second inequality follows from \(M_k\geq N_0\), submultiplicativity, and \(\cnum{\,\cdot\,}\leq1\). 
            Iterating this  estimate gives
            \begin{equation}
            \label{eq:dyadic_sequence}
                \avg{\cnum{\Phi^{(M_k)}}}
                \leq
                4^{-(k-1)p}\,
                \avg{\cnum{\Phi^{(N_0)}}}
                \qquad(k\geq1).
            \end{equation}
        
            Now let $n\ge N_0$ and choose $k$ with $n\in[M_k,M_{k+1})$. By submultiplicativity,
            $\avg{\cnum{\Phi^{(n)}}}\le \avg{\cnum{\Phi^{(M_k)}}}$. Since $M_{k+1}=(2^{k+1}-1)N_0\le 2^{k+1}N_0\le 4^{k-1}N_0$ for $k\ge 3$, we have $4^{k-1}\ge n/N_0$ (for $n$ large enough). Combining this with \eqref{eq:dyadic_sequence} gives
            \[
                  \avg{\cnum{\Phi^{(n)}}}
                  \ \le\
                  \avg{\cnum{\Phi^{(N_0)}}}\,\frac{N_0^{p}}{n^{p}}
                  \qquad\text{for all large }n.
            \]
            Absorbing the finite initial segment into the constant completes the proof of \eqref{eq:poly_decay_expectation}.
            \end{proof}
    
    Using \Cref{lemma:C_decay_in_poly}, we obtain the polynomial high–probability clustering bound in \Cref{thm:decay_cor}.

    \DECAY*
        \begin{proof}
            The proof follows the argument of \Cref{thm:IID}, replacing the exponential expectation bound with the polynomial one from \Cref{lemma:C_decay_in_poly}. Fix $k\in\mbN$ and apply Markov's inequality to $\cnum{\Phi^{(|n-m|-1)}}$ with the threshold $t=2^{2k}C_{2k}/|n-m|^{k}$; using \Cref{lemma:C_decay_in_poly} at exponent $2k$ and the observation $|n-m|-1\ge |n-m|/2$ for $|n-m|\ge 2$, we get
            \[
                  \pr\!\left\{\cnum{\Phi^{(|n-m|-1)}}\le \frac{2^{2k}C_{2k}}{|n-m|^{k}}\right\}
                  \ \ge\
                  1-\frac{C_{2k}/(|n-m|-1)^{2k}}{(2^{2k}C_{2k})/|n-m|^k}
                  \ \ge\
                  1-\frac{1}{|n-m|^{k}}.
            \]
            By \Cref{lemma:main_bound},
                \[
                  f^\omega(n,m)\ \le\ 8\,D\,\norm{\mbO_n}_\infty\,\norm{\mbO_m}_\infty\,
                  \cnum{\phi^\omega_{n-1}\circ\cdots\circ\phi^\omega_{m+1}}
            \]
            Therefore for $C_{\mathrm{poly}}(k):=8\,D\,2^{2k}C_{2k}$ we have 
            \[
                \pr\left\{f^\omega(n,m)\ \le\ \frac{C_{\mathrm{poly}}(k)\,\norm{\mbO_n}_\infty\,\norm{\mbO_m}_\infty}{|n-m|^{k}}\right\}
                    \, \ge \,
                1 - \dfrac{1}{|n-m|^k},
            \]
            where we have used stationarity to obtain that 
            \[
                \pr\!\left\{\cnum{\Phi^{(|n-m|-1)}}\le \frac{2^{2k}C_{2k}}{|n-m|^{k}}\right\}
                =
                \pr\!\left\{\cnum{\phi_{n-1}\circ\ldots\circ\phi_{m+1}}\le \frac{2^{2k}C_{2k}}{|n-m|^{k}}\right\}.
            \]
        \end{proof}


\subsection{Sufficiently Fast \texorpdfstring{$\rho$}{rho}-mixing}
\label{section:rho_fast}


    In the previous section, we worked under the minimal decorrelation hypothesis $\rho_n\to0$, which yielded (super) polynomial decay estimates for $\avg{\cnum{\Phi^{(n)}}}$. 
    We now strengthen the assumption to a stretched–exponential or exponential spatial $\rho$–-mixing profile and show how this accelerates the contraction bootstrapping accordingly. 
    Throughout this subsection we assume that there exist constants $C_\rho,c_\rho>0$ and $\gamma\in(0,1]$ such that
    \begin{equation}
    \label{ineq:rho_mixing_fast}
        \rho_n\ \le\ C_\rho\,\exp\{-c_\rho\,n^\gamma\}
        \qquad(n\in\mbN).
    \end{equation}
    
    We present the following elementary lemma.

    \begin{lemma}
    \label{lemma:fast_mixing_1}
        Assume \eqref{ineq:rho_mixing_fast}. Then for every $\alpha\in(0,\gamma)$ and every $\theta\in(0,1)$ there exists $L_0\in\mbN$ such that for all $L\ge L_0$,
        \begin{equation}
        \label{ine:rho_fast_1}
              C_\rho\,\exp{-\tfrac{c_\rho}{2}\,(\theta L)^{\gamma}}
              \ \le\
              \tfrac12\,\exp{-2L^\alpha}\,.
        \end{equation}
    \end{lemma}

        \begin{proof}
            Inequality \eqref{ine:rho_fast_1} is equivalent to
            \[
                \frac{c_\rho}{2}\,\theta^\gamma\,L^{\gamma} - 2L^\alpha\ \ge\ \ln(2C_\rho).
            \]
            A sufficient condition is to enforce simultaneously
            \begin{equation}
            \label{eq:two_conditions}
                \frac{c_\rho}{2}\,\theta^\gamma\,L^{\gamma}\ \ge\ 4L^\alpha,
              \qquad
              2L^\alpha\ \ge\ \bigl(\ln(2C_\rho)\bigr)^+,
            \end{equation}
            where $x^+:=\max\{x,0\}$. Since $\alpha\in(0,\gamma)$, the first inequality in \eqref{eq:two_conditions} is implied by
            \[
                  L\ \ge\ \left(\frac{4}{\tfrac{c_\rho}{2}\,\theta^\gamma}\right)^{\!1/(\gamma-\alpha)}
                  \ =\ \left(\frac{8}{c_\rho\,\theta^\gamma}\right)^{\!1/(\gamma-\alpha)}.
            \]
            The second is implied by
            \[
                L\ \ge\ \left(\frac{(\ln(2C_\rho))^+}{2}\right)^{\!1/\alpha}.
            \]
            Therefore the choice
            \[
                  L_0\ :=\ \Biggl\lceil\,
                  \max\!\left\{
                        \left(\frac{8}{c_\rho\,\theta^\gamma}\right)^{\!1/(\gamma-\alpha)},
                        \ \left(\frac{(\ln(2C_\rho))^+}{2}\right)^{\!1/\alpha},
                        \ 1
                  \right\}
                  \,\Biggr\rceil
            \]
            ensures \eqref{eq:two_conditions} for all $L\ge L_0$, which in turn yields \eqref{ine:rho_fast_1}.
        \end{proof}

    Next, we present a simple yet useful comparison of stretched–exponential rates.

    \begin{lemma}
    \label{lemma:fast_mixing_2}
        For every $\alpha\in(0,1)$ there exists $\theta\in(0,1)$ and $L_0\in\mbN$ such that
        \begin{equation}
        \label{ine:rho_fast_2}
              \frac{3}{2}\,\exp\{-2L^\alpha\}\ \le\ \exp\{-((2+\theta)L)^\alpha\}
              \qquad\text{for all }L\ge L_0.
        \end{equation}
    \end{lemma}
    
        \begin{proof}
            Fix $\alpha\in(0,1)$ and choose $\theta\in(0,1)$ so small that $(2+\theta)^\alpha<2$, e.g.
            \[
                \theta\ :=\ \min\!\left\{\frac{2^{1/\alpha}-2}{2},\,\frac12\right\}\in(0,1).
            \]
            Since the map $x\mapsto e^{-x}$ is strictly decreasing on $(0,\infty)$, \eqref{ine:rho_fast_2} is equivalent to
            \[
              2L^\alpha - ((2+\theta)L)^\alpha\ \ge\ \ln(3/2).
            \]
            By the choice of $\theta$ the coefficient $2-(2+\theta)^\alpha$ is strictly positive. Therefore, it suffices to take
            \[
              L_0\ \ge\ \left(\frac{\ln(3/2)}{\,2-(2+\theta)^\alpha\,}\right)^{\!1/\alpha},
            \]
            which yields the desired bound for all $L\ge L_0$.
        \end{proof}

    Equipped with the lemmas above, we are now ready to prove the following result.

    \begin{lemma}
    \label{lemma:main_fast_decay}
        Let $(\mcA_n)_{n\in\mbZ}$ be a strictly stationary sequence of random local tensors satisfying \Cref{assumption1,assumption2}.
        Suppose there are constants $C_\rho,c_\rho>0$ and $\gamma\in(0,1]$ such that
        \[
            \rho_n\ \le\ C_\rho\,\exp\{-c_\rho n^\gamma\}\qquad(n\in\mbN).
        \]
        Then for every $\alpha\in(0,\gamma)$ there exist $C_\alpha,c_\alpha>0$ such that
        \[
            \avg{\cnum{\Phi^{(n)}}}\ \le\ C_\alpha\,\exp\{-c_\alpha n^\alpha\}\qquad\forall\,n\in\mbN.
        \]
    \end{lemma}

        \begin{proof}
            Fix $\alpha\in(0,\gamma)$. Choose $\theta\in(0,1)$ as in \Cref{lemma:fast_mixing_2}, so that $(2+\theta)^\alpha<2$. Increase $L_0\in\mbN$ so that simultaneously:
            \begin{enumerate}
                \item \Cref{lemma:fast_mixing_1} holds for all $L\ge L_0$ (i.e.\ $C_\rho e^{-\tfrac{c_\rho}{2}(\theta L)^\gamma}\le \tfrac12 e^{-2L^\alpha}$),
                \item \Cref{lemma:fast_mixing_2} holds for all $L\ge L_0$ (i.e.\ $\tfrac32 e^{-2L^\alpha}\le e^{-((2+\theta)L)^\alpha}$),
                \item $\avg{\cnum{\Phi^{(L_0)}}}\le\tfrac12$,
                \item $\theta L_0\ge 2$.
            \end{enumerate}
            Define recursively $L_{k+1}:=2L_k+\lfloor\theta L_k\rfloor$ for $k\ge0$, and abbreviate $c(n):=\avg{\cnum{\Phi^{(n)}}}$.

            We start by establishing the following claim. 
            
            \medskip\noindent\emph{Claim 1.} There exists $K_0$ such that $c(L_{K_0})\le e^{-L_{K_0}^\alpha}$.
        
            \smallskip
                \begin{proof}[Proof of Claim 1:]
                    Assume for contradiction that $c(L_k) > e^{-L_k^\alpha}$ for all $k\ge 0$. 
                    From \Cref{lemma:covariance_recursive_bound} (with $i=r=L_k$ and $q  =q_k = \floor{{\theta L_k}}$)  we have that
                    \[
                        c(L_{k+1})\ \le\ c(L_k)^2+\rho_q\,c(L_k).
                    \]
                    Since $\theta L_k\ge2$, we have $\lfloor\theta L_k\rfloor\ge\tfrac12\theta L_k$, whence
                    \[
                        \rho_q\ \le\ C_\rho\,\exp{-c_\rho\bigl(\dfrac{\theta L_k}{2}\bigr)^\gamma}
                        \ \le\ 
                        C_\rho\,\exp{-\dfrac{c_\rho}{2}\,\theta^\gamma L_k^\gamma}
                        \ \le\ 
                        \dfrac12\,e^{-2L_k^\alpha},
                    \]
                    by \Cref{lemma:fast_mixing_1} and $2^{-\gamma}\ge\tfrac12$. Using $c(L_k)\le1$ and $c(L_k) > e^{-L_k^\alpha}$,
                    \[
                        c(L_{k+1})\ \le\ c(L_k)^2 + \dfrac12 e^{-2L_k^\alpha}\ \le\ \dfrac32\,c(L_k)^2.
                    \]
                    Iterating this with $c(L_0)\le\tfrac12$ gives
                    \begin{equation}
                    \label{ine:contradictory_rate_1}
                        c(L_k)\ \le\ \left(\dfrac32\right)^{2^k-1} c(L_0)^{2^k}
                        \ \le\ \dfrac{2}{3}\,\left(\dfrac{3}{4}\right)^{2^k}
                        \ =\ \frac{2}{3}\,\exp{-\ln(4/3)\,2^k}.
                    \end{equation}
                    On the other hand, $L_{k+1}\le(2+\theta)L_k$ implies $L_k\le L_0(2+\theta)^k$, and hence, by the standing assumption $c(L_k)>e^{-L_k^\alpha}$,
                    \begin{equation}
                    \label{ine:contradictory_rate_2}
                        c(L_k)\ >\ \exp\{-L_k^\alpha\}\ \ge\ \exp\{-L_0^\alpha(2+\theta)^{\alpha k}\}.
                    \end{equation}
                    Since $(2+\theta)^\alpha<2$, the two bounds in \eqref{ine:contradictory_rate_1} and \eqref{ine:contradictory_rate_2} contradict each other for large $k$. This proves Claim~1.
                \end{proof}

            Having obtained $K_0$ so that $c(L_{K_0})\le e^{-L_{K_0}^\alpha}$ we claim the following: 

            \medskip\noindent\emph{Claim 2.} For all $k\ge K_0$ one has $c(L_k)\le e^{-L_k^\alpha}$.

            \smallskip
            \begin{proof}[Proof of Claim 2.]
                Assume $c(L_k)\le e^{-L_k^\alpha}$ for some $k\ge K_0$. Then, as above,
                \[
                    c(L_{k+1})\ \le\ c(L_k)^2 + C_\rho \exp{-\dfrac{c_\rho}{2}\theta^\gamma L_k^\gamma}
                    \ \le\ 
                    \exp{-2L_k^\alpha} + \dfrac12 \exp{-2L_k^\alpha}
                    \ =\ 
                    \dfrac32\,\exp{-2L_k^\alpha}.
                \]
                By \Cref{lemma:fast_mixing_2}, $\dfrac32 e^{-2L_k^\alpha}\le \exp{-((2+\theta)L_k)^\alpha}$. Since $L_{k+1}=2L_k+\floor{\theta L_k}\le(2+\theta)L_k$ and $x\mapsto e^{-x}$ decreases,
                \[
                  c(L_{k+1})\ \le\ \exp{-((2+\theta)L_k)^\alpha}\ \le\ \exp{-(L_{k+1})^\alpha}.
                \]
                Thus, the property propagates by induction to all $k\ge K_0$.
            \end{proof}
            
            \medskip

            Finally, let $n>2L_{K_0}$. Since $L_{k+1}\ge2L_k$, there exists $k\ge K_0$ such that $2L_k<n\le 2L_{k+1}$. 
            Writing $n=2L_k+q$ with $q>0$, \eqref{eq:covariance_recursive_bound} yields
                \[
                  c(n)\ \le\ c(L_k)^2 + \rho_q\,c(L_k)
                  \ \le\ 2\,c(L_k)
                  \ \le\ 2\,\exp{-L_k^\alpha}.
                \]
                Since $n\le 2(2+\theta)L_k$, we have $L_k\ge n/(2(2+\theta))$, hence
                \[
                  c(n)\ \le\ 2\,\exp\!\left\{-\bigl(2(2+\theta)\bigr)^{-\alpha}\,n^\alpha\right\}.
                \]
                Absorbing the finite initial segment into $C_\alpha$ and setting $c_\alpha:=(2(2+\theta))^{-\alpha}$ gives the desired bound.
        \end{proof}

    The proof of \Cref{thm:fast_decay} is now a small exercise using \Cref{lemma:main_fast_decay}.

    \fastdecay*
        \begin{proof}
            Fix $\alpha\in(0,\gamma)$. By \Cref{lemma:main_fast_decay} there are $C_\alpha,c_\alpha>0$ with
            \begin{equation}
            \label{eq:exp_expect}
                \avg{\cnum{\Phi^{(k)}}}\ \le\ C_\alpha\,e^{-c_\alpha k^\alpha}\qquad\forall\,k\in\mbN.
            \end{equation}
            Without loss of generality, assume $m<n$ and set $L:=n-m-1\ge 1$. By \emph{strict stationarity}, the law of the length-$L$ block is shift invariant, hence
            \begin{equation}
            \label{eq:law_equal_blocks}
                \cnum{\phi_{n-1}\circ\cdots\circ\phi_{m+1}}\ \stackrel{d}{=}\ \cnum{\Phi^{(L)}}.
            \end{equation}
            Applying Markov’s inequality to the nonnegative random variable $\cnum{\Phi^{(L)}}$ and using \eqref{eq:exp_expect}, we obtain
            \[
                \pr \left\{\cnum{\Phi^{(L)}} \le C_\alpha \exp{-\dfrac{c_\alpha}{2}L^\alpha} \right\}
                \ \ge\ 1-\dfrac{C_\alpha \exp{-c_\alpha L^\alpha}}{C_\alpha \exp{-\dfrac{c_\alpha}{2}L^\alpha}}
                \ =\ 1-\exp{-\dfrac{c_\alpha}{2}L^\alpha}.
            \]
            Invoking \eqref{eq:law_equal_blocks} and the fact that $x\mapsto e^{-x}$ is \emph{decreasing} on $(0,\infty)$, and using $|n-m|\ge 2\Rightarrow L=n-m-1\ge \tfrac12|n-m|$, we obtain
            \begin{multline}
            \label{eq:prob_cnum_block}
                \pr\!\left\{\cnum{\phi_{n-1}\circ\cdots\circ\phi_{m+1}}\ \le\ C_\alpha\,\exp{-\dfrac{c_\alpha}{2^{\alpha+1}}|n-m|^\alpha}\right\} \\
                \ \ge\ 1-\exp{-\dfrac{c_\alpha}{2}L^\alpha}
                \ \ge\ 1-\exp{-\dfrac{c_\alpha}{2^{\alpha+1}}|n-m|^\alpha}.
            \end{multline}

            By \Cref{lemma:main_bound}, for $\pr$-almost every $\omega$,
            \[
                f^\omega(n,m)
                \le
                8D\,\norm{\mbO_n}_\infty\,\norm{\mbO_m}_\infty\,
                \cnum{
                    \phi_{n-1}^\omega\circ\cdots\circ\phi_{m+1}^\omega
                }.
            \]
            Intersecting the event in \eqref{eq:prob_cnum_block} with this full-measure event does not change its probability. 
            Therefore, setting
            \[
                K_\alpha:=8D\,C_\alpha,
                \qquad
                k_\alpha:=\frac{c_\alpha}{2^{\alpha+1}},
            \]
            we obtain
            \[
                \pr\!\left\{
                    \omega:
                    f^\omega(n,m)
                    \le
                    K_\alpha\,
                    \norm{\mbO_n}_\infty\,
                    \norm{\mbO_m}_\infty\,
                    e^{-k_\alpha|n-m|^\alpha}
                \right\}
                \ge
                1-e^{-k_\alpha|n-m|^\alpha}.
            \]
    \end{proof}
    
    Due to the hierarchy in \eqref{mixing_hierarchy}, we obtain the following immediate consequences.
    
    \begin{cor}
    \label{cor:poly_from_psi_phi}
    Under the hypotheses of \Cref{thm:decay_cor} except that we do \emph{not} assume $\rho_n\to0$, the same conclusion holds if either $\psi_n\to0$ or $\varphi_n\to0$.
    \end{cor}

        \begin{proof}
            By \eqref{mixing_hierarchy}, for any two sub-$\sigma$-algebras,
            \[
                \rho(\mcA,\mcB)\ \le\ \psi(\mcA,\mcB),
                \qquad
                \rho(\mcA,\mcB)\ \le\ 2\varphi(\mcA,\mcB)^{1/2}.
            \]
            Hence $\psi_n\to0$ implies $\rho_n\to0$, and likewise $\varphi_n\to0$ implies $\rho_n\to0$. 
            The proof of \Cref{thm:decay_cor} depends only on the fact that $\rho_n\to0$ and therefore, the same bound follows verbatim under either $\psi$–mixing or $\varphi$–mixing.
        \end{proof}

    \begin{cor}
    \label{cor:fast_from_psi_phi}
    Assume the hypotheses of \Cref{thm:fast_decay} except that instead of assuming a stretched–exponential bound on $\rho_n$, we assume a stretched–exponential bound on either $\psi_n$ or $\varphi_n$:
    \[
        \psi_n\ \le\ C_\psi\,e^{-c_\psi n^\gamma}
        \quad\text{or}\quad
        \varphi_n\ \le\ C_\varphi\,e^{-c_\varphi n^\gamma}
        \qquad(n\in\mbN),
    \]
    for some $C_\bullet,c_\bullet>0$ and $\gamma\in(0,1]$. Then the conclusion of \Cref{thm:fast_decay} holds as stated. In particular, for every $\alpha\in(0,\gamma)$ there exist $K_\alpha,k_\alpha>0$ such that for all local observables $\mbO_n,\mbO_m$ with $|n-m|\ge2$,
    \[
          \pr\!\left\{\omega:\ f^\omega(n,m)\ \le\ K_\alpha\,\norm{\mbO_n}_\infty\norm{\mbO_m}_\infty\,e^{-k_\alpha |n-m|^\alpha}\right\}
          \ \ge\ 1-e^{-k_\alpha |n-m|^\alpha}.
    \]
    \end{cor}

        \begin{proof}
            If $\psi_n\le C_\psi e^{-c_\psi n^\gamma}$, then $\rho_n\le\psi_n$ by \eqref{mixing_hierarchy}, hence $\rho_n$ itself obeys the same stretched–exponential bound with the same $(C_\rho,c_\rho,\gamma)=(C_\psi,c_\psi,\gamma)$. Applying \Cref{lemma:main_fast_decay} and then \Cref{thm:fast_decay} yields the claim.
            
            If $\varphi_n\le C_\varphi e^{-c_\varphi n^\gamma}$, then again by \eqref{mixing_hierarchy}
            \[
                \rho_n\ \le\ 2\varphi_n^{1/2}
                \ \le\ 2\sqrt{C_\varphi}\,e^{-(c_\varphi/2)\,n^\gamma}.
            \]
            Thus $\rho_n$ is stretched–exponentially decaying with the same exponent $\gamma$ and modified constants $(C_\rho,c_\rho)=(2\sqrt{C_\varphi},\,c_\varphi/2)$. Applying \Cref{lemma:main_fast_decay} and then \Cref{thm:fast_decay} with these constants gives the stated result.
        \end{proof}


\subsection{Sufficiently Fast \texorpdfstring{$\beta$}{beta}-mixing}
\label{section:beta_mixing}


    Assume throughout this subsection that the random system is \emph{exponentially $\beta$-mixing}, i.e., there exist constants $C_\beta,c_\beta>0$ such that $\beta_n\le C_\beta e^{-c_\beta n}$ for all $n\in\mbN$. 
    By the hierarchy \eqref{mixing_hierarchy}, this entails exponential strong mixing:
    \[
        \alpha_n\ \le\ \tfrac12\,\beta_n\ \le\ \tfrac12 C_\beta\,e^{-c_\beta n}\qquad(n\in\mbN).
    \]

    We shall need the following Bernstein-type concentration inequality (see \cite[Theorem~1]{Merlev_de_2009}).

    \begin{theorem}[Merlev\`ede--Peligrad--Rio]
    \label{thm:MPR}
        Let $(X_n)_{n\ge1}$ be a sequence of centered, real-valued, uniformly bounded random variables, i.e.\ $\mathbb E[X_n]=0$ and $\sup_n\|X_n\|_{L^\infty(\Omega)}\le M<\infty$, which is strongly mixing with $\alpha_n\le e^{-2cn}$ for some $c>0$. Then there exists $C_{\mathrm{MPR}}>0$ such that for all $n\ge4$ and $x\ge0$,
            \begin{equation}
              \Pr\!\left\{\left|\sum_{k=1}^n X_k\right|\ge x\right\}
              \ \le\
              \exp\!\left(\!-\,\frac{C_{\mathrm{MPR}}x^2}{\,nM^2 + Mx(\ln n)(\ln\ln n)}\right).
            \end{equation}
    \end{theorem}

    Here $\norm{\, \cdot \,}_{L^\infty(\Omega)}$ denotes the essential supremum. 
    The theorem above shall play a central role in the proof of \Cref{thm:fast_beta_decay}.

    \fastdecaybeta*

        \begin{proof}
            By \Cref{assumption1,assumption2} and \Cref{lemma:cnum_properties}, there exists $b\in\mbN$, $\lambda\in(0,1)$, and $p_*:=\Pr\{\,\cnum{\Phi^{(b)}}\le\lambda\,\}>0$. 
            Define the block indicators and centered random variables
            \[
                  I_i:=\mathbf 1_{\{\,\cnum{\phi_{(i+1)b-1}\circ\cdots\circ\phi_{ib}}\le\lambda\,\}},
                  \qquad
                  X_i:=I_i-p_*,
                  \qquad \forall i\in \mbZ.
            \]
            Since $(\mcA_n)_{n\in\mbZ}$ is \emph{strictly stationary}, $(I_i)_i$ is stationary with $\mathbb E[I_i]=p_*$, hence $\mathbb E[X_i]=0$. Clearly $\|X_i\|_{L^\infty}\le1$.
            We also take $B_i:=\phi_{(i+1)b-1}\circ\cdots\circ\phi_{ib}$. 

            We wish to apply \Cref{thm:MPR} to the sequence $(X_n)_{n\ge 0}$. 
            Consider the forward and backward $\sigma$-algebras generated by $(X_n)_n$:
            \[
                  \mcG_k := \sigma(X_i :\, i \le k), 
                  \qquad 
                  \mcG^k:= \sigma(X_i:\,i\ge k).
            \]
            From the definition of $X_n$,
            \[
                  \mcG_k \subseteq \sigma(\mcA_n: n\le (k+1)b-1) = \mcF_{(k+1)b-1},
                  \qquad
                  \mcG^{k+n} \subseteq \sigma(\mcA_s : s \ge (k+n)b) = \mcF^{(k+n)b}.
            \]
            Therefore, for the $\alpha$-mixing profile of $X:=(X_n)_n$,
            \[
                  \alpha_n^{X} 
                  \ :=\ \sup_{k\in\mbZ} \alpha\!\big(\mcG_k, \mcG^{k+n}\big)
                  \ \le\ \sup_{k\in\mbZ} \alpha\!\big(\mcF_{(k+1)b-1}, \mcF^{(k+n)b}\big)
                  \ \le\ \alpha_{(n-1)b+1}.
            \]
            So, 
            \[
                \alpha_n^{X} \ \le\ \dfrac12 C_\beta\, e^{-c_\beta((n-1)b+1)}
                  \ \le\ C'\,e^{-c' n},
            \]
            for suitable $C',c'>0$. 
            Since every strong-mixing coefficient is at most $1/4$, the preceding estimate implies, after decreasing the exponential rate if necessary, that there exists $c_X>0$ such that
            \[
                \alpha_n^X\le e^{-2c_Xn}
                \qquad(n\ge1).
            \]
            Hence \Cref{thm:MPR} applies to $(X_j)_{j\ge 0}$ with $M=1$.
       
            \medskip
            Now let 
            \[
                S_k:=\sum_{i=0}^{k-1}X_i=\sum_{i=0}^{k-1}I_i-kp_*.
            \]
            By \Cref{thm:MPR}, for all $k\ge 4$,
            \begin{equation}
            \label{eq:bernstein_blocks_exact}
                  \Pr\!\left\{|S_k|\ge \dfrac{p_*}{2}k\right\}
                  \ \le\
                  \exp\!\left(
                    -\dfrac{C_{\mathrm{MPR}} \, p_*^2}{4}\cdot
                    \dfrac{k}{\,1 + \tfrac{p_*}{2}\,\ln k\,\ln\ln k\,}
                  \right).
            \end{equation}
            Choose
            \[
              K_\star\ :=\ \max\!\left\{\,4,\ \ceil{\exp{\exp{2/p_*}}}\right\}.
            \]
            Then for all $k\ge K_\star$ we have $\ln\ln k>0$, and $\dfrac{p_*}{2}\,\ln k\,\ln\ln k\ge 1$, hence
            \[
                1+\dfrac{p_*}{2}\,\ln k\,\ln\ln k\ \le\ p_*\,\ln k\,\ln\ln k,
            \]
            and \eqref{eq:bernstein_blocks_exact} yields
            \[
                  \Pr\!\left\{|S_k|\ge \dfrac{p_*}{2}k\right\}
                  \ \le\
                  \exp{
                    -\,\frac{C_{\mathrm{MPR}}\, p_*}{4}\cdot\frac{k}{\ln k\,\ln\ln k}
                  }.
            \]
            Since $\{\sum_{i=0}^{k-1}I_i\le \tfrac{p_*}{2}k\}\subseteq\{|S_k|\ge \tfrac{p_*}{2}k\}$, letting $C_2:=C_{\mathrm{MPR}}\, p_*/4>0$ we get, for all $k\ge K_\star$,
            \begin{equation}
            \label{eq:bernstein_blocks_large_k}
              \Pr\!\left\{\sum_{i=0}^{k-1}I_i\le \dfrac{p_*}{2}k\right\}
              \ \le\
              \exp{-\,\frac{C_2\,k}{\ln k\,\ln\ln k}}.
            \end{equation}

            \medskip
            Fix $m<n$ and set $L:=n-m-1$. Then for sufficiently large $L$ there is $k\in\mbN$ so that $kb\ \le\ L-1\ <\ (k+1)b$. 
            Then by submultiplicativity and $0\le\cnum{\,\cdot\,}\le 1$,
            \[
              \cnum{\Phi^{(L)}}\ \le\ \prod_{i=0}^{k-1}\cnum{\phi_{(i+1)b-1}\circ\ldots\circ\phi_{ib}} =\prod_{i=0}^{k-1}\cnum{B_i} .
            \]
            On the event $\sum_{i=0}^{k-1} I_i \ge \frac{p_*}{2}k$ we have
            \[
                  \cnum{\Phi^{(L)}}
                  \ \le\ 
                  \prod_{i=0}^{k-1}\cnum{B_i}
                  \ \le\
                  \lambda^{\sum_{i=0}^{k-1} I_i}
                  \ \le\
                  \lambda^{\frac{p_*}{2}k}.
            \]
            Since $k\ge \frac{L-1}{b}-1$ and $0<\lambda<1$, 
            \[
                  \lambda^{\frac{p_*}{2}k}
                  \ \le\
                  \lambda^{\frac{p_*}{2}\left(\frac{L-1}{b}-1\right)}
                  \ =\
                  \lambda^{-\frac{p_*}{2}}\,
                  \Bigl(\lambda^{\frac{p_*}{2b}}\Bigr)^{L-1}
                  \ =\
                  \lambda^{-\frac{p_*}{2}\left(1+\frac{1}{b}\right)}\,
                  \Bigl(\lambda^{\frac{p_*}{2b}}\Bigr)^{L}.
            \]
            Recalling $L=n-m-1$ and using the previous inequality, we have on the event
            $\sum_{i=0}^{k-1} I_i \ge \frac{p_*}{2}k$ that
            \[
              \cnum{\Phi^{(n-m-1)}}\ \le\ A_\beta\,e^{-\kappa_\beta (n-m-1)}.
            \]
            for some $A_\beta, \kappa_\beta >0$.
            Specifically, we take 
            \[
                A_\beta = \lambda^{-\frac{p_*}{2}\left(1+\frac{1}{b}\right)}\,
                \qquad
                \kappa_\beta = -\ln{\lambda^{\frac{p_*}{2b}}}\,
            \]
            and thus
            \begin{equation}
            \label{eq:cnum_prob_L_to_nm}
                  \Pr\!\left\{\cnum{\Phi^{(n-m-1)}} \le A_\beta\,e^{-\kappa_\beta (n-m-1)}\right\}
                  \ \ge\ 1-\exp\!\left(-\,\frac{C_2\,k}{\ln k\,\ln\ln k}\right).
            \end{equation}

            Since
            \[
                k=\left\lfloor\frac{n-m-2}{b}\right\rfloor,
            \]
            after increasing the large-separation threshold if necessary, we have
            \[
                k\ge\frac{n-m}{2b},
                \qquad
                \ln k\le\ln{(n-m)},
                \qquad
                \ln\ln k\le\ln\ln{(n-m)}.
            \]
            Therefore, with
            \[
                p_\beta:=\frac{C_2}{2b},
            \]
            we obtain, for all sufficiently large $(n-m)$,
            \begin{equation}
            \label{eq:beta_cnum_final_corrected}
                \Pr\!\left\{
                    \cnum{\Phi^{(n-m-1)}}
                    \le
                    A_\beta e^{-\kappa_\beta(n-m-1)}
                \right\}
                \ge
                1-\exp\!\left(
                    -\frac{p_\beta (n-m)}{\ln{(n-m)}\,\ln\ln {(n-m)}}
                \right).
            \end{equation}

            By strict stationarity, the law of
            $\cnum{\phi_{n-1}\circ\cdots\circ\phi_{m+1}}$ coincides with that of
            $\cnum{\Phi^{(n-m-1)}}$.
            Using \Cref{lemma:main_bound}, we have, for $\pr$-almost every $\omega$,
            \[
              f^\omega(n,m)\ \le\ 8\, D\,\norm{\mbO_n}_\infty\,\norm{\mbO_m}_\infty\,
            \cnum{\phi^\omega_{n-1}\circ\cdots\circ\phi^\omega_{m+1}}.
            \]
            Therefore, there exists $N_\beta\in\mbN$ such that for all $|n-m|\ge N_\beta$,
            \begin{equation}
            \label{eq:beta_large_sep}
                  \Pr\!\left\{\,f^\omega(n,m)\ \le\ K'_\beta\,\norm{\mbO_n}_\infty\,\norm{\mbO_m}_\infty\,
                  e^{-\kappa_\beta |n-m|}\right\}
                  \ \ge\
                  1-\exp\!\left(-\,\frac{p_\beta\,|n-m|}{\ln|n-m|\,\ln\ln|n-m|}\right),
            \end{equation}
            where $K'_\beta:=8\,D\,A_\beta e^{\kappa_\beta}$.
            
            For the finitely many separations $2\le |n-m|\le N_\beta-1$, we use the deterministic bound
            \[
              f^\omega(n,m)\ \le\ 8\,D\,\norm{\mbO_n}_\infty\,\norm{\mbO_m}_\infty
              \ \le\ 8\,D\,e^{\,\kappa_\beta(N_\beta-1)}\,\norm{\mbO_n}_\infty\,\norm{\mbO_m}_\infty\,e^{-\kappa_\beta |n-m|}
              \quad\text{a.s.}
            \]
            Hence, with
            \begin{equation}
            \label{eq:Kbeta_final}
                K_\beta\ :=\ \max\!\left\{\,K'_\beta,\ 8\,D\,e^{\,\kappa_\beta(N_\beta-1)}\,\right\},
            \end{equation}
            we obtain the uniform bound
            \begin{equation}
            \label{eq:beta_uniform_all_final}
                \Pr\!\left\{\,f^\omega(n,m)\ \le\ K_\beta\,\norm{\mbO_n}_\infty\,\norm{\mbO_m}_\infty\,
                e^{-\kappa_\beta |n-m|}\right\}
                \ \ge\
                1-\exp\!\left(-\,\frac{p_\beta\,|n-m|}{\ln|n-m|\,\ln\ln|n-m|}\right),
            \end{equation}
            for all $|n-m|\ge 2$.
        \end{proof}


\section*{Acknowledgments}
\addcontentsline{toc}{section}{Acknowledgments}
\noindent 
The authors acknowledge support from Villum Fonden Grant No. 25452 and Grant No. 60842, as well as QMATH Center of Excellence Grant No. 10059. 
LP was also supported for a part of this work by the Danish e-infrastructure Consortium (DeiC) Grant No.~5260-00014B.


\appendix
\section*{Appendices}
\crefalias{section}{appendix}

\makeatletter 
\renewcommand{\thesection}{\Alph{section}}
\makeatother


\section{Quantitative refinement of \texorpdfstring{\Cref{lemma:boundary_states_main}}{the boundary-state lemma}}
\label{section:appen_1}
\makeatletter 
\renewcommand{\thesection}{\Alph{section}}
\makeatother


Before proving the quantitative refinement, we recall the contraction coefficient
\(\cnum{\,\cdot\,}:\mapspace\to[0,1]\) introduced in \cite{Movassagh_2022} for positive cone-preserving superoperators.
This notion is purely deterministic and independent of any randomness.

\paragraph{Projective action.}
    Throughout this appendix, the projective action \(\phi\proj(\,\cdot\,)\) of a positive map \(\phi\in\mapspace\) acts on the positive cone by trace-normalization:
    \begin{equation}
    \label{eq:projective_action}
        \phi\proj X
        :=
        \begin{cases}
            \displaystyle \frac{\phi(X)}{\norm{\phi(X)}_1}, & \text{if }\phi(X)\neq 0,\\[1.2ex]
            0, & \text{if }\phi(X)=0.
        \end{cases}
    \end{equation}
    Thus, whenever \(X\in\states\) and \(\phi(X)\neq0\), we have \(\phi\proj X\in\states\).

\begin{definition}[Contraction coefficient]
\label{dfn:cnum}
    For a positive map \(\phi\in\mapspace\) with \(\ker{\phi}\cap\states=\emptyset\), define
    \begin{equation}
    \label{eq:cnum}
        \cnum{\phi}
        :=
        \sup\Bigl\{
            \mathrm d\!\bigl(\phi\proj A,\phi\proj B\bigr)
            :
            A,B\in\states
        \Bigr\}\in[0,1],
    \end{equation}
    where \(\mathrm d\) is the metric on \(\states\) given by
    \begin{equation}
    \label{eq:d}
        \mathrm d(A,B)
        :=
        \frac{1-m(A,B)m(B,A)}{1+m(A,B)m(B,A)},
        \qquad
        m(A,B):=\sup\{\lambda\ge0:\lambda B\le A\}.
    \end{equation}
\end{definition}

\begin{remark}
\label{remark:c_and_d}
    We shall use, without reproving, the properties of \(\mathrm d\) and \(\cnum{\,\cdot\,}\) established in \cite{Movassagh_2022}.
    For convenience, we recall below the properties of \(\mathrm d\) and \(\cnum{\,\cdot\,}\) that will be used later.
\end{remark}

\begin{prop}[{\cite[Lemma~3.9]{Movassagh_2022}}]
\label{prop:d}
    Let \(\rho,\delta\in\states\). Then:
    \begin{enumerate}
        \item \(\dfrac12\norm{\rho-\delta}_1\le \mathrm d(\rho,\delta)\).
        \item \(\sup_{\rho,\delta\in\states}\mathrm d(\rho,\delta)=1\).
        \item If \(\rho\in\states^{\mathrm o}\) and \(\delta\in\states\), then \(\mathrm d(\rho,\delta)=1\) iff \(\delta\in\partial\states\).
        \item On \(\states^{\mathrm o}\), the topology induced by \(\norm{\,\cdot\,}_1\) coincides with that induced by \(\mathrm d\).
    \end{enumerate}
\end{prop}

\begin{lemma}
\label{lemma:cnum_properties}
    Suppose \(\phi\in\mapspace\) satisfies \(\ker{\phi}\cap\states=\emptyset\). Then:
    \begin{enumerate}
        \item For all \(\rho,\delta\in\states\),
        \[
            \mathrm d\!\bigl(\phi\proj\rho,\phi\proj\delta\bigr)
            \le
            \cnum{\phi}\,\mathrm d(\rho,\delta).
        \]
        \item \(\cnum{\phi}\le1\), and if \(\phi\) is strictly positive, then \(\cnum{\phi}<1\).
        \item If there exist \(\rho,\delta\in\states\) with \(\phi\proj\rho\in\states^{\mathrm o}\) and \(\phi\proj\delta\in\partial\states\), then \(\cnum{\phi}=1\).
        \item If \(\psi\in\mapspace\) also satisfies \(\ker{\psi}\cap\states=\emptyset\), then
        \[
            \cnum{\phi\circ\psi}\le \cnum{\phi}\,\cnum{\psi}.
        \]
        \item If additionally \(\ker{\phi\adj}\cap\states=\emptyset\), then
        \[
            \cnum{\phi}=\cnum{\phi\adj}.
        \]
    \end{enumerate}
\end{lemma}

\begin{cor}
    \label{cor:appn_sp_iff_cnum_1}
    Assume \(\ker{\phi}\cap\states=\emptyset\) and \(\ker{\phi\adj}\cap\states=\emptyset\). Then
    \[
        \cnum{\phi}<1
        \quad\Longleftrightarrow\quad
        \phi \text{ is strictly positive}.
    \]
\end{cor}
    \begin{proof}
        If \(\phi\) is strictly positive, then \(\cnum{\phi}<1\) by \Cref{lemma:cnum_properties}(2).
        Conversely, suppose \(\cnum{\phi}<1\), and assume toward a contradiction that \(\phi\) is not strictly positive.
        Then there exists \(\delta\in\states\) such that \(\phi(\delta)\in\partial\states\).
        We first observe that \(\phi\) maps \(\states^{\mathrm o}\) into \(\states^{\mathrm o}\). Indeed, let \(\rho\in\states^{\mathrm o}\). If \(\phi(\rho)\) had a nontrivial kernel, then for some rank-one projection \(P\in\states\) we would have
        \[
            0=\tr{P\phi(\rho)}=\tr{\phi\adj(P)\rho}.
        \]
        Since \(\phi\adj(P)\ge0\) and \(\rho>0\), this forces \(\phi\adj(P)=0\), contradicting \(\ker{\phi\adj}\cap\states=\emptyset\).
        Thus \(\phi\proj\rho\in\states^{\mathrm o}\) for every \(\rho\in\states^{\mathrm o}\).
        Hence there exist \(\rho,\delta\in\states\) such that \(\phi\proj\rho\in\states^{\mathrm o}\) and \(\phi\proj\delta\in\partial\states\). By \Cref{lemma:cnum_properties}(3), this forces \(\cnum{\phi}=1\), a contradiction.
    \end{proof}

\begin{remark}[Measurability convention]
\label{rem:cnum_measurable}
    Whenever \(\omega\mapsto\phi_\omega\) is a measurable positive random map satisfying
    \[
        \ker{\phi_\omega}\cap\states=\emptyset
        \qquad\text{almost surely},
    \]
    we regard \(\omega\mapsto\cnum{\phi_\omega}\) as a measurable random variable after redefining it arbitrarily off the full-measure kernel event. This convention is harmless in almost all arguments below.
    See \Cite{Pathirana_2025} for measurability of \(\cnum{\,\cdot\,}\). 
\end{remark}

\begin{lemma}[Deterministic rank--one approximation]
\label{lemma:deterministic_rank_one_reference}
    Let \(T\in\mapspace\) be a positive map satisfying \(\ker{T}\cap\states=\emptyset\). 
    Set
    \[
        \rho_*:=\frac{\mbI_D}{D},
        \qquad
        s_T:=\tr{T\adj(\mbI_D)}.
    \]
    Then \(s_T>0\).
    Define
    \[
        \widehat T:=\frac{T}{s_T},
        \qquad
        r_T:=T\proj\rho_*,
        \qquad
        \ell_T:=T\adj\proj\rho_* =
        \frac{T\adj(\mbI_D)}{\tr{T\adj(\mbI_D)}}.
    \]
    Define the rank--one super operator, \(\Xi_T(X) := \tr{\ell_T X}\,r_T \), for all  \(X\in\mbM_D\).
    Then
    \[
        \norm{\widehat T-\Xi_T}_{1\to1}
        \le
        4\,\cnum{T}.
    \]
\end{lemma}

\begin{proof}
    Since \(T\) is positive and \(\ker{T}\cap\states=\emptyset\), one has \(T(\rho)\ge0\) and \(T(\rho)\neq0\) for every \(\rho\in\states\).
    Hence
    \[
        \tr{T(\rho)}>0
        \qquad
        \text{for every }\rho\in\states.
    \]
    In particular,
    \[
        s_T
        =
        \tr{T\adj(\mbI_D)}
        =
        D\,\tr{T(\rho_*)}
        >
        0.
    \]
    Let \(\rho\in\states\).
    Since
    \[
        \tr{T(\rho)}
        =
        \tr{T\adj(\mbI_D)\rho},
    \]
    we have
    \[
        \frac{\tr{T(\rho)}}{s_T}
        =
        \tr{\ell_T\rho}.
    \]
    Therefore
    \[
        \widehat T(\rho)
        =
        \frac{T(\rho)}{s_T}
        =
        \tr{\ell_T\rho}\,T\proj\rho.
    \]
    It follows that
    \[
        \widehat T(\rho)-\Xi_T(\rho)
        =
        \tr{\ell_T\rho}\,
        \left(
            T\proj\rho-T\proj\rho_*
        \right).
    \]
    Since \(\ell_T\in\states\), we have
    \[
        0\le\tr{\ell_T\rho}\le1.
    \]
    Consequently,
    \[
    \begin{aligned}
        \norm{\widehat T(\rho)-\Xi_T(\rho)}_1
        &\le
        \norm{T\proj\rho-T\proj\rho_*}_1
        \\
        &\le
        2\,\mathrm d\left(T\proj\rho,T\proj\rho_*\right)
        \\
        &\le
        2\,\cnum{T}\,\mathrm d\left(\rho,\rho_*\right)
        \\
        &\le
        2\,\cnum{T}.
    \end{aligned}
    \]
    Here we used \Cref{prop:d}(1), \Cref{lemma:cnum_properties}(1), and \Cref{prop:d}(2).
    Now let \(X\in\mbM_D\).
    Write
    \[
        X=A+iB,
        \qquad
        A=A\adj,
        \qquad
        B=B\adj.
    \]
    Decompose
    \[
        A=A_+-A_-,
        \qquad
        B=B_+-B_-,
    \]
    with \(A_\pm,B_\pm\ge0\).
    After omitting zero terms and normalizing the nonzero positive parts, one may write
    \[
        X=\sum_{j=1}^4\lambda_j\rho_j,
        \qquad
        \rho_j\in\states,
        \qquad
        \sum_{j=1}^4|\lambda_j|\le2\norm{X}_1.
    \]
    By linearity,
    \[
    \begin{aligned}
        \norm{\left(\widehat T-\Xi_T\right)(X)}_1
        &\le
        \sum_{j=1}^4|\lambda_j|\,
        \norm{\left(\widehat T-\Xi_T\right)(\rho_j)}_1
        \\
        &\le
        2\,\cnum{T}\sum_{j=1}^4|\lambda_j|
        \\
        &\le
        4\,\cnum{T}\,\norm{X}_1.
    \end{aligned}
    \]
    Taking the supremum over \(\norm{X}_1\le1\) gives the claim.
\end{proof}

We now prove the quantitative refinement of \Cref{lemma:boundary_states_main}.
The qualitative lemma stated in the main text follows immediately from the result below.

\begin{lemma}[Quantitative refinement of \texorpdfstring{\Cref{lemma:boundary_states_main}}{the boundary-state lemma}]
\label{lemma:boundary_states_quantitative}
    Let \((\mcA_n)_{n\in\mbZ}\subset\ten\) be the stationary realization from \Cref{sec:Main_Results}, and let \((\phi_n^\omega)_{n\in\mbZ}\) be the associated transfer maps on \(\mbM_D\).
    Assume \Cref{assumption1,assumption2}.
    Then there exists a \(\theta\)-invariant set \(\Omega_0\subseteq\Omega\) with \(\pr(\Omega_0)=1\) such that, for every \(\omega\in\Omega_0\), there exist two families of states
    \[
        \{Z_n(\omega)\}_{n\in\mbZ}\subset\states^{\mathrm o},
        \qquad
        \{Z_n'(\omega)\}_{n\in\mbZ}\subset\states^{\mathrm o},
    \]
    with the following properties.

    \begin{enumerate}
        \item[\emph{(A)}] \textbf{Projective limits and cocycle relations.}
        For every fixed \(k\in\mbZ\),
        \begin{equation}
        \label{eq:quant_projective_limits}
            \lim_{N\to\infty}
            \bigl(\phi_{k+N}^\omega\circ\cdots\circ\phi_k^\omega\bigr)\adj\proj\states
            =
            \{Z_k'(\omega)\},
            \qquad
            \lim_{N\to\infty}
            \bigl(\phi_k^\omega\circ\cdots\circ\phi_{-N}^\omega\bigr)\proj\states
            =
            \{Z_k(\omega)\},
        \end{equation}
        in \(\norm{\,\cdot\,}_1\), uniformly over the initial state. Moreover,
        \begin{equation}
        \label{eq:quant_cocycle}
            \phi_k^\omega\proj Z_{k-1}(\omega)=Z_k(\omega),
            \qquad
            (\phi_k^\omega)\adj\proj Z_{k+1}'(\omega)=Z_k'(\omega)
            \qquad\text{for all }k\in\mbZ.
        \end{equation}
        The boundary families may be chosen equivariantly:
        \[
            Z_k(\omega)=Z_0(\theta^k\omega),
            \qquad
            Z_k'(\omega)=Z_0'(\theta^k\omega).
        \]
    
        \item[\emph{(B)}] \textbf{Asymptotic rank--one form.}
        For every \(m\le n\),
        \begin{equation}
        \label{eq:rank-one-approx}
            \norm{
                \frac{\phi_n^\omega\circ\cdots\circ\phi_m^\omega}
                     {\tr{(\phi_n^\omega\circ\cdots\circ\phi_m^\omega)\adj(\mbI_D)}}
                -
                \Xi_{[m,n]}^\omega
            }_{1\to1}
            \le
            8\,\cnum{\phi_n^\omega\circ\cdots\circ\phi_m^\omega},
        \end{equation}
        where
        \[
            \Xi_{[m,n]}^\omega(X)
            :=
            \tr{Z_m'(\omega)\,X}\,Z_n(\omega),
            \qquad X\in\mbM_D.
        \]
    
        \item[\emph{(C)}] \textbf{Two-sided Lyapunov-type contraction exponent.}
        There exists a \(\theta\)-invariant random variable
        \[
            \xi:\Omega\to[-\infty,0)
        \]
        such that, for every \(\omega\in\Omega_0\) and every \(x\in\mbZ\),
        \begin{equation}
        \label{eq:xi_plus_all_x}
            \lim_{k\to\infty}
            \frac1k
            \ln\cnum{\phi_{x+k-1}^\omega\circ\cdots\circ\phi_x^\omega}
            =
            \xi(\omega),
        \end{equation}
        and
        \begin{equation}
        \label{eq:xi_minus_all_x}
            \lim_{k\to\infty}
            \frac1k
            \ln\cnum{\phi_x^\omega\circ\cdots\circ\phi_{x-k+1}^\omega}
            =
            \xi(\omega).
        \end{equation}
        In particular,
        \[
            \cnum{\phi_n^\omega\circ\cdots\circ\phi_m^\omega}\to0
            \quad\text{as }n\to+\infty\text{ with }m\text{ fixed},
        \]
        and
        \[
            \cnum{\phi_n^\omega\circ\cdots\circ\phi_m^\omega}\to0
            \quad\text{as }m\to-\infty\text{ with }n\text{ fixed}.
        \]
    \end{enumerate}
    \end{lemma}
    \begin{proof}
        We divide the proof into four steps.
        
        \medskip
        \noindent
        \textbf{Step 1: almost sure exponential contraction in both spatial directions.}
        
        Let
        \[
            E_{\mathrm{ker}}
            :=
            \Bigl\{
                \omega\in\Omega:
                \ker{\phi_0^\omega}\cap\states=\emptyset
                \text{ and }
                \ker{(\phi_0^\omega)\adj}\cap\states=\emptyset
            \Bigr\}.
        \]
        By \Cref{assumption1}, \(\pr(E_{\mathrm{ker}})=1\). Define
        \[
            \Omega_{\mathrm{ker}}:=\bigcap_{j\in\mbZ}\theta^{-j}(E_{\mathrm{ker}}).
        \]
        Then \(\Omega_{\mathrm{ker}}\) is \(\theta\)-invariant and \(\pr(\Omega_{\mathrm{ker}})=1\), since $\theta$ is $\pr$-preserving.
        For \(\omega\in\Omega_{\mathrm{ker}}\), all maps \(\phi_j^\omega\) and \((\phi_j^\omega)\adj\) satisfy the kernel hypotheses of \Cref{dfn:cnum}, for every \(j\in\mbZ\).
        By induction, the same is true for every finite composition and its adjoint.
        
        For \(n\ge1\), define the forward and backward blocks
        \[
            \Phi_{\omega,+}^{(n)}
            :=
            \phi_{n-1}^\omega\circ\cdots\circ\phi_0^\omega,
            \qquad
            \Phi_{\omega,-}^{(n)}
            :=
            \phi_0^\omega\circ\phi_{-1}^\omega\circ\cdots\circ\phi_{-n+1}^\omega.
        \]
        Set
        \[
            F_n^+(\omega):=\ln\cnum{\Phi_{\omega,+}^{(n)}},
            \qquad
            F_n^-(\omega):=\ln\cnum{\Phi_{\omega,-}^{(n)}}.
        \]
        Both variables take values in \([-\infty,0]\), and \((F_1^\pm)^+=0\in L^1(\pr)\).
        
        For the forward blocks, submultiplicativity gives
        \[
            F_{n+m}^+(\omega)
            \le
            F_n^+(\omega)+F_m^+(\theta^n\omega)
            \qquad(m,n\ge1).
        \]
        For the backward blocks, submultiplicativity gives
        \[
            F_{n+m}^-(\omega)
            \le
            F_n^-(\omega)+F_m^-(\theta^{-n}\omega)
            \qquad(m,n\ge1).
        \]
        Indeed,
        \[
            \Phi_{\omega,-}^{(n+m)}
            =
            \Phi_{\omega,-}^{(n)}
            \circ
            \Phi_{\theta^{-n}\omega,-}^{(m)}.
        \]

        Let
        \[
            \mcJ:=\{A\in\mcF:\theta^{-1}A=A \text{ mod }\pr\}
        \]
        be the invariant \(\sigma\)-algebra.
        Since the positive parts of \(F_n^\pm\) are integrable, Kingman's subadditive ergodic theorem \cite{Kingman_1973} applies in the extended-valued sense. We also note that it is enough that \((F_n^\pm)^+\) is integrable to use a generalized notion of conditional expectation without requiring \(F_n^\pm\in L^1(\pr)\), see \cite[Remark~8.16]{Klenke_2020}. Thus there exist invariant random variables
        \[
            \xi_+,\xi_-:\Omega\to[-\infty,0]
        \]
        such that
        \begin{equation}
        \label{eq:kingman_plus}
            \lim_{n\to\infty}\frac1n F_n^+(\omega)=\xi_+(\omega),
        \end{equation}
        and
        \begin{equation}
        \label{eq:kingman_minus}
            \lim_{n\to\infty}\frac1n F_n^-(\omega)=\xi_-(\omega)
        \end{equation}
        almost surely. Moreover, by the theory of subadditive functions \cite{hille1996functional}, applied pointwise to the conditional expectations, we have
        \begin{equation}
        \label{eq:kingman_plus_inf}
            \xi_+=\inf_{n\ge1}\frac1n\,\mbE[F_n^+\mid\mcJ],
        \end{equation}
        and
        \begin{equation}
        \label{eq:kingman_minus_inf}
            \xi_-=\inf_{n\ge1}\frac1n\,\mbE[F_n^-\mid\mcJ],
        \end{equation}
        almost surely.
        
        We next show that \(\xi_+=\xi_-\) almost surely. For every \(n\ge1\),
        \[
            F_n^-(\omega)=F_n^+(\theta^{-n+1}\omega).
        \]
        Since \(\mcJ\) is invariant under \(\theta\), conditional expectation onto \(\mcJ\) is unchanged by composition with powers of \(\theta\). Hence
        \[
            \mbE[F_n^-\mid\mcJ]
            =
            \mbE[F_n^+\mid\mcJ].
        \]
        Using \eqref{eq:kingman_plus_inf} and \eqref{eq:kingman_minus_inf}, we obtain
        \[
            \xi_-
            =
            \inf_{n\ge1}\frac1n\,\mbE[F_n^-\mid\mcJ]
            =
            \inf_{n\ge1}\frac1n\,\mbE[F_n^+\mid\mcJ]
            =
            \xi_+
            \qquad\text{almost surely}.
        \]
        We denote their common value by
        \[
            \xi:=\xi_+=\xi_-.
        \]
        
        We claim that \(\xi<0\) almost surely.
        Let
        \[
            A:=\{\xi=0\}\in\mcJ.
        \]
        Since \(F_n^+\le0\), on \(A\) the identity \eqref{eq:kingman_plus_inf} gives
        \[
            0=\xi\le\frac1n\mbE[F_n^+\mid\mcJ]\le0.
        \]
        Thus \(\mbE[F_n^+\mid\mcJ]=0\) almost surely on \(A\), and hence
        \[
            \mbE[F_n^+\mathbf1_A]
            =
            \mbE\!\left[\mathbf1_A\mbE(F_n^+\mid\mcJ)\right]
            =
            0.
        \]
        Since \(F_n^+\mathbf1_A\le0\), this implies \(F_n^+=0\) almost surely on \(A\), for every \(n\). Equivalently,
        \[
            \cnum{\Phi_{\omega,+}^{(n)}}=1
            \qquad\text{for every }n\ge1
        \]
        for almost every \(\omega\in A\).
        
        On the other hand, by \Cref{assumption2} and \Cref{prop:ESP}, \(\Phi_{\omega,+}^{(n)}\) is strictly positive for all sufficiently large \(n\), almost surely. On \(\Omega_{\mathrm{ker}}\), \Cref{cor:appn_sp_iff_cnum_1} applies to these finite compositions, and therefore
        \[
            \cnum{\Phi_{\omega,+}^{(n)}}<1
        \]
        for all sufficiently large \(n\), almost surely. Hence \(\pr(A)=0\), and \(\xi<0\) almost surely.
        
        Now choose a full-measure set \(\Omega_1\subseteq\Omega_{\mathrm{ker}}\) on which the two Kingman limits hold and on which \(\xi<0\). Define
        \[
            \Omega_0:=\bigcap_{j\in\mbZ}\theta^{-j}\Omega_1.
        \]
        Then \(\Omega_0\) is \(\theta\)-invariant and \(\pr(\Omega_0)=1\).
        For \(\omega\in\Omega_0\) and \(x\in\mbZ\), applying the forward limit to \(\theta^x\omega\) gives
        \[
            \lim_{k\to\infty}
            \frac1k
            \ln\cnum{\phi_{x+k-1}^\omega\circ\cdots\circ\phi_x^\omega}
            =
            \xi(\theta^x\omega)
            =
            \xi(\omega),
        \]
        because \(\xi\) is \(\theta\)-invariant.
        Similarly, applying the backward limit to \(\theta^x\omega\) gives
        \[
            \lim_{k\to\infty}
            \frac1k
            \ln\cnum{\phi_x^\omega\circ\cdots\circ\phi_{x-k+1}^\omega}
            =
            \xi(\theta^x\omega)
            =
            \xi(\omega).
        \]
        This proves \eqref{eq:xi_plus_all_x} and \eqref{eq:xi_minus_all_x}. In particular, the corresponding contraction coefficients tend to zero in both spatial directions.
        
        \medskip
        \noindent
        \textbf{Step 2: construction of the boundary states.}
        
        Fix \(k\in\mbZ\) and \(\omega\in\Omega_0\). Let
        \[
            \rho_*:=\frac{\mbI_D}{D}\in\states.
        \]
        For \(N\) sufficiently large, define
        \[
            Z_k^{(N)}(\omega)
            :=
            \bigl(\phi_k^\omega\circ\cdots\circ\phi_{-N}^\omega\bigr)\proj\rho_*,
            \qquad
            Z_k^{\prime,(N)}(\omega)
            :=
            \bigl(\phi_{k+N}^\omega\circ\cdots\circ\phi_k^\omega\bigr)\adj\proj\rho_*.
        \]
        
        We first treat \(Z_k^{(N)}\). For \(M>N\), both \(Z_k^{(N)}(\omega)\) and \(Z_k^{(M)}(\omega)\) belong to
        \[
            \mcK_N^{(k)}(\omega)
            :=
            \bigl(\phi_k^\omega\circ\cdots\circ\phi_{-N}^\omega\bigr)\proj\states,
        \]
        and \(\mcK_M^{(k)}(\omega)\subseteq\mcK_N^{(k)}(\omega)\).
        Hence
        \[
            \norm{Z_k^{(M)}(\omega)-Z_k^{(N)}(\omega)}_1
            \le
            2\,\operatorname{diam}_{\mathrm d}\!\bigl(\mcK_N^{(k)}(\omega)\bigr)
            \le
            2\,\cnum{\phi_k^\omega\circ\cdots\circ\phi_{-N}^\omega}.
        \]
        The right-hand side tends to zero by Step~1, so \((Z_k^{(N)}(\omega))_N\) is Cauchy in trace norm. Define
        \[
            Z_k(\omega):=\lim_{N\to\infty}Z_k^{(N)}(\omega)\in\states.
        \]
        
        The same argument, applied to
        \[
            \mcL_N^{(k)}(\omega)
            :=
            \bigl(\phi_{k+N}^\omega\circ\cdots\circ\phi_k^\omega\bigr)\adj\proj\states,
        \]
        shows that
        \[
            Z_k'(\omega):=\lim_{N\to\infty}Z_k^{\prime,(N)}(\omega)
        \]
        exists in trace norm and belongs to \(\states\).

        We shall use the following elementary consequence of the nested construction.
        For each fixed \(N\), the set \(\mcK_N^{(k)}(\omega)\) is compact in trace norm, because it is the continuous projective image of the compact set \(\states\).
        Moreover, for every \(M\ge N\), one has
        \[
            Z_k^{(M)}(\omega)\in\mcK_N^{(k)}(\omega).
        \]
        Passing to the limit \(M\to\infty\), and using that \(\mcK_N^{(k)}(\omega)\) is closed, gives
        \[
            Z_k(\omega)\in\mcK_N^{(k)}(\omega).
        \]
        Similarly,
        \[
            Z_k'(\omega)\in\mcL_N^{(k)}(\omega)
        \]
        for every \(N\).
        
        The maps \(\omega\mapsto Z_k(\omega)\) and \(\omega\mapsto Z_k'(\omega)\) are measurable as pointwise limits of measurable maps.
        
        We next prove that \(Z_k(\omega),Z_k'(\omega)\in\states^{\mathrm o}\).
        By Step~1, for \(N\) large enough,
        \[
            \cnum{\phi_k^\omega\circ\cdots\circ\phi_{-N}^\omega}<1.
        \]
        Since \(\omega\in\Omega_{\mathrm{ker}}\), \Cref{cor:appn_sp_iff_cnum_1} implies that \(\phi_k^\omega\circ\cdots\circ\phi_{-N}^\omega\) is strictly positive.
        Hence \(\mcK_N^{(k)}(\omega)\subset\states^{\mathrm o}\) for all sufficiently large \(N\).
        If \(Z_k(\omega)\in\partial\states\), then by \Cref{prop:d}(3),
        \[
            \mathrm d\!\bigl(Z_k^{(N)}(\omega),Z_k(\omega)\bigr)=1
        \]
        for all sufficiently large \(N\). But \(Z_k(\omega)\in\mcK_N^{(k)}(\omega)\), and hence
        \[
            \mathrm d\!\bigl(Z_k^{(N)}(\omega),Z_k(\omega)\bigr)
            \le
            \operatorname{diam}_{\mathrm d}\!\bigl(\mcK_N^{(k)}(\omega)\bigr)
            \le
            \cnum{\phi_k^\omega\circ\cdots\circ\phi_{-N}^\omega}
            \longrightarrow0,
        \]
        a contradiction. Thus \(Z_k(\omega)\in\states^{\mathrm o}\).
        The proof that \(Z_k'(\omega)\in\states^{\mathrm o}\) is identical.
        
        For the uniform projective limit, let \(\rho\in\states\).
        Since \(Z_k(\omega)\in\mcK_N^{(k)}(\omega)\), we have
        \[
            \norm{
                \bigl(\phi_k^\omega\circ\cdots\circ\phi_{-N}^\omega\bigr)\proj\rho
                -
                Z_k(\omega)
            }_1
            \le
            2\,\operatorname{diam}_{\mathrm d}
            \left(
                \mcK_N^{(k)}(\omega)
            \right)
            \le
            2\,\cnum{\phi_k^\omega\circ\cdots\circ\phi_{-N}^\omega}.
        \]
        Taking the supremum over \(\rho\in\states\) yields
        \[
            \sup_{\rho\in\states}
            \norm{
                \bigl(\phi_k^\omega\circ\cdots\circ\phi_{-N}^\omega\bigr)\proj\rho
                -
                Z_k(\omega)
            }_1
            \le
            2\,\cnum{\phi_k^\omega\circ\cdots\circ\phi_{-N}^\omega}.
        \]
        Similarly,
        \[
            \sup_{\rho\in\states}
            \norm{
                \bigl(\phi_{k+N}^\omega\circ\cdots\circ\phi_k^\omega\bigr)\adj\proj\rho
                -
                Z_k'(\omega)
            }_1
            \le
            2\,\cnum{
                \bigl(\phi_{k+N}^\omega\circ\cdots\circ\phi_k^\omega\bigr)\adj
            }.
        \]
        Since \(\omega\in\Omega_{\mathrm{ker}}\), \Cref{lemma:cnum_properties}(5) gives
        \[
            \cnum{
                \bigl(\phi_{k+N}^\omega\circ\cdots\circ\phi_k^\omega\bigr)\adj
            }
            =
            \cnum{\phi_{k+N}^\omega\circ\cdots\circ\phi_k^\omega}.
        \]
        The two right-hand sides tend to zero by Step~1.
        Hence the projective limits in \eqref{eq:quant_projective_limits} hold uniformly over the initial state.      
        
        The equivariance follows from uniqueness of these limits. Indeed,
        \[
            Z_k(\omega)=Z_0(\theta^k\omega),
            \qquad
            Z_k'(\omega)=Z_0'(\theta^k\omega).
        \]
        
        \medskip
        \noindent
        \textbf{Step 3: cocycle relations.}
        
        For \(N\) sufficiently large,
        \[
            Z_k^{(N)}(\omega)
            =
            \phi_k^\omega\proj Z_{k-1}^{(N)}(\omega).
        \]
        Since \(\omega\in\Omega_{\mathrm{ker}}\), the projective action of \(\phi_k^\omega\) is continuous on \(\states\).
        Passing to the limit \(N\to\infty\) gives
        \[
            \phi_k^\omega\proj Z_{k-1}(\omega)=Z_k(\omega).
        \]
        Similarly,
        \[
            Z_k^{\prime,(N+1)}(\omega)
            =
            (\phi_k^\omega)\adj\proj Z_{k+1}^{\prime,(N)}(\omega),
        \]
        and letting \(N\to\infty\) gives
        \[
            (\phi_k^\omega)\adj\proj Z_{k+1}'(\omega)=Z_k'(\omega).
        \]
        This proves \eqref{eq:quant_cocycle}.
        
        \medskip
        \noindent
        \textbf{Step 4: asymptotic rank--one approximation.}
        
        Fix \(m\le n\), and write
        \[
            T:=\phi_n^\omega\circ\cdots\circ\phi_m^\omega.
        \]
        By iterating the cocycle relations from Step~3,
        \[
            T\proj Z_{m-1}(\omega)=Z_n(\omega),
            \qquad
            T\adj\proj Z_{n+1}'(\omega)=Z_m'(\omega).
        \]
        
        Let \(\rho\in\states\). Then
        \begin{align}
        \label{eq:ket-bound}
            \norm{T\proj\rho-Z_n(\omega)}_1
            &=
            \norm{T\proj\rho-T\proj Z_{m-1}(\omega)}_1 \notag\\
            &\le
            2\,\mathrm d\!\bigl(T\proj\rho,T\proj Z_{m-1}(\omega)\bigr) \notag\\
            &\le
            2\,\cnum{T}\,\mathrm d\!\bigl(\rho,Z_{m-1}(\omega)\bigr) \notag\\
            &\le
            2\,\cnum{T}.
        \end{align}
        Here we used \Cref{prop:d}(1), \Cref{lemma:cnum_properties}(1), and \Cref{prop:d}(2).
        
        Since \(\omega\in\Omega_{\mathrm{ker}}\), we also have \(\ker{T\adj}\cap\states=\emptyset\), and hence \(\cnum{T\adj}=\cnum{T}\) by \Cref{lemma:cnum_properties}(5).
        Applying the same argument to \(T\adj\), for all \(\sigma\in\states\),
        \begin{equation}
        \label{eq:bra-bound-general}
            \norm{T\adj\proj\sigma-Z_m'(\omega)}_1
            \le
            2\,\cnum{T}.
        \end{equation}
        Choosing \(\sigma=\rho_*=\mbI_D/D\), we obtain
        \begin{equation}
        \label{eq:bra-bound-I}
            \norm{
                \frac{T\adj(\mbI_D)}{\tr{T\adj(\mbI_D)}}
                -
                Z_m'(\omega)
            }_1
            \le
            2\,\cnum{T}.
        \end{equation}
        
        Consequently, for any \(\rho\in\states\),
        \begin{equation}
        \label{eq:scalar-gap}
            \left|
                \frac{\tr{T(\rho)}}{\tr{T\adj(\mbI_D)}}
                -
                \tr{Z_m'(\omega)\rho}
            \right|
            =
            \bigl|\tr{X\rho}\bigr|
            \le
            \norm{X}_1
            \le
            2\,\cnum{T},
        \end{equation}
        where
        \[
            X:=
            \frac{T\adj(\mbI_D)}{\tr{T\adj(\mbI_D)}}
            -
            Z_m'(\omega),
        \]
        and we used \(|\tr{AB}|\le\norm{A}_1\norm{B}_\infty\) together with \(\norm{\rho}_\infty\le1\).
        
        Now set
        \[
            C(\rho):=\frac{\tr{T(\rho)}}{\tr{T\adj(\mbI_D)}}\,Z_n(\omega).
        \]
        Then
        \begin{align*}
            \norm{
                \frac{T(\rho)}{\tr{T\adj(\mbI_D)}}
                -
                \tr{Z_m'(\omega)\rho}\,Z_n(\omega)
            }_1
            &\le
            \norm{
                \frac{T(\rho)}{\tr{T\adj(\mbI_D)}}
                -
                C(\rho)
            }_1
            +
            \norm{
                C(\rho)
                -
                \tr{Z_m'(\omega)\rho}\,Z_n(\omega)
            }_1                                      \\
            &=
            \frac{\tr{T(\rho)}}{\tr{T\adj(\mbI_D)}}\,
            \norm{T\proj\rho-Z_n(\omega)}_1
            +
            \left|
                \frac{\tr{T(\rho)}}{\tr{T\adj(\mbI_D)}}
                -
                \tr{Z_m'(\omega)\rho}
            \right|
            \norm{Z_n(\omega)}_1.
        \end{align*}
        Since \(T\adj(\mbI_D)\ge0\) and \(\rho\in\states\),
        \[
            0
            \le
            \frac{\tr{T(\rho)}}{\tr{T\adj(\mbI_D)}}
            =
            \frac{\tr{T\adj(\mbI_D)\rho}}{\tr{T\adj(\mbI_D)}}
            \le
            1,
        \]
        and \(\norm{Z_n(\omega)}_1=\tr{Z_n(\omega)}=1\).
        Combining this with \eqref{eq:ket-bound} and \eqref{eq:scalar-gap} gives
        \begin{equation}
        \label{eq:pointwise-4c}
            \norm{
                \frac{T(\rho)}{\tr{T\adj(\mbI_D)}}
                -
                \tr{Z_m'(\omega)\rho}\,Z_n(\omega)
            }_1
            \le
            4\,\cnum{T}.
        \end{equation}
        
        Finally, let \(X\in\mbM_D\). As usual, \(X\) can be written as a linear combination of at most four states,
        \[
            X=\sum_{j=1}^4 \lambda_j\rho_j,
            \qquad
            \sum_{j=1}^4|\lambda_j|\le2\norm{X}_1.
        \]
        Applying \eqref{eq:pointwise-4c} termwise and using linearity yields
        \[
            \norm{
                \frac{T(X)}{\tr{T\adj(\mbI_D)}}
                -
                \tr{Z_m'(\omega)X}\,Z_n(\omega)
            }_1
            \le
            8\,\cnum{T}\,\norm{X}_1.
        \]
        Taking the supremum over \(\norm{X}_1\le1\) proves \eqref{eq:rank-one-approx}. This completes the proof.
        \end{proof}
        
        \begin{remark}
        \label{rem:main_qualitative_from_appendix}
        The qualitative statement of \Cref{lemma:boundary_states_main} follows from \Cref{lemma:boundary_states_quantitative} by taking
        \[
            \varepsilon_{m,n}(\omega)
            :=
            8\,\cnum{\phi_n^\omega\circ\cdots\circ\phi_m^\omega},
        \]
        on a full probability event. 
        The convergence of \(\varepsilon_{m,n}(\omega)\) in the two one-sided limits follows from \Cref{lemma:boundary_states_quantitative}(C).
    \end{remark}
    
    
\section{Examples of Random MPS Satisfying the Standing Assumptions}
\label{section:examples}


Before we present the examples, we need the following lemma. 

 \begin{lemma}
    \label{lem:full-span-strict-positive-appendix}
        Let $\{K_r\}_{r=1}^R\subset\mbM_D(\mbC)$ satisfy $\mathrm{span}\{K_r\}_{r=1}^R=\mbM_D(\mbC)$ and define
        $\Psi(X):=\sum_{r=1}^R K_r X K_r\adj$. Then $\Psi$ is \emph{strictly positive}:
        for every non-zero positive semi-definite $X\in\matrices$, we have $\Psi(X)$ is strictly positive definite.
    \end{lemma}

        \begin{proof}
            Suppose $\Psi(X) \not> 0$ for some $X\ge0$ with $X\ne 0$. 
            Then there exists $v\ne 0$ with
            \[
                0\ =\ v\adj\Psi(X)v\ =\ \sum_{r=1}^R \norm{\,X^{1/2} K_r\adj v\,}_2^2,
            \]
            so $X^{1/2}K_r\adj v=0$ for all $r$. Hence $K_r\adj v\in\ker X = \ker{X^{1/2}}$ for all $r$ and, by taking  complex linear combinations, $Z\adj v\in\ker X$ for every $Z$ in the span of $\{K_r\}$, i.e.\ for all $Z\in\mbM_D(\mbC)$. 
            But $Z\mapsto Z\adj v$ is surjective. 
            Indeed, for any $w\in\mbC^D$, taking \(Z=\frac{vw\adj}{\norm{v}_2^2}\) gives  $Z\adj v=w$.
            Thus $\{Z\adj v: Z\in\mbM_D(\mbC)\}=\mbC^D$, whence $\ker X=\mbC^D$, i.e.\ $X=0$, a contradiction. 
        \end{proof}

        \begin{proof}[Alternative]
            We may use the same method in \cite[prop. 1]{Sanz_2010} to prove this for the inhomogeneous case:

            Recall that there is a one-to-one correspondence  between super operators $\Psi\in\mapspace$ and their Choi matrix:
            \[
            \Psi(\, \cdot\,) = \sum_{r=1}^R K_r(\,\cdot\,)K_r\adj \leftrightarrow \mcJ(\Psi) = \sum_{r=1}^R \vec{K_r}\vec{K_r}\adj
            \]
            Hence $\rank\{\mcJ(\Psi)\} = \dim\{\mathrm{span}\{{\vec{K_r} : 1\le r\le R\}}\}$.
            By the hypothesis $\mathrm{span}\{K_r\}=\mbM_D(\mbC)$ and the linear isomorphism $\mathrm{vec}:\mbM_D(\mbC)\to\mbC^{D^2}$, we obtain $\rank\,\mcJ(\Psi)=D^2$, so $\mcJ(\Psi)$ is full-rank. 
            With the convention
            \[
                \vec{A\rho B^T}=(B\otimes A)\vec{\rho},
                    \]
            we have, for every \(v\in\mbC^D\),
            \[
                \Psi(\ket{v}\bra{v})
                =
                (\bra{\overline v}\otimes I)
                 \mcJ(\Psi)
                (\ket{\overline v}\otimes I).
            \]
            Thus $\Psi(\ket{v}\bra{v})$ is also full-rank, whenever $v\neq 0$.  
            Indeed, for \(u,v\in\mbC^D\setminus\{0\}\), we have that 
            \[
                 u\adj\Psi(\ket{v}\bra{v})u = (\bra{\overline v}\otimes\bra u)   \mcJ(\Psi) (\ket{\overline v}\otimes\ket u) > 0
            \]
            Since \( \mcJ(\Psi)>0\) and \(\ket{\overline v}\otimes\ket u\ne0\), we get that \(\Psi(\ket{v}\bra{v})>0\) for each $v\neq 0$.
            This improves to any non-zero positive semi-definite matrix $X$ by the spectral decomposition $X = \sum_j \lambda_j \ket{v_j}\bra{v_j}$.
        \end{proof}
        
\begin{example}[Homogeneous absolutely continuous MPS]
\label{example:ti}
    Let \((\Omega,\mcF)=(\ten,\mcB)\), and let \(\nu\ll\mathrm{Leb}\) on \(\ten\).
    Draw a single Kraus tuple
    \[
        \mcA=(A_1,\ldots,A_d)\sim \nu,
    \]
    and define the homogeneous transfer map
    \[
        \phi(X):=\sum_{i=1}^d A_i\,X\,A_i\adj,
        \qquad X\in\mbM_D(\mbC).
    \]
    Then, \(\nu\)-almost surely, the resulting homogeneous random MPS satisfies
    \Cref{assumption1,assumption2}.
\end{example}

    \begin{proof}
        We verify \Cref{assumption1} and \Cref{assumption2}.
    
        \smallskip
        \noindent\emph{\Cref{assumption1}.}
        For each \(i\), the singular set
        \[
            \mcN_i
            :=
            \left\{
                \mcA=(A_1,\ldots,A_d)\in\ten:
                \det A_i=0
            \right\}
        \]
        is a proper algebraic variety and hence has Lebesgue measure zero.
        Since \(\nu\ll\mathrm{Leb}\) and \(d<\infty\), we have \(\nu\)-almost surely that every \(A_i\) is invertible.
        If \(X\in\states\) and \(\phi(X)=0\), then each term  \(A_iXA_i\adj\) is positive semidefinite and
        \[
            \sum_{i=1}^d A_iXA_i\adj=0
        \]
        forces \(A_iXA_i\adj=0\) for every \(i\).
        Since \(A_i\) is invertible, this implies \(X=0\), a contradiction.
        Hence
        \[
            \ker\phi\cap\states=\emptyset.
        \]
        Applying the same argument to
        \[
            \phi\adj(Y)=\sum_{i=1}^d A_i\adj Y A_i
        \]
        gives
        \[
            \ker\phi\adj\cap\states=\emptyset.
        \]
        This proves \Cref{assumption1}.
    
        \smallskip
        \noindent\emph{\Cref{assumption2}.}
        For \(L\in\mbN\), write
        \[
            S_L(\mcA)
            :=
            \operatorname{span}
            \left\{
                A_{i_L}\cdots A_{i_1}:
                1\le i_1,\ldots,i_L\le d
            \right\}.
        \]
        Set
        \[
            L_\star
            :=
            2\Big\lceil\log_d D\Big\rceil.
        \]
        By \cite[Corollary~1]{Jia_2024}, Lebesgue-almost every Kraus tuple \(\mcA\) satisfies
        \[
            S_{L_\star}(\mcA)=\mbM_D(\mbC).
        \]
        Since \(\nu\ll\mathrm{Leb}\), the same holds  \(\nu\)-almost surely.
        The Kraus operators of \(\phi^{L_\star}\) are
        \[
            A_{i_{L_\star}}\cdots A_{i_1},
            \qquad
            1\le i_1,\ldots,i_{L_\star}\le d.
        \]
        Since these operators span \(\mbM_D(\mbC)\), \Cref{lem:full-span-strict-positive-appendix} shows that \(\phi^{L_\star}\) is strictly positive \(\nu\)-almost surely.
        Since the same Kraus tuple is used at every site,\(\Phi_\omega^{(L_\star)} = \phi^{L_\star}\).
        Thus \Cref{assumption2} holds with the deterministic choice  \(n_*(\omega)=L_\star\).
    \end{proof}

\begin{example}[Homogeneous Gaussian MPS]
\label{example:gaussian_homogeneous}
    Let \(\nu\) be a nondegenerate multivariate Gaussian law on \(\ten\), i.e.\ a Gaussian law on the
    underlying real vector space of \(\ten\) with positive definite covariance.
    Draw a single Kraus tuple
    \[
        \mcA=(A_1,\ldots,A_d)\sim\nu,
    \]
    and place the same tensor at every site.
    Then, almost surely, the resulting homogeneous Gaussian random MPS satisfies
    \Cref{assumption1,assumption2}.
    In particular, this covers the homogeneous Gaussian models studied in \cite{Lancien_2021}.
\end{example}

\begin{proof}
    Any nondegenerate Gaussian law on a finite-dimensional real vector space is absolutely continuous
    with respect to Lebesgue measure.
    Thus \(\nu\ll\mathrm{Leb}\) on \(\ten\), and the conclusion follows immediately from \Cref{example:ti}.
\end{proof}

\begin{example}[IID Absolutely Continuous MPS]
    \label{example:IID}
        Let $\nu\ll\mathrm{Leb}$ on $\ten$ and set $(\Omega,\mcF,\pr)=(\tens^\mbZ,\mcB^{\otimes\mbZ},\nu^{\otimes\mbZ})$, where 
        \[
            \tens = \ten. 
        \]
        Define the strictly stationary local tensors by $\mcA_n(\bar a):=a_n$ for $\bar a=(a_k)_{k\in\mbZ}$, and let $\phi_n$ be the associated transfer maps. Then $(\mcA_n)_{n\in\mbZ}$ satisfies \Cref{assumption1} and \Cref{assumption2}.
    \end{example}

        \begin{proof}
            \emph{\Cref{assumption1}.}
                Let $\phi_{\mcA}(X):=\sum_{i=1}^d A_i X A_i\adj$ for $\mcA=(A_1,\ldots,A_d)\in\ten$ and set
                \[
                    E\ :=\ \Big\{\,\mcA\in\ten:\ \ker{\phi_{\mcA}}\cap\states=\emptyset=\ker{\phi_{\mcA}}\adj\cap\states\,\Big\}.
                \]
                By \Cref{example:ti} we have that $\nu(E)=1$.
                Now on $(\Omega,\mcF,\pr)$, define the pullback events
                \[
                    E_n\ :=\ \mcA_n^{-1}(E)\ =\ \{\omega\in\Omega:\ \mcA_n(\omega)\in E\}\qquad(n\in\mbZ).
                \]
                Because the coordinates are IID with $\mathrm{Law}(\mcA_n)=\nu$, we have
                \[
                    \Pr(E_n)\ =\ \nu(E)\ =\ 1\qquad\text{for every }n\in\mbZ.
                \]
                Thus, $\pr\!\left(\bigcap_{n\in\mbZ} E_n\right)\ =\ 1.$
                Therefore on the full-probability event $\bigcap_{n\in\mbZ} E_n$ we have, for every $n\in\mbZ$, $\ker{\phi_n}\cap\states=\emptyset=\ker{\phi_n\adj}\cap\states$. This establishes \Cref{assumption1} for the IID process under the product measure $\pr$.
            
            \medskip
            
            \noindent\emph{\Cref{assumption2}:}
                We prove that, under $\nu\ll\mathrm{Leb}$ on $\tens$, there exists a finite length
                \[
                    L_\star\ :=\ 2\Big\lceil \log_d D\Big\rceil
                \]
                such that for \emph{$\nu^{\otimes L_\star}$-almost every} length-$L_\star$ block
                $\pmb{A}=(a_0,\ldots,a_{L_\star-1})\in\tens^{L_\star}$ the composed transfer map
                $\Phi^{(L_\star)}:=\phi_{L_\star-1}\circ\cdots\circ\phi_0$ is \emph{strictly positive}. By stationarity of the IID sequence, this yields \Cref{assumption2}.
                The technique closely follows the proof in \Cite{Jia_2024}.

                \medskip

                \noindent\emph{Step 1:}
                    Let $S=\{1,\ldots,d\}$ and denote by $S^{L_\star}_{\mathrm{words}}$ the words of length $L_\star$ on $S$.
                    By the ``sweeping words'' construction \cite{Klep_2016}, there exist $D^2$ words
                    $w_1,\ldots,w_{D^2}\in S^{L_\star}_{\mathrm{words}}$ and a $d$-tuple
                    $\mcB=(B_1,\ldots,B_d)\in\tens$ such that
                    \[
                        \big\{\,w_\ell(\mcB)\,:=\ B_{w_\ell^{(L_\star)}}\cdots B_{w_\ell^{(1)}}\ \big\}_{\ell=1}^{D^2}
                        \quad\text{is linearly independent in } \mbM_D(\mbC),
                    \]
                    hence spans $\mbM_D(\mbC)$.

                \medskip

                \noindent\emph{Step 2:}
                    For a block $\pmb{A}=(a_0,\ldots,a_{L_\star-1})$, with $a_j=(A_1^{(j)},\ldots,A_d^{(j)})\in\tens$,
                    define the \emph{length-$L_\star$ Kraus products}
                    \[
                        K_w(\pmb{A})\ :=\ A_{w^{(L_\star)}}^{(L_\star-1)}\cdots A_{w^{(1)}}^{(0)}\in \mbM_D(\mbC),
                        \qquad w\in S^{L_\star}_{\mathrm{words}}.
                    \]
                    Form the $D^2\times D^2$ matrix of vectorized products
                    \[
                        M(\pmb{A})\ :=\ \big[\,\mathrm{vec}\,K_{w_1}(\pmb{A})\ \ \cdots\ \ \mathrm{vec}\,K_{w_{D^2}}(\pmb{A})\,\big],
                        \qquad
                        P(\pmb{A})\ :=\ \det M(\pmb{A}).
                    \]
                    Each entry of $M(\pmb{A})$ is a polynomial in the entries of the $A_i^{(j)}$, hence $P$ is a
                    polynomial on $\tens^{L_\star}$. Evaluating at the \emph{homogeneous} block
                    \[
                        \pmb{B}\ :=\ (\mcB,\ldots,\mcB)\in\tens^{L_\star}
                    \]
                    gives $K_{w_\ell}(\pmb{B})=w_\ell(\mcB)$; by the choice of $\{w_\ell\}$ and $\mcB$ we have
                    $P(\pmb{B})\neq 0$. Therefore $P\not\equiv 0$ and its zero set $\{P=0\}$ has Lebesgue measure zero in $\tens^{L_\star}$.
                    
                    Since $\nu\ll\mathrm{Leb}$ on $\tens$, the product law $\nu^{\otimes L_\star}$ is absolutely
                    continuous with respect to Lebesgue measure on $\tens^{L_\star}$. Consequently,
                    \begin{equation}
                    \label{eq:block_span_full}
                          \nu^{\otimes L_\star}\!\left\{\,
                                \pmb{A}\in\tens^{L_\star}:\ 
                                \mathrm{span}\{K_{w_\ell}(\pmb{A})\}_{\ell=1}^{D^2}\ =\ \mbM_D(\mbC)\,
                          \right\}\ =\ 1.
                    \end{equation}

                    \medskip

                \noindent\emph{Step 3:}
                    For every block $\pmb{A}$, the composed map has the Kraus
                    representation
                    \[
                        \Phi^{(L_\star)}(X)
                        =
                        \sum_{w\in S^{L_\star}_{\mathrm{words}}}
                        K_w(\pmb{A})XK_w(\pmb{A})\adj.
                    \]
                    On the event in \eqref{eq:block_span_full}, the selected family $\{K_{w_\ell}(\pmb{A})\}_{\ell=1}^{D^2}$ is a subfamily of this full Kraus family and already spans $\mbM_D(\mbC)$. 
                    Hence the full Kraus family also spans $\mbM_D(\mbC)$. 
                    Applying \Cref{lem:full-span-strict-positive-appendix} to the full Kraus family shows that, for $\nu^{\otimes L_\star}$-a.e.\ $\pmb{A}$, the map $\Phi^{(L_\star)}$ is strictly positive.
                    Now on $(\Omega,\mcF,\pr)=(\tens^{\mbZ},\mcB^{\otimes\mbZ},\nu^{\otimes\mbZ})$, every consecutive block
                    $(\mcA_n,\ldots,\mcA_{n+L_\star-1})$ has law $\nu^{\otimes L_\star}$. Thus, for each fixed $n\in\mbZ$,
                    \[
                        \pr\!\left\{\phi_{n+L_\star-1}\circ\cdots\circ\phi_n\ \text{is strictly positive}\right\}\ =\ 1.
                    \]
                    Taking the countable intersection over $n\in\mbZ$ yields probability one for the event
                    that \emph{all} such length-$L_\star$ translates are strictly positive. This proves \Cref{assumption2}.
        \end{proof}


\section{Further Examples of Stochastically Generated MPS}
\label{app:examples}


The purpose of this section is to illustrate how several classical stochastic mechanisms can be used to modulate local tensor laws and thereby generate strictly stationary random MPS ensembles. 
The constructions below are valid for arbitrary branch laws. 
If every branch law almost surely produces transfer maps satisfying the two
kernel conditions in \Cref{assumption1}, and if the modulator reaches in finite nonnegative time a branch that almost surely produces strictly positive transfer maps, then the resulting transfer-map process also satisfies \Cref{assumption1,assumption2}; see \Cref{rem:modulated_examples_assumptions} at the end of this section.


\subsection{Markov–Modulated Random MPS}
\label{app:markov_modulated}


    Fix $m\ge 2$ and $S:=\{1,\dots,m\}$. Let $T\in[0,1]^{m\times m}$ be an irreducible, aperiodic transition matrix with stationary distribution $\pi=(\pi_1, \ldots, \pi_m)$ (so that $\pi T=\pi$). 
    Let $X=(X_n)_{n\in\mbZ}$ be the $S$–valued Markov chain with transition matrix $T$ started in stationarity, so $(X_n)$ is strictly stationary. 
    For each branch $i\in S$ let $\lambda_{B_i}\in\mcP(\tens)$ be a one–site law on rank–three local Kraus tensors, and let $(B^{(i)}_n)_{n\in\mbZ}$ be i.i.d.\ with $\Law(B^{(i)}_0)=\lambda_{B_i}$. Assume the $m$ arrays $\{(B^{(i)}_n)_{n\in\mbZ}:1\le i\le m\}$ and the modulator $X$ are mutually independent. Set
    \[
        (\Omega,\mcF,\pr)
        \ :=\ 
        \left(
            S^{\mbZ}\times\prod_{i=1}^m\tens^{\mbZ},\ \mcB(S)^{\otimes\mbZ}\otimes\bigotimes_{i=1}^m\mcB(\tens)^{\otimes\mbZ},\ \mbP_X\otimes\bigotimes_{i=1}^m\lambda_{B_i}^{\otimes\mbZ}
        \right),
    \]
    where $\mbP_X$ is the stationary path measure for the Markov chain with transition kernel T started at the stationary law $\pi$.
    A typical point $\omega \in \Omega$ is
        \[
                \omega=\left((x_k)_{k\in\mbZ},\ \left(b^{(1)}_k\right)_{k\in\mbZ},\ldots,\left(b^{(m)}_k\right)_{k\in\mbZ} \right)
                \quad\text{with }x_k\in S,\ b^{(i)}_k\in\tens.
        \]
        
    Let $\theta:\Omega\to\Omega$ be the left shift acting simultaneously on all coordinates, $(\theta\omega)_k=(x_{k+1},\,b^{(1)}_{k+1},\ldots,b^{(m)}_{k+1})$.
    Define the \emph{selected branch tensor} at site $n$ by
        \[
            \mcA_n(\omega)\ := \ b_n^{(x_n)}\ \in\tens,\qquad n\in\mbZ.
        \]
    
    Then $\theta$ is $\pr$-preserving as each factor measure $\mbP_X$ and $\lambda_{B_i}  ^{\otimes\mbZ}$ is shift–invariant. And 
    \[
        \mcA_n(\omega) = (\mcA_0\circ\theta^n)(\omega),
    \]
    whence the random local tensors sequence is strictly stationary, with one-site marginal
    \[
        \Law(\mcA_0)
        =
        \sum_{i=1}^m\pi_i\lambda_{B_i}.
    \]

    We next verify the stochastic mixing property of the selected tensor sequence. 
    Set
    \[
        Y_n
        :=
        \bigl(X_n,B_n^{(1)},\ldots,B_n^{(m)}\bigr).
    \]
    Since $X$ is a stationary, finite-state, irreducible, and aperiodic Markov chain, there exist constants $C,c>0$ such that
    \[
        \rho_n^X\le Ce^{-cn},
        \qquad n\in\mbN;
    \]
    see \cite[Theorem~3.1 ]{bradley2005basic}.
    The independent site noise does not increase maximal stochastic correlation. 
    More precisely, if
    \[
        U\in L^2\bigl(\sigma(Y_k:k\le0)\bigr),
        \qquad
        V\in L^2\bigl(\sigma(Y_k:k\ge n)\bigr)
    \]
    are centered, then, upon setting
    \[
        \widetilde U
        :=
        \mbE\!\left[U\mid\sigma(X_k:k\le0)\right],
        \qquad
        \widetilde V
        :=
        \mbE\!\left[V\mid\sigma(X_k:k\ge n)\right],
    \]
    the independence of the branch arrays gives
    \[
        \mbE[UV]=\mbE[\widetilde U\widetilde V].
    \]
    Indeed, let
    \[
        \mcX:=\sigma(X_k:k\in\mbZ).
    \]
    Conditionally on \(\mcX\), the variables \(U\) and \(V\) depend on disjoint independent branch-noise blocks. 
    Hence
    \[
        \mbE[UV\mid\mcX]
        =
        \mbE[U\mid\mcX]\,
        \mbE[V\mid\mcX].
    \]
    Moreover, independence of the branch arrays from \(X\) gives
    \[
        \mbE[U\mid\mcX]=\widetilde U,
        \qquad
        \mbE[V\mid\mcX]=\widetilde V.
    \]
    Therefore
    \[
        \mbE[UV]=\mbE[\widetilde U\widetilde V].
    \]
    Consequently,
    \[
        |\mbE[UV]|
        \le
        \rho_n^X
        \norm{\widetilde U}_{L^2}
        \norm{\widetilde V}_{L^2}
        \le
        \rho_n^X
        \norm{U}_{L^2}
        \norm{V}_{L^2}.
    \]
    Hence $\rho_n^Y\le\rho_n^X$. Since $X$ is a coordinate of $Y$, the reverse inequality also holds, and therefore
    \[
        \rho_n^Y=\rho_n^X.
    \]
    Finally, $\mcA_n$ is a measurable function of $Y_n$, so monotonicity of  maximal correlation under measurable factors yields
    \[
        \rho_n^{\mcA}
        \le
        \rho_n^Y
        =
        \rho_n^X
        \le
        Ce^{-cn}.
    \]
    Thus the selected tensor process is exponentially $\rho$-mixing.

\subsection{Bernoulli–Modulated Random MPS}
\label{app:bernoulli_modulated}


    Fix $p\in(0,1)$ and two branch laws $\lambda_B,\lambda_C\in\mcP(\tens)$ on rank–three local Kraus tensors. 
    Let $(X_n)_{n\in\mbZ}$ be i.i.d.\ $\mathrm{Bernoulli}(p)$ and independent of two i.i.d.\ arrays $(\mcB_n)_{n\in\mbZ}$ and $(\mcC_n)_{n\in\mbZ}$ with one–site laws $\lambda_B$ and $\lambda_C$, respectively. 
    Work on the product space
    \[
        (\Omega,\mcF,\pr)
        \ :=\ 
        \left(
            \{0,1\}^{\mbZ}\times\tens^{\mbZ}\times\tens^{\mbZ},\ 
        \mcB(\{0,1\})^{\otimes\mbZ}\otimes\mcB(\tens)^{\otimes\mbZ}\otimes\mcB(\tens)^{\otimes\mbZ},\ 
        \kappa^{\otimes\mbZ}\otimes\lambda_B^{\otimes\mbZ}\otimes\lambda_C^{\otimes\mbZ}
        \right),
    \]
    where $\kappa:=(1-p)\,\delta_0+p\,\delta_1$, and let $\theta$ denote the left shift on all coordinates. 
    For $\omega = \left((x_n)_{n\in\mbZ}, (b_n)_{n\in\mbZ}, (c_n)_{n\in\mbZ}\right)$, define the selected branch tensor at site $n$ by
    \[
        \mcA_n(\omega)\ :=\
        \begin{cases}
        b_n,& x_n=0,\\
        c_n,& x_n=1,
        \end{cases}
        \qquad n\in\mbZ.
    \]
    Then $(\mcA_n)_{n\in\mbZ}$ is i.i.d.\ (hence strictly stationary) under $\pr$, with one–site law
    \[
        \Law(\mcA_0)\ =\ (1-p)\,\lambda_B\ +\ p\,\lambda_C\ \in\ \mcP(\tens).
    \]
    Placing the tuples $\mcA_n(\omega)$ along $\mbZ$ yields the Bernoulli–modulated random MPS with branches $(\lambda_B,\lambda_C)$.

\subsection{Conditionally Bernoulli–Modulated Random MPS }
\label{app:cond_bernoulli_modulated}

    Let $\nu\in\mcP([0,1])$ be a \emph{prior law} for a latent parameter $P$, and let $\lambda_B,\lambda_C\in\mcP(\tens)$ be branch laws. 
    Work on
    \[
    \begin{aligned}
        \Omega &:= [0,1]\times(0,1)^{\mbZ}\times\tens^{\mbZ}\times\tens^{\mbZ},\\
        \mcF   &:= \mcB([0,1])\otimes\mcB((0,1))^{\otimes\mbZ}\otimes\mcB(\tens)^{\otimes\mbZ}\otimes\mcB(\tens)^{\otimes\mbZ},\\
        \pr    &:= \nu\otimes L^{\otimes\mbZ}\otimes\lambda_B^{\otimes\mbZ}\otimes\lambda_C^{\otimes\mbZ}.
    \end{aligned}
    \]
    and write a typical point as $\omega=(p,(u_k)_k,(b_k)_k,(c_k)_k)$, where $(u_k)$ are i.i.d.\ $\mathrm{Unif}(0,1)$, $(b_k)$ are i.i.d.\ with law $\lambda_B$, $(c_k)$ are i.i.d.\ with law $\lambda_C$, and the three blocks are independent, all independent of $P\sim\nu$. 
    Let $\theta$ be the left shift on the $\mbZ$–indexed coordinates, leaving $p$ fixed:
    \[
        \theta\bigl(p,(u_k)_k,(b_k)_k,(c_k)_k\bigr)
        =\bigl(p,(u_{k+1})_k,(b_{k+1})_k,(c_{k+1})_k\bigr).
    \]

    Define the Bernoulli selectors and the selected branch tensors by
    \[
        H_n(\omega):=\mathbf 1_{\{u_n\le p\}},\qquad 
        \mcA_n(\omega):=\mathbf 1_{\{u_n\le p\}}\,c_n+\mathbf 1_{\{u_n>p\}}\,b_n\ \in\tens,\qquad n\in\mbZ.
    \]

    Then $\theta$ is $\pr$-preserving and
    \[
        \mcA_n=\mcA_0\circ\theta^n,
        \qquad n\in\mbZ.
    \]
    Hence $(\mcA_n)_{n\in\mbZ}$ is strictly stationary, with one-site law
    \[
        \Law(\mcA_0)
        =
        (1-\mbE_\nu[P])\,\lambda_B
        +
        \mbE_\nu[P]\,\lambda_C
        \in\mcP(\tens).
    \]
    Provided $\lambda_B\neq\lambda_C$, the sequence $(\mcA_n)_{n\in\mbZ}$ is not i.i.d.\ unless $P$ is $\nu$-a.s.\ constant.
   

\subsection{Renewal–Modulated Random MPS}
\label{app:renewal_modulated}


    Before we define these MPS, we give a brief informal description of the selection process of the local tensors.
    
   \paragraph{Informal description.}
        Let $m(\ell)=\Pr\{G=\ell\}$ be an inter–arrival law on $\mathbb N$ with mean $\mu\in(1,\infty)$, and set 
        $q(\ell):=\ell\,m(\ell)/\mu$ (size–biased). 
        Draw $L\sim q$ (the length of the block containing the origin) and, conditional on $L$, pick 
        $U\in\{0,\dots,L-1\}$ uniformly (the offset of the origin within its block). 
        Independently, draw i.i.d.\ gaps $\{G_k\}_{k\in\mathbb Z\setminus\{0\}}$ with law $m$ to specify the block lengths away from the origin.
        Define block endpoints by
        \[
            S_0:=-U,\qquad S_1:=S_0+L,\qquad 
            S_{n+1}:=S_n+G_n\ (n\ge 1),\quad S_{n-1}:=S_n-G_{n-1}\ (n\le 0).
        \]
        Mark $n$ as a \emph{block start} iff $n\in\{S_j:j\in\mathbb Z\}$. 
        Independently sample branch arrays $(a_n)$ and $(b_n)$ with one–site laws $\lambda_B$ and $\lambda_C$, and set
        \[
            \mcA_n\ :=\ 
                \begin{cases}
                    a_n, & \text{$n$ is a block start},\\[2pt]
                    b_n, & \text{otherwise}.
                \end{cases}
        \]  
        The size--biased choice of $L$, together with the uniform offset $U$, produces the stationary renewal environment; fixing a block start at the origin would not. 
        The law $m$ is the \emph{inter-arrival distribution}, while $q$ is the size--biased law of the gap containing the origin.

        \begin{figure}
        \centering
        \resizebox{\textwidth}{!}{%
            \begin{tikzpicture}[
              x=0.95cm,y=0.95cm,>=Stealth,
              renew/.style   ={line width=0.9pt,dashed},
              brace/.style   ={decorate,decoration={brace,amplitude=4pt}},
              start/.style   ={draw,fill=teal!35,minimum width=4.6mm,minimum height=4.2mm,inner sep=0pt,rounded corners=0.6mm,thick}, 
              cont/.style    ={draw,fill=orange!28,minimum width=4.6mm,minimum height=4.2mm,inner sep=0pt,rounded corners=0.6mm,thick},
              blockOrigin/.style={fill=blue!7  ,draw=blue!35  ,rounded corners=1mm,thick},
              blockRight/.style ={fill=green!9 ,draw=green!35 ,rounded corners=1mm,thick},
              blockLeft/.style  ={fill=purple!8,draw=purple!35,rounded corners=1mm,thick},
              lab/.style     ={font=\footnotesize},
              tiny/.style    ={font=\scriptsize},
              every node/.style={align=center}
            ]
                \def\U{3}      
                \def\L{5}      
                \def\Gmone{2}  
                \def\Gmtwo{3}  
                \def\Gone{4}  
                \def\Gtwo{3}  
        
                \pgfmathtruncatemacro{\Szero}{-1*\U}                 
                \pgfmathtruncatemacro{\Sone}{\Szero + \L}            
                \pgfmathtruncatemacro{\Sminusone}{\Szero - \Gmone}   
                \pgfmathtruncatemacro{\Sminustwo}{\Sminusone - \Gmtwo}
                \pgfmathtruncatemacro{\Stwo}{\Sone + \Gone}         
                \pgfmathtruncatemacro{\Sthree}{\Stwo + \Gtwo}

                \newcommand{\DrawBlockAB}[3]{
                  \begin{pgfonlayer}{background}
                    \ifnum#3=0
                      \path[blockOrigin] (#1+0.05,-0.75) rectangle (#2-0.05,0.75);
                    \else\ifnum#3>0
                      \path[blockRight ] (#1+0.05,-0.75) rectangle (#2-0.05,0.75);
                    \else
                      \path[blockLeft  ] (#1+0.05,-0.75) rectangle (#2-0.05,0.75);
                    \fi\fi
                  \end{pgfonlayer}
                  \draw[renew] (#1,-2.0) -- (#1,2.25);
                  \draw[renew] (#2,-2.0) -- (#2,2.25);
                  \foreach \n in {#1,...,\numexpr#2-1\relax} {
                    \ifnum\n=#1
                      \node[start] at (\n,0) {$a$};   
                    \else
                      \node[cont]  at (\n,0) {$b$};   
                    \fi
                  }
                  \node[start] at (\Sthree,0) {$a$};
                }

                \DrawBlockAB{\Sminustwo}{\Sminusone}{-1}  
                \DrawBlockAB{\Sminusone}{\Szero}{-1}     
                \DrawBlockAB{\Szero}{\Sone}{0}            
                \DrawBlockAB{\Sone}{\Stwo}{+1}            
                \ifnum\Gtwo>0
                  \DrawBlockAB{\Stwo}{\Sthree}{+1}        
                \fi

                \node[lab,above=2pt] at (\Sminustwo, 2.25) {};
                \node[lab,above=2pt] at (\Sminusone, 2.25) {$S_{-1}=S_0-G_{-1}$};
                \node[lab,above=2pt] at (\Szero    , 2.25) {$S_0$};
                \node[lab,above=2pt] at (\Sone     , 2.25) {$S_1=-U+L$};
                \node[lab,above=2pt] at (\Stwo     , 2.25) {};
                \node[lab,above=2pt] at (\Sthree     , 2.25) {$S_3 = S_2+G_2$};
        
                \draw[brace] (\Szero+0.12,-1.35) -- (-0.12,-1.35)
                  node[midway,lab,yshift=-8pt] {$U\in\{0,\dots,L-1\}$};
                
                \draw[brace] (\Szero+0.12,1.55) -- (\Sone-0.12,1.55)
                  node[midway,lab,yshift=7pt] {$L$};
                
                \draw[brace] (\Sminusone+0.12,1.15) -- (\Szero-0.12,1.15)
                  node[midway,lab,yshift=7pt] {$G_{-1}$};
                
                \draw[brace] (\Sone+0.12,1.15) -- (\Stwo-0.12,1.15)
                  node[midway,lab,yshift=7pt] {$G_{1}$};
                
                \draw[brace] (\Stwo+0.12,1.15) -- (\Sthree-0.12,1.15)
                  node[midway,lab,yshift=7pt] {$G_{2}$};
                
                \draw[brace] (\Sminustwo+0.12,1.15) -- (\Sminusone-0.12,1.15)
                  node[midway,lab,yshift=7pt] {$G_{-2}$};

                \pgfmathsetmacro{\xmin}{\Sminustwo}
                \pgfmathsetmacro{\xmax}{\ifnum\Gtwo>0 \Sthree \else \Stwo \fi}
                \pgfmathsetmacro{\xmarg}{1.2}
                \pgfmathsetmacro{\xa}{\xmin - \xmarg}
                \pgfmathsetmacro{\xb}{\xmax + \xmarg}
                \pgfmathtruncatemacro{\TickL}{floor(\xa)}
                \pgfmathtruncatemacro{\TickR}{ceil (\xb)}
        
                \draw[->,thick] (\xa,-1.8) -- (\xb+0.8,-1.8) node[below] {$\mathbb{Z}$};
                \draw[very thick] (0,-1.74) -- (0,-1.86); 
                \node[below=6pt] at (0,-1.85) {$0$};
                \draw[very thick] (-3,-1.74) -- (-3,-1.86); 
                \node[below=6pt] at (-3,-1.85) {$-U$};
                \draw[very thick] (-8,-1.74) -- (-8,-1.86); 
                \node[below=6pt] at (-8,-1.85) {$S_{-2}= S_{-1}-G_{-2}$};
                \draw[very thick] (6,-1.74) -- (6,-1.86); 
                \node[below=6pt] at (6,-1.85) {$S_{2}= S_{1}+G_{1}$};
            \end{tikzpicture}
        } 
        \caption{One realization of the renewal--modulated placement of local tensors}
        \end{figure}

    \paragraph{Formal description.}
        Let $G$ be an $\mbN$–valued inter-arrival random variable with mass $m(\ell):=\Pr\{G=\ell\}$ and finite mean
        \[
            \mu\ :=\ \mbE[G]\ \in(1,\infty).
        \]
        Define the \emph{size–biased} law $q(\ell):=\ell\,m(\ell)/\mu$ on $\mbN$. For each $\ell\in\mbN$, let $K(\ell,\cdot)$ be the \emph{uniform offset kernel} on $\{0,1,\dots,\ell-1\}$:
        \[
            K(\ell,C)\ :=\ \frac{1}{\ell}\,\#\bigl(C\cap\{0,1,\dots,\ell-1\}\bigr),\qquad C\subseteq\mbN_0.
        \]
        
        \paragraph{Environment space and probability.}
        Set
        \[
            \Omega_{\mathrm{env}}\ :=\ \left\{(\bar{g},r):\ \bar{g}=(g_k)_{k\in\mbZ}\in\mbN^{\mbZ},\ r\in\mbN_0,\ 0\le r\le g_0-1\right\},
        \]
        equipped with the subspace $\sigma$--algebra
        \[
            \mathcal A
            :=
            \left\{
                E\cap\Omega_{\mathrm{env}}:
                E\in
                \mcB(\mbN)^{\otimes\mbZ}\otimes\mcB(\mbN_0)
            \right\}.
        \]

        Define a probability measure $\mbP_{\mathrm{env}}$ on $\Omega_{\mathrm{env}}$ by the identity
        \begin{equation}
        \label{eq:env_measure_def}
            \int_{\Omega_{\mathrm{env}}} f(\bar{g},r)\,d\mbP_{\mathrm{env}}(\bar{g},r)
            \ :=\ \frac{1}{\mu}\;
            \mbE_{\bar{g}\sim m^{\otimes\mbZ}}\!\left[\ \sum_{r=0}^{g_0-1} f(\bar{g},r)\ \right]
        \end{equation}
        for all bounded measurable $f$. In particular, conditionally on $\bar{g}\in\mbN^\mbZ$, the fiber $r\in\{0,\dots,g_0-1\}$ is uniform.
        Indeed, it can be shown that for $A\subseteq \mbN_0$ 
        \[
            \mbP_{\mathrm{env}}\left(r \in A \, | \, \bar{g}\right) = \dfrac{\#\left(A \cap \{0,1,\ldots, g_0-1\}\right)}{g_0}.
        \]
        Unconditionally we get
        \[
            \mbP_{\mathrm{env}}(r=0)\ =\ \frac{1}{\mu}.
        \]
    
        \paragraph{Block–shift map.}
        Define $T:\Omega_{\mathrm{env}}\to\Omega_{\mathrm{env}}$ by
        \[
            T(\bar{g},r)\ :=\
            \begin{cases}
                (\bar{g},r+1), & \text{if } r<g_0-1,\\
                (\sigma(\bar{g}),0), & \text{if } r=g_0-1,
            \end{cases}
        \]
        where $(\sigma(\bar{g}))_k:=g_{k+1}$ is the left shift on the bi–infinite gap sequence.
       The map \(T\) is \(\mathcal A\)-measurable and invertible, with \(\mathcal A\)-measurable inverse
        \[
            T^{-1}(\bar g,r)
            =
            \begin{cases}
                (\bar g,r-1), & \text{if } r>0,\\[2pt]
                (\sigma^{-1}(\bar g),g_{-1}-1), & \text{if } r=0.
            \end{cases}
        \]
        The following lemma verifies that \(T\) preserves \(\mbP_{\mathrm{env}}\).
    
        \begin{lemma}
        \label{lem:T_preserving}
            $T$ preserves $\mbP_{\mathrm{env}}$.
        \end{lemma}
    
            \begin{proof}
                Let $f$ be a bounded measurable function. 
                Using \eqref{eq:env_measure_def} and splitting the fiber sum at $r=g_0-1$,
                \begin{align*}
                    \int f\!\circ T\,d \ \mbP_{\mathrm{env}}
                        &=
                            \frac{1}{\mu}\,\mbE_{\bar{g}\sim m^{\otimes\mbZ}}\left[ \sum_{r=0}^{g_0-2} f(\bar{g},r+1)\ +\ f(\sigma(\bar{g}),0)\right]\\
                        &=
                            \frac{1}{\mu}\,\mbE_{\bar{g}\sim m^{\otimes\mbZ}}\left[ \sum_{r=1}^{g_0-1} f(\bar{g},r)\right]\ +\ \frac{1}{\mu}\,\mbE_{\bar{g}\sim m^{\otimes\mbZ}}\left[f(\sigma(\bar{g}),0)\right].
                \end{align*}
            Since $m^{\otimes\mbZ}$ is shift–invariant, $\mbE_{\bar{g}\sim m^{\otimes\mbZ}}[f(\sigma (\bar{g}),0)]=\mbE_{\bar{g}\sim m^{\otimes\mbZ}}[f(\bar{g},0)]$. Hence
            \[
            \int f\!\circ T\,d\mbP_{\mathrm{env}}
            =\frac{1}{\mu}\,\mbE\left[\sum_{r=0}^{g_0-1} f(\bar{g},r)\right]
            =\int f\,d\mbP_{\mathrm{env}}.
            \]
            \end{proof}
            
        \paragraph{Renewal Modulated MPS}
            Let $\lambda_B,\lambda_C\in\mcP(\tens)$ be branch laws on rank–three local tensors. On
            \[
                (\Omega,\mcF,\pr)\ :=\ \bigl(\Omega_{\mathrm{env}}\times\tens^{\mbZ}\times\tens^{\mbZ},\ 
                \mathcal{A}\otimes\mcB(\tens)^{\otimes\mbZ}\otimes\mcB(\tens)^{\otimes\mbZ},\ 
                \mbP_{\mathrm{env}}\otimes\lambda_B^{\otimes\mbZ}\otimes\lambda_C^{\otimes\mbZ}\bigr),
            \]
            let $(a_k)_{k\in\mbZ}$ and $(b_k)_{k\in\mbZ}$ be the coordinate arrays, independent of $(\bar g,r)$. Define
            \[
                \theta\ :=\ T\ \times\ \text{(left shift)}\ \times\ \text{(left shift)}.
            \]
            By Lemma~\ref{lem:T_preserving} and product structure, $\theta$ preserves $\pr$.
            For $n\in\mbZ$, define $(\bar g^{(n)},r_n):=T^n(\bar g,r)$,  and set
            \[
                \mcA_n \ := \ \mathbf 1_{\{r_n=0\}}a_n + \mathbf 1_{\{r_n\neq 0\}}b_n .
            \]
            Thus $\mcA_n = \mcA_0\circ\theta^n$ and we have that the sequence of local tensors is strictly stationary with one-site marginal given by 
            \[
            \frac{1}{\mu}\,\lambda_B\ +\ \Bigl(1-\frac{1}{\mu}\Bigr)\,\lambda_C\ \in\mcP(\tens).
            \]

\begin{remark}
\label{rem:modulated_examples_assumptions}
    The constructions in \Cref{app:examples} can also be chosen to satisfy \Cref{assumption1,assumption2}. For a local tensor $A$, let  $\phi_A$ denote its associated transfer map. Suppose that:

    \begin{enumerate}
        \item
            every branch law gives full measure to tensors $A$ such that
            \[
                \ker{\phi_A}\cap\states=\emptyset
                \qquad\text{and}\qquad
                \ker{\phi_A\adj}\cap\states=\emptyset;
            \]
        \item
            there is at least one distinguished branch law giving full measure to tensors whose associated transfer maps are strictly positive, and the background modulation process selects such a branch at some finite nonnegative time almost surely.
    \end{enumerate}

    Condition~(1) immediately gives \Cref{assumption1}. 
    Moreover, the first kernel condition ensures that every branch map sends nonzero positive semidefinite matrices to nonzero positive semidefinite matrices. 
    The adjoint kernel condition ensures that every branch map sends positive-definite matrices to positive-definite matrices (see \cite[Lemma~3.1]{Movassagh_2022}). 
    More precisely, if
    \[
        \tau
        :=
        \inf\{n\ge0:\text{a distinguished branch is selected at site }n\},
    \]
    then, on the event $\{\tau<\infty\}$,
    \[
        \Phi_\omega^{(\tau+1)}
        =
        \phi_\tau^\omega\circ\cdots\circ\phi_0^\omega
    \]
    is strictly positive. Thus \Cref{assumption2} holds with \( n_*(\omega)=\tau(\omega)+1\).
\end{remark}


\section*{Data Availability Statement}
\noindent
Data sharing not applicable to this article as no datasets were generated or analyzed during the current study.

\section*{Competing Interests}
\noindent
The authors have no competing interests to declare that are relevant to the content of this article. 


\addcontentsline{toc}{section}{Bibliography}
\printbibliography


\end{document}